\documentclass[11pt]{article}
\usepackage[hmargin=1in,vmargin=1in]{geometry}
\usepackage[dvipsnames,usenames,table]{xcolor}
\usepackage{hyperref}
\hypersetup{colorlinks=true, urlcolor=Blue, citecolor=Green, linkcolor=BrickRed, breaklinks, unicode}
\usepackage{hchang}
\usepackage{tikz}
\usepackage{tikz-3dplot}
\usetikzlibrary{calc}

\usepackage{url}
\usepackage{graphicx}
\usepackage[noend]{algpseudocode}
\usepackage{algorithm}
\usepackage{bbm}
\usepackage{float}
\usepackage{framed}
\usepackage{enumerate}
\usepackage{color}
\usepackage{thm-restate}
\usepackage{comment}
\usepackage{subcaption}
\usepackage{tikz}
\usepackage{cleveref}
\usepackage{subfiles}
\usepackage{cancel}

\usepackage{setspace}
\usepackage{mdframed}

\usepackage{enumitem}
\newlist{inlinelist}{enumerate*}{1}
\setlist[inlinelist]{label=(\roman*),font={\color{red!50!black}\bfseries}, itemjoin={{, }}, itemjoin*={{, }}}
\newtheorem{theorem}{Theorem}

\newtheorem{lemma}{Lemma}
\newtheorem{claim}{Claim}

\newtheorem{observation}{Observation}
\newtheorem{remark}{Remark}

\newtheorem{hypothesis}{Hypothesis}
\newtheorem{question}{Question}

\newcommand{\eps}{\varepsilon}
\newcommand{\real}{\mathbb{R}}

\newcommand{\lotte}[1]{\textcolor{orange}{Lotte: #1}}

\newcommand{\gnote}[1]{\note{G: #1}}
\newcommand{\geert}[1]{\textcolor{teal}{Geert: #1}}

\usepackage{xspace}
\newcommand{\kcenter}{$k\text{-}\mathsf{center}$\xspace}
\newcommand{\tencenter}{$10\text{-}\mathsf{center}$\xspace}
\newcommand{\sixcenter}{$6\text{-}\mathsf{center}$\xspace}
\newcommand{\threecenter}{$3\text{-}\mathsf{center}$\xspace}
\newcommand{\twocenter}{$2\text{-}\mathsf{center}$\xspace}
\newcommand{\radius}{\ensuremath{1-\eps+\eps^{1.7}}\xspace}
\newcommand{\da}{\ensuremath{\Delta}\xspace}
\newcommand{\aaa}{\ensuremath{a}\xspace}
\newcommand{\bb}{\ensuremath{b}\xspace}
\newcommand{\cc}{\ensuremath{c}\xspace}
\newcommand{\dd}{\ensuremath{d}\xspace}
\newcommand{\ee}{\ensuremath{e}\xspace}
\newcommand{\ff}{\ensuremath{f}\xspace}
\newcommand{\PSEdge}[2]{\ensuremath{P(#1, #2)}\xspace}

\newcommand{\pset}{{\mathcal P}}
\newcommand{\dset}{{\mathcal D}}

\newcommand{\apoints}{\widetilde{{\mathcal A}}}
\newcommand{\bpoints}{\widetilde{{\mathcal B}}}
\newcommand{\cpoints}{\widetilde{{\mathcal C}}}
\newcommand{\dpoints}{\widetilde{{\mathcal D}}}

\newcommand{\even}{{\sf Even}}
\newcommand{\odd}{{\sf Odd}}

\title{Fine-Grained Complexity of Continuous Euclidean $k$-Center}
\author{Lotte Blank\thanks{University of Bonn, Germany. Supported by the Deutsche Forschungsgemeinschaft (DFG, German Research Foundation) – 459420781 (FOR AlgoForGe)} \and 
Karl Bringmann\thanks{ETH Zurich, Switzerland. 
Part of this work was done while affiliated to Saarland University and Max Planck Institute for Informatics, Saarbrücken, Germany, where this work was part of the project TIPEA that has received funding from the European Research Council (ERC) under the European Unions Horizon 2020 research and innovation programme (grant agreement No.\ 850979).} \and  
Parinya Chalermsook\thanks{The University of Sheffield, UK.} \and 
Karthik C.\ S.\thanks{Rutgers University, USA. Supported by NSF grant CCF-2313372 and NSF CAREER Award CCF-2443697.} \and
Benedikt Kolbe \thanks{University of Bonn, Hausdorff Center for Mathematics, Lamarr Institute for Machine Learning and Artificial Intelligence, Germany.} \and
Hung Le\thanks{University of Massachusetts, Amherst. Supported by NSF grant CCF-2517033 and NSF CAREER Award CCF-2237288.} 
\and
Geert van Wordragen  \thanks{Aalto University, Finland} 
}
\date{}

\begin{document}

\maketitle

\begin{abstract}
 In the (continuous) Euclidean $k$-center problem, given $n$ points in $\mathbb{R}^d$ and an integer $k$, the goal is to find $k$ center points in $\mathbb{R}^d$ that minimize the maximum Euclidean distance from any input point to its closest center. In this paper, we establish conditional lower bounds for this problem in constant dimensions in two settings.
    
    \begin{description}
        \item[Parameterized by $k$:] Assuming the Exponential Time Hypothesis (ETH), we show that  there is no $f(k)n^{o(k^{1-1/d})}$-time algorithm for the Euclidean $k$-center problem. This result shows that the algorithm of Agarwal and Procopiuc~[SODA 1998; Algorithmica 2002] is essentially optimal. Furthermore, our lower bound rules out any $(1+\varepsilon)$-approximation algorithm running in time $(k/\varepsilon)^{o(k^{1-1/d})}n^{O(1)}$, thereby establishing near-optimality of the corresponding approximation scheme by the same authors.
        
        \item[Small $k$:] Assuming the 3-SUM hypothesis, we prove that for any $\varepsilon>0$ there is no $O(n^{2-\varepsilon})$-time algorithm for the Euclidean $2$-center problem in $\mathbb{R}^3$. This settles an open question posed by Agarwal, Ben Avraham, and Sharir~[SoCG 2010; Computational Geometry 2013]. In addition, under the same hypothesis, we prove that for any $\varepsilon > 0$, the Euclidean $6$-center problem in $\mathbb{R}^2$ also admits no $O(n^{2-\varepsilon})$-time algorithm.
    \end{description}

       The technical core of all our proofs is a novel geometric embedding of a system of linear equations. We construct a point set where each variable corresponds to a specific collection of points, and the geometric structure ensures that a small-radius clustering is possible if and only if the system has a valid solution.

    \end{abstract}

\thispagestyle{empty}

\newpage
\thispagestyle{empty}
\tableofcontents
\thispagestyle{empty}
\newpage

\clearpage
\setcounter{page}{1}

\section{Introduction}

\begin{sloppypar}The \kcenter problem is a fundamental task in combinatorial optimization and a cornerstone of \mbox{center-based} clustering, with the Euclidean \kcenter variant being a classic problem in computational geometry. The Euclidean \kcenter problem seeks to identify the optimal locations for a specific number of facilities, denoted by $k$, to serve a given set of $n$ points in a $d$-dimensional Euclidean space. The objective is to minimize the maximum distance from any point to its nearest facility. This minimax criterion ensures no single point is excessively far from a center, making it a suitable model for diverse fields such as facility location \cite{drezner2004facility}, data summarization \cite{leskovec2020mining}, robotics \cite{bullo2009distributed}, and pattern recognition \cite{duda2006pattern,theodoridis2006pattern}, where worst-case service time or response distance is critical. \end{sloppypar}

Formally, given $n$ points in $\mathbb{R}^d$ and an integer $k$ as input, the goal of the Euclidean \kcenter problem\footnote{This formulation is also known in literature as the \emph{Continuous} Euclidean \kcenter problem, as opposed to the Discrete Euclidean \kcenter problem where the centers must be picked from an input set. } is to find a set of $k$ center points in $\mathbb{R}^d$ such that the largest Euclidean distance from any of the input points to its closest center is minimized. Equivalently, it asks for $k$ balls of equal radius that cover all input points, while minimizing that radius. The problem is known to be NP-hard to approximate within a factor better than 1.82~\cite{FG88} even in the Euclidean plane, while a 2-approximation algorithm exists that runs in polynomial time~\cite{HS85,Gon85,HS86,FG88} for arbitrary dimensions.

The growing importance of big data, particularly large-scale Euclidean datasets, in statistics and machine learning \cite{rajaraman2011mining}, combined with the close connection between various data analysis tasks and \kcenter clustering (for e.g.\ \cite{krause2008near}), has motivated extensive research into the parameterized and fine-grained complexity of the Euclidean \kcenter problem.
\vspace{-0.15in}

\paragraph{Parameterized Complexity.}
The parameterized complexity of Euclidean \kcenter, with the dimensionality $d$ and the number of clusters $k$ as parameters, has been studied since the early 1990s. For the planar case ($d=2$), early algorithms achieved running times of $O(n^{2k-1})$~\cite{drezner1984p}, which was later improved to $n^{O(\sqrt{k})}$ by Hwang et al.~\cite{hwang1993b}. This line of work culminated in the algorithm by Agarwal and Procopiuc~\cite{agarwal2002exact}, who presented an algorithm with a running time of $n^{O(k^{1-1/d})}$ for any constant dimension~$d$.

Despite receiving significant attention from the computational geometry and theory of clustering communities \cite{agarwal1998efficient}, this 27-year-old result\footnote{The conference version of \cite{agarwal2002exact} appeared in 1998 \cite{AgarwalP98}.} from \cite{agarwal2002exact} remarkably remains the state-of-the-art for exact solutions. Breaking this long-standing barrier, even for the planar case, is a central open problem in computational geometry.\vspace{-0.05in}

\begin{question}\label{q1}
	For any $d\ge 2$, is there an algorithm for Euclidean \kcenter on $n$ points in $\mathbb{R}^d$ running in time $n^{o(k^{1-1/d})}$?\vspace{-0.05in}
\end{question}

For the \kcenter problem under the $\ell_{\infty}$-metric in $\mathbb{R}^2$ Marx~\cite{marx2005efficient} ruled out the existence of a fixed-parameter tractable (FPT) algorithm, i.e., an algorithm with a running time of $f(k) \cdot \text{poly}(n)$, under standard complexity assumptions like W[1]$\neq$FPT (see Hypothesis~\ref{hyp:w1}). It is conceivable that this lower bound can be extended to the Euclidean metric as well. In fact, it is conceivable that using existing techniques (such as combining \cite{marx2005efficient} and \cite{MS14}) one can rule out algorithms under the Exponential Time Hypothesis (ETH; see Hypothesis~\ref{hyp:eth}) for the Euclidean \kcenter problem on $n$ points in $\mathbb{R}^2$ running in time $n^{o(\sqrt{k})}$. However, it is not immediately clear how to answer Question~\ref{q1} for $d>2$. 
 
This stands in stark contrast to the complexity when parameterizing by dimension $d$, for which Cabello et al.~\cite{cabello2011geometric} established that Euclidean \kcenter is W[1]-hard even for just $k=2$ centers. 
In fact, their work provides a tight $n^{\Omega(d)}$ running time lower bound for Euclidean \twocenter under ETH.

Given the hardness of finding exact solutions, the focus naturally turns to approximation algorithms. 
Improving upon the classic 2-approximation algorithm~\cite{FG88}, Agarwal and Procopiuc~\cite{agarwal2002exact} provided a $(1+\varepsilon)$-approximation with a running time of $O(n \log k) + (k/\varepsilon)^{O(dk^{1-1/d})}$. 
While different constructions involving coresets in subsequent works have improved the dependence on the dimension $d$ from exponential to linear~\cite{badoiu2002approximate,badoiu2003smaller,kumar2003comuting}, \cite{agarwal2002exact} remains the state-of-the-art for any fixed, constant dimension. 
We ask whether this long-standing bound is optimal in its dependence on $k$:\vspace{-0.05in}

\begin{question}\label{q2}
	For any $d\ge 2$, is there a $(1+\varepsilon)$-factor approximation algorithm for Euclidean \kcenter on $n$ points in $\mathbb{R}^d$ running in time $ (k/\varepsilon)^{o(k^{1-1/d})} \cdot n^{O(1)}$?\vspace{-0.15in}
\end{question}

\paragraph{Fine-Grained Complexity.}

Beyond parameterized complexity, the fine-grained complexity of Euclidean \kcenter for small constants $k$ and $d$ has been a central theme in computational geometry. The planar Euclidean \twocenter problem, for instance, saw a flurry of activity in the 1990s \cite{agarwal1994planar,jaromczyk1994efficient,katz1997expander}, culminating in Sharir's celebrated near-linear time algorithm \cite{Sharir97}. This line of work ultimately settled on an optimal $O(n\log n)$ running time \cite{Eppstein97,chan1999more,wang2022planar,cho2024optimal}.

This success in the plane, however, has not been replicated in higher dimensions. Progress on the Euclidean \twocenter problem in $\mathbb{R}^3$ has been substantially slower, with the best-known algorithm running in $O(n^{3}(\log n)^5)$ time \cite{agarwal20102}, a modest improvement over the initial $O(n^{3+\varepsilon})$ bound from fifteen years prior \cite{agarwal1995vertical}. This significant gap between the planar and 3D cases led Agarwal, Ben Avraham, and Sharir~\cite{agarwal20102} to pose the following question on the  problem's underlying hardness:\vspace{-0.05in}


\begin{question}[Posed by Agarwal, Ben Avraham, and Sharir \cite{agarwal20102}] Is the \twocenter problem in
$\mathbb{R}^3$ 3-SUM-hard? This would suggest that a near-quadratic algorithm is (almost) the best possible for
this problem.\label{q3}\vspace{-0.05in}
\end{question}

Their question is motivated by a second, randomized algorithm from the same paper, which achieves an expected running time of $O((n^2
(\log  n)^5)/(1 - r^*/r_0)^3)$,   where  $r^*$ is the common radius of the \twocenter balls and $r_0$ is the radius of the smallest enclosing ball of the input pointset. This hints that a near-quadratic time is achievable, at least in cases where the geometric quantity $(1 - r^*/r_0)$ is not too small.

A similar complexity gap may also exist when increasing the number of centers in the plane. While the planar \twocenter problem is now considered solved, the complexity of the planar \threecenter problem remains far less understood and, as noted by \cite{ChanHY23}, has not been the subject of focused study. The current best-known running time is a decades-old $O(n^5 \log n)$
 bound, stemming from a general algorithm for 
\kcenter discussed by Hwang et al.~\cite[Page 1]{hwang1993b}. The wide gap between the near-linear time for $k=2$ and this high-polynomial time for $k=3$ motivates the following fine-grained complexity question:\vspace{-0.05in}

\begin{question}\label{q4}
	Is there a conditional lower bound of $n^{1+\varepsilon}$ for some constant $\varepsilon>0$ for the Euclidean \threecenter problem on $n$ points in $\mathbb{R}^2$ under a standard fine-grained hypothesis?\vspace{-0.05in}
\end{question}

In this paper, we answer Questions~\ref{q1}, \ref{q2}, and \ref{q3} completely, and also make partial progress on Question \ref{q4}. But before we detail our results in Section~\ref{sec:results}, we highlight in the next subsection some existing technical challenges that we had to address and overcome.

\subsection{Technical Challenges} 
The difficulty in resolving Questions~\ref{q1}, \ref{q2}, \ref{q3}, and \ref{q4} stems from two fundamental properties of the problem setting: the search for a continuous solution and the constraints imposed by the Euclidean metric. We elaborate on these challenges below.\vspace{-0.15in}

\paragraph{Discrete \kcenter vs.\ Continuous \kcenter.} Recall that the version of \kcenter discussed so far in the paper is referred to also as the continuous \kcenter, as opposed to the discrete \kcenter problem, in which the $k$ centers need to be picked from the input\footnote{In the literature, there is a distinction between the case where the $k$ centers need to be picked from the input points versus the case where apart from the input points (called clients), a separate set of points (called facilities) is also given as input and the goal is to pick $k$ facilities to cover all clients. The latter problem is also called the $k$-\textsf{supplier} problem.}. 
The continuous version of a clustering problem is often algorithmically easier than its discrete counterpart, as it allows an algorithm designer the flexibility to place centers anywhere in space rather than being forced to select from the input points. This performance gap is   illustrated by the \twocenter problem in the plane: the continuous version is solvable in optimal $O(n\log n)$ time \cite{cho2024optimal}, whereas the best-known algorithm for discrete \twocenter remains $O(n^{4/3}\text{polylog } n)$ \cite{AgarwalSW98}. Not surprisingly, this flexibility that aids algorithm design is a major impediment to proving conditional lower bounds for the continuous version.

This intuition  is confirmed by literature: the discrete versions of our open questions have recently seen substantial progress on lower bounds where their continuous counterparts have not. 
For the parameterized complexity questions, the W[1]-hardness and an $n^{\Omega(\sqrt{k})}$ ETH-based lower bound for planar discrete \kcenter were established in \cite{feldmann2020parameterized}. Subsequently, Questions~\ref{q1} and \ref{q2} were fully resolved\footnote{Both the exact and approximate algorithm for the continuous Euclidean \kcenter of \cite{agarwal2002exact} also extend with the same running time to the discrete version.} for the discrete case in all constant dimensions by Chitnis and Saurabh~\cite{ChitnisS22}. Similarly for the fine-grained complexity questions, and in particular Question~\ref{q3}, Chan, He, and Yu~\cite{ChanHY23} proved that under the  Hyperclique Hypothesis, the discrete \twocenter problem in $\mathbb{R}^{13}$ cannot be solved in time $n^{2-\varepsilon}$, for any $\varepsilon>0$, and in general that the discrete \kcenter problem in $\mathbb{R}^{7k}$ cannot be solved in time $n^{k-\varepsilon}$.

 At a technical level, the hardness proofs in the literature for all discrete clustering problems typically follow a common template: they reduce from some hard combinatorial problem $\Pi$ by mapping its constraints to a carefully constructed point set. In this geometric embedding, distances between points are designed to represent the consistency between constraints, such that a valid selection of $k$ discrete centers minimizing the clustering objective corresponds to a satisfying assignment for~$\Pi$. This entire framework, however, faces a common roadblock when adapted to the continuous setting. The core challenge is that an optimal continuous solution may place centers in ``spurious'' locations that have no meaningful interpretation as assignments to the original instances of $\Pi$, thus breaking the reduction's soundness argument. Overcoming this hurdle is a well-known and non-trivial task, as demonstrated by the adaptations required for the NP-hardness proofs of continuous \kcenter itself~\cite{FG88}, as well as for continuous $k$-\textsf{means} and $k$-\textsf{median}~\cite{CKL22}, and even geometric Steiner tree~\cite{FGK25}.
\vspace{-0.15in}

\paragraph{The Challenge of the Euclidean Metric.}
For many geometric problems, a lower bound for the Euclidean metric implies a similar lower bound for every other $\ell_p$-metric, but the reverse is not necessarily true. This asymmetry stems from the fact that any $n$ points in the $\ell_2$-metric can be embedded into any $\ell_p$-metric in $O(\log n)$ dimensions with low distortion \cite{johnson1984extensions} (in particular see \cite{Racke_2006}). Consequently, in high dimensions, the Euclidean metric presents itself as the ``last frontier'' for proving lower bounds. This trend is clearly observed in the  NP-hardness results  of Feder and Greene~\cite{FG88} for \kcenter, wherein they established optimal inapproximability factors for the $\ell_1$ and $\ell_{\infty}$ metrics, but closing the gap between approximation and hardness in the Euclidean metric remains a major open problem. This pattern extends to other fundamental geometric problems such as $k$-\textsf{means} and $k$-\textsf{median}~\cite{CKL22}, Steiner tree~\cite{T00,FGKMP25}, and the Traveling Salesman Problem~\cite{T00}.

\noindent This difficulty also extends to the fixed-dimensional setting, albeit for different reasons. The axis-aligned nature of metrics like $\ell_{\infty}$ is highly amenable to the rigid, grid-like gadgets used in hardness reductions; indeed, the W[1]-hardness proof for planar \kcenter by Marx~\cite{marx2005efficient} was for the $\ell_{\infty}$-metric. In contrast, the abundance of geometric structure in Euclidean space, such as the uniqueness of barycenters and the existence of an inner product, provides powerful tools for algorithm design that can break the fragile analysis of such combinatorial encodings.

 \subsection{Our Results}\label{sec:results}
In this paper, we provide conditional lower bounds that address all four open questions raised earlier. Our results share a common core reduction, which we adapt to different settings. For the parameterized complexity results, this core reduction is combined with the machinery in  \cite{MS14} developed under ETH. For the fine-grained complexity results, the core reduction is combined with assumptions such as the 3-SUM hypothesis. We provide a high-level overview of this common core reduction and our techniques in Section~\ref{sec:overview}. 

\paragraph{Parameterized Lower Bounds.}
We begin by addressing Questions~\ref{q1} and \ref{q2} concerning the optimality of the $n^{O(k^{1-1/d})}$ running time and $(k/\varepsilon)^{O(dk^{1-1/d})}$ running time for the exact and approximate Euclidean \kcenter versions, respectively. Our first two theorems show that these long-standing bounds are essentially tight under ETH.

\begin{theorem}[Exact \kcenter in $\mathbb{R}^d$]\label{thm:k-center}
\begin{sloppypar}
Assuming ETH, for any computable function $f$, there is no \mbox{$f(k)\cdot n^{o(k^{1-1/d})}$}-time algorithm for the Euclidean \kcenter problem on $n$ points in $\mathbb{R}^d$.
\end{sloppypar}
\end{theorem} 

In addition, the proof of the above theorem also implies the following inapproximability result.

\begin{theorem}[Approximate \kcenter in $\mathbb{R}^d$]\label{thm:k-center-approx}
Assuming ETH, for any computable function $f$, there is no $f(k)\cdot (k/\varepsilon)^{o(k^{1-1/d})}\cdot n^{O(1)}$ time $(1+\varepsilon)$-factor approximation algorithm for the Euclidean \kcenter problem on $n$ points in $\mathbb{R}^d$.
\end{theorem}

While the above lower bound holds against $(1+\varepsilon)$-factor approximation algorithms whose guarantee holds for all $\varepsilon$ (not necessarily fixed),   
\cite{bandyapadhyay2022parameterized} considered the case where $\varepsilon$ is a fixed positive constant and showed that for some constant $\varepsilon_0>0$ no $(1+\varepsilon_0)$-approximation for the $2$-dimensional Euclidean  $k$-center problem is possible in time 
$2^{o({k^{1/4}})}\cdot n^{O(1)}$, unless ETH fails. We can improve this running time lower bound to \EMPH{near-optimality} and prove that  for some constant $\varepsilon_0>0$ no $(1+\varepsilon_0)$-approximation for the $2$-dimensional Euclidean  $k$-center problem is possible in time 
$2^{o({\sqrt{k}})}\cdot n^{O(1)}$, unless ETH fails (see \Cref{rem:exp} for details).

\paragraph{Fine-Grained Lower Bounds.}

Next, we turn to the fine-grained complexity for small, constant~$k$ and~$d$. We affirmatively answer Question~\ref{q3}, posed by Agarwal, Ben Avraham, and Sharir, by proving that the \twocenter problem in $\mathbb{R}^3$ is indeed 3-SUM-hard.

\begin{theorem}[2-center problem in $\mathbb{R}^3$]\label{thm:2-center}
Assuming the 3-SUM hypothesis, for every constant $\varepsilon > 0$, there is no $O(n^{2-\varepsilon})$ time algorithm for the Euclidean \twocenter problem on $n$ points in $\mathbb{R}^3$. 
\end{theorem}

\noindent Finally, we make progress towards answering Question~\ref{q4}, which asked about a super-linear lower bound for planar \threecenter. While our techniques do not resolve the complexity for $k=3$ directly, we establish the first quadratic barrier for the planar Euclidean \kcenter problem for any small constant $k$, showing that the problem is 3-SUM-hard already for $k=6$. This provides strong evidence that the complexity of the planar Euclidean \kcenter problem jumps from near-linear for $k=2$ to quadratic for small $k>2$.

\begin{theorem}[6-center problem in $\mathbb{R}^2$]\label{thm:6-center}
Assuming the 3-SUM hypothesis, there is no $O(n^{2-\varepsilon})$ time algorithm for the Euclidean \sixcenter problem in $\mathbb{R}^2$ for any constant $\varepsilon > 0$.
\end{theorem}

Henceforth, for the rest of the paper, we refer to the Euclidean \kcenter problem simply as the \kcenter problem.


\subsection{Organization of Paper}

In \Cref{sec:overview}, we give an overview of the technical ideas behind our main results. In \Cref{sec:prelims} we state the hypotheses used in this paper and in \Cref{sec:10center} we present a warm up proof of a conditional lower bound for the 10-center problem. Next, in \Cref{sec: parameterized results}, we give the lower bound for  $k$-center in $\real^{d}$, proving \Cref{thm:k-center} and \Cref{thm:k-center-approx}. Then, \Cref{sec: 2 center} describes a detailed reduction from the Convolution-3SUM to 2-center in $\real^{3}$, providing a proof of \Cref{thm:2-center}. Finally, in \Cref{sec:6CenterPlane}, we show the reduction from the same problem to $6$-center on the plane to prove \Cref{thm:6-center}. 

\section{Proof Overview}\label{sec:overview}

A fundamental building block in all of our constructions is a geometric gadget encoding the solution of the following Gap Convolution-3SUM problem: Given an array $X[-n\ldots n]$ of $2n+1$ integers in $\{-n^2, \ldots, n^2\}$, decide whether:
\begin{itemize}
    \item  \textbf{YES Case:} $\exists~ i, j, k \in \{-\lfloor \frac{n}{100} \rfloor, \dots, \lfloor \frac{n}{100} \rfloor\}$ such that $i+j+k=0$ and $X[i] + X[j] + X[k] = 0$.
    \item \textbf{NO Case:}  $\not\exists~ i, j, k \in \{-n , \dots, n\}$ such that $i+j+k=0$ and $X[i] + X[j] + X[k] = 0$. 
\end{itemize}

By a simple reduction from the Convolution-3SUM problem, one can show that this problem is 3SUM-hard (see \Cref{lm:gap-3sum-hardness}). The construction of this gadget is perhaps best illustrated in our lower bound for the 2-center problem in $\real^3$.



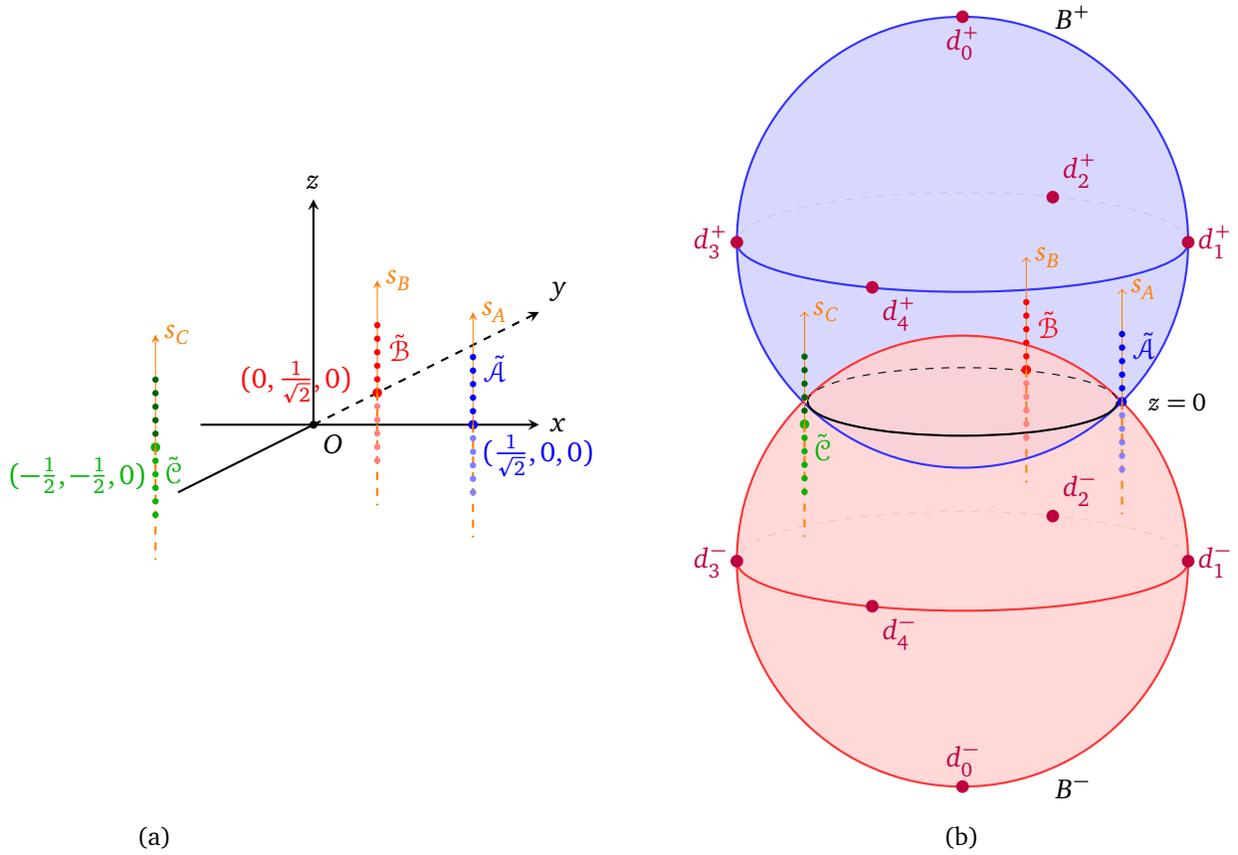
\begin{figure}[!htb]
\centering
\begin{subfigure}[b]{0.25\textwidth}
\centering
\begin{tikzpicture}[scale=3,>=stealth]

\draw[->,thick] (0,0) -- (1,0) node[right] {$x$};
 \draw[thick] (0,0) -- (-0.5,0); 
\draw[->,thick] (0,0) -- (0,1) node[above] {$z$};
\draw[thick] (0,0) -- (-0.6,-0.3); 
\draw[dashed,->,thick] (0,0) -- (1.0,0.5) node[above right] {$y$}; 

\fill (0,0) circle (0.5pt) node[below right] {$O$};

\pgfmathsetmacro{\a}{1/sqrt(2)}

\coordinate (A) at (\a,0);
\coordinate (B) at ({0 + \a*0.4}, {0 + \a*0.2});
\coordinate (C) at ({-0.5 + (-0.5)*0.4}, {0 + (-0.5)*0.2});

\fill[blue] (A) circle (0.6pt) node[below right] {$(\frac{1}{\sqrt2},0,0)$};
\fill[red](B) circle (0.6pt);
\node[red] at ($(B)+(-0.36,0.05)$) {$(0,\frac{1}{\sqrt2},0)$};
\fill[green!70!black] (C) circle (0.6pt) node[below left] {$(-\frac12,-\frac12,0)$};

 \node[blue]at ($(A)+(0.1,0.25)$) {$\tilde{\mathcal{A}}$};
\node[red] at ($(B)+(0.1,0.2)$) {$\tilde{\mathcal{B}}$};
\node[green!70!black] at ($(C)+(0.08,-0.1)$) {$\tilde{\mathcal{C}}$};

\draw[->,orange] (A) -- ++(0,0.5) node[right] {$s_A$};
\draw[dashed,thick,orange] (A) -- ++(0,-0.5);
\foreach \i in {1,2,3,4,5}{\fill[blue] ($(A)+(0,0.06*\i)$) circle (0.4pt);}
\foreach \i in {1,2,3,4,5}{\fill[blue!50] ($(A)+(0,-0.06*\i)$) circle (0.4pt);}

\draw[->,orange] (B) -- ++(0,0.5) node[right] {$s_B$};
\draw[dashed,thick,orange] (B) -- ++(0,-0.5);
\foreach \i in {1,2,3,4,5}{\fill[red] ($(B)+(0,0.06*\i)$) circle (0.4pt);}
\foreach \i in {1,2,3,4,5}{\fill[red!50] ($(B)+(0,-0.06*\i)$) circle (0.4pt);}

\draw[->,orange] (C) -- ++(0,0.5)node[right] {$s_C$};
\draw[dashed,thick,orange] (C) -- ++(0,-0.5);
\foreach \i in {1,2,3,4,5}{\fill[green!70!black] ($(C)+(0,-0.06*\i)$) circle (0.4pt);}
\foreach \i in {1,2,3,4,5}{\fill[green!40!black] ($(C)+(0,0.06*\i)$) circle (0.4pt);}

\end{tikzpicture}
 \vspace{2.8cm}
\caption{}
\end{subfigure}
\hfill
\begin{subfigure}[b]{0.45\textwidth}
\centering
\begin{tikzpicture}[scale=3]
\pgfmathsetmacro{\h}{1/sqrt(2)}
\pgfmathsetmacro{\r}{1} 
\pgfmathsetmacro{\rin}{sqrt(1 - \h*\h)} 
\pgfmathsetmacro{\yr}{0.22} 

\pgfmathsetmacro{\es}{0.97} 

\pgfmathsetmacro{\rinvis}{\es * \rin}
\pgfmathsetmacro{\yrvis}{\es * \yr}

\coordinate (C1) at (0, {\h});
\coordinate (C2) at (0, -{\h}); 

\draw[dashed,gray] (C1) ellipse ({\r} and {\yr*\r});
\draw[dashed,gray] (C2) ellipse ({\r} and {\yr*\r});

\fill[blue!18,opacity=0.85] (C1) circle (\r);
\fill[red!18,opacity=0.85](C2) circle (\r);

\draw[thick,blue!80] (C1) circle (\r);
\draw[thick,red!80](C2) circle (\r);

\draw[thick,blue!80] ($(C1)+(-\r,0)$) arc (180:360:{\r} and {\yr*\r});
\draw[thick,red!80]($(C2)+(-\r,0)$) arc (180:360:{\r} and {\yr*\r});

\draw[dashed,black] ({-\rinvis},0) arc (180:0:{\rinvis} and {\yrvis*\rin});
\draw[thick,black]({-\rinvis},0) arc (180:360:{\rinvis} and {\yrvis*\rin});

\node[right,scale=0.9] at ({\rinvis+0.1},0) {$z=0$};

\pgfmathsetmacro{\a}{1/sqrt(2)}

\coordinate (A) at (\a,0);
\coordinate (B) at ({0 + \a*0.4}, {0 + \a*0.2});
\coordinate (C) at ({-0.5 + (-0.5)*0.4}, {0 + (-0.5)*0.2});

\node[blue]at ($(A)+(0.1,0.25)$) {$\tilde{\mathcal{A}}$};
\node[red] at ($(B)+(0.1,0.2)$) {$\tilde{\mathcal{B}}$};
\node[green!70!black] at ($(C)+(0.08,-0.1)$) {$\tilde{\mathcal{C}}$};
\fill[blue] (A) circle (0.6pt);
\fill[red](B) circle (0.6pt);
\fill[green!70!black] (C) circle (0.6pt);

\draw[->,orange] (A) -- ++(0,0.5) node[right] {$s_A$};
\draw[dashed,thick,orange] (A) -- ++(0,-0.5);
\foreach \i in {1,2,3,4,5}{\fill[blue] ($(A)+(0,0.06*\i)$) circle (0.4pt);}
\foreach \i in {1,2,3,4,5}{\fill[blue!50] ($(A)+(0,-0.06*\i)$) circle (0.4pt);}

\draw[->,orange] (B) -- ++(0,0.5) node[right] {$s_B$};
\draw[dashed,thick,orange] (B) -- ++(0,-0.5);
\foreach \i in {1,2,3,4,5}{\fill[red] ($(B)+(0,0.06*\i)$) circle (0.4pt);}
\foreach \i in {1,2,3,4,5}{\fill[red!50] ($(B)+(0,-0.06*\i)$) circle (0.4pt);}

\draw[->,orange] (C) -- ++(0,0.5)node[right] {$s_C$};
\draw[dashed,thick,orange] (C) -- ++(0,-0.5);

\foreach \i in {1,2,3,4,5}{\fill[green!70!black] ($(C)+(0,-0.06*\i)$) circle (0.4pt);}
\foreach \i in {1,2,3,4,5}{\fill[green!40!black] ($(C)+(0,0.06*\i)$) circle (0.4pt);}


\coordinate (d0) at (0,{1+\a});
\fill[purple] (d0) circle (0.8pt) node[below] {$d^+_0$};

\coordinate (d1) at (1,{\a});
\fill[purple] (d1) circle (0.8pt) node[right] {$d^+_1$};

\coordinate (d2) at ({0+1*0.4},{\a+1*0.2});
\fill[purple] (d2) circle (0.8pt) node[above right] {$d^+_2$};

\coordinate (d3) at (-1,{\a});
\fill[purple] (d3) circle (0.8pt) node[left] {$d^+_3$};

\coordinate (d4) at ({0-0.4},{\a-0.2});
\fill[purple] (d4) circle (0.8pt) node[below right] {$d^+_4$};


\coordinate (dm0) at (0,{-1-\a});
\fill[purple] (dm0) circle (0.8pt) node[above] {$d^-_0$};

\coordinate (dm1) at (1,{-\a});
\fill[purple] (dm1) circle (0.8pt) node[right] {$d^-_1$};

\coordinate (dm2) at ({0+1*0.4},{-\a+1*0.2});
\fill[purple] (dm2) circle (0.8pt) node[above right] {$d^-_2$};

\coordinate (dm3) at (-1,{-\a});
\fill[purple] (dm3) circle (0.8pt) node[left] {$d^-_3$};

\coordinate (dm4) at ({0-0.4},{-\a-0.2});
\fill[purple] (dm4) circle (0.8pt) node[below right] {$d^-_4$};

\node[right=3em] at (d0) {$B^+$};
\node[right=3em] at (dm0) {$B^-$};
\end{tikzpicture}
\caption{}
\end{subfigure}
\caption{(a) Three sets of points $\apoints,\bpoints,\cpoints$ constructed from $X$. The point sets $\apoints,\bpoints, \cpoints$ are oriented upward. (b) Two approximate unit balls associated with a solution of the Gap Convolution 3SUM problem, encoding a 2-clustering for the point sets $\apoints$, $\bpoints$, $\cpoints$, and anchor points $d_i^+$ and $d_i^{-}$ for $i \in \{0,\ldots,4\}$.}
\label{fig:twoballs-gadget-intro}
\end{figure}

\paragraph{$2$-Center in $\real^3$.} Given an instance $X[-n \dots n]$ of the Gap Convolution-3SUM problem where each array entry is bounded by $n^2$, we create three point sets $\apoints[-2n,\ldots, (2n+1)], \bpoints[-2n,\ldots, (2n+1)]$, and $\cpoints[-2n,\ldots, (2n+1)]$ from $X$ in three vertical lines denoted by $s_A, s_B$ and $s_C$ (along the $z$-axis) that pass through three points $(\frac{1}{\sqrt{2}}, 0,0)$, $(0, \frac{1}{\sqrt{2}},0)$ and $(-1/2,-1/2,0)$, respectively. More precisely, $\apoints$, $\bpoints$  and $\cpoints$ are as follows. 

\begin{equation}
    \begin{split}
        \apoints[2i] &=  (\frac{1}{\sqrt{2}}, 0,0) + \left(3i\eps  + X[i]\eps^{1.5}-\eps\right)(0,0,1)\\
        \apoints[2i+1] &=  (\frac{1}{\sqrt{2}}, 0,0) + \left(3i\eps  + X[i]\eps^{1.5}+\eps\right)(0,0,1)\\
        \bpoints[2j] &=  (0, \frac{1}{\sqrt{2}},0) + \left(3j\eps  + X[j]\eps^{1.5}-\eps\right)(0,0,1)\\
        \bpoints[2j+1] &=  (0, \frac{1}{\sqrt{2}},0) + \left(3j\eps  + X[j]\eps^{1.5}+\eps\right)(0,0,1)\\
         \cpoints[2k] &=  (-1/2, -1/2,0) + \left(3k\eps  + X[k]\eps^{1.5}-\sqrt{2}\eps\right)(0,0,1/\sqrt{2})\\
        \cpoints[2k+1] &= (-1/2, -1/2,0) + \left(3k\eps  + X[k]\eps^{1.5}+\sqrt{2}\eps\right)(0,0,1/\sqrt{2})
    \end{split}
\end{equation}

Notice that all three point sets are ordered increasingly in the $z$-coordinate. 
Here $\eps = (10n)^{-100}$ is a very small constant. The spacing between consecutive points in $\apoints$ and $\bpoints$ is almost the same, \emph{up to tiny perturbations due to $X[i]\eps^{1.5}$ and $X[j]\eps^{1.5}$}, while the spacing between consecutive points in $\cpoints$ is $\sqrt{2}$ times smaller. The precise spacing is important in the encoding of a solution to the $2$-center problem. 

Specifically, the $2$-center solution of $\apoints\cup \bpoints\cup \cpoints$ will be two (almost) unit balls $B^+$, called the \EMPH{top ball}, centered near $(0,0,1/\sqrt{2})$ and $B^{-}$, called the \EMPH{bottom ball}, centered near $(0,0,-1/\sqrt{2})$. Our key idea is to guarantee that the coverage of the two balls w.r.t the point sets $\apoints, \bpoints,\cpoints$ corresponds exactly to a solution to the Gap Convolution 3SUM problem. More formally:
\begin{quote}
The top ball will cover points in $\apoints[> 2i], \bpoints[> 2j], \cpoints[> 2k]$ while the bottom ball will cover points in $\apoints[\leq 2i], \bpoints[\leq 2j], \cpoints[\leq 2k]$ for even indices $2i, 2j, 2k \in \{-2n,\ldots, 2n+1\}$ such that $i+j+k = 0$ and $X[i] + X[j] + X[k] = 0$. 
\end{quote}

For this idea to work, at a more technical level, we have to precisely control the potential positions of $B^+$ and $B^-$, and we do so by adding anchor points $D^+ = \{d^+_0\ldots,d^+_4\}$ and $D^- = \{d^-_0\ldots,d^-_4\}$, as shown in \Cref{fig:twoballs-gadget-intro}(b). The intuition is that $D^+$ must be covered by $B^+$ and $D^-$ must be covered by $B^-$, therefore, $B^+$ cannot cover all points in $\apoints,\bpoints$ and $\cpoints$; the same holds for $B^-$. Furthermore, both balls together cover all the points if and only if the system of linear equations has a solution. (We use irrational numbers in the reduction above, but for our purpose, a rational approximation suffices.)

\begin{figure}[!htb]
    \centering
    \begin{subfigure}[b]{0.45\textwidth}
        \centering
        \includegraphics[page=2,height=5.5cm]{figs/Gadget_6Center_Plane1.pdf}
    \end{subfigure}
    \begin{subfigure}[b]{0.50\textwidth}
        \centering
        \includegraphics[page=1,height=6cm]{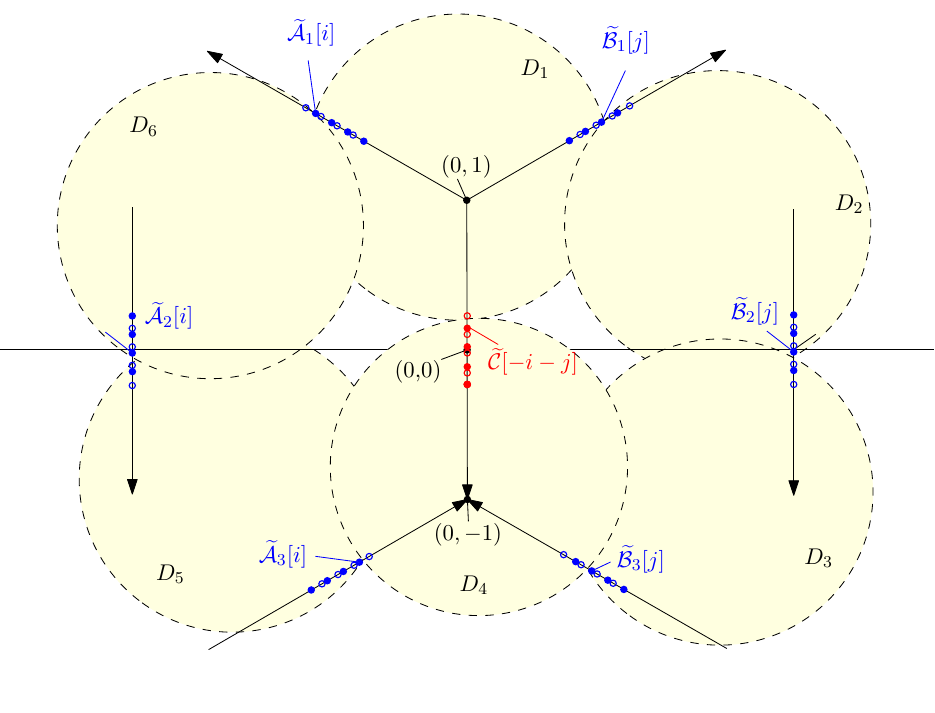}
    \end{subfigure}
    \caption{(left) The point sets $\apoints_1$, $\bpoints_1$, and $\cpoints$ are defined on edges that contain the point $(0,1)$. (right) The upward movement of $D_1$ and $D_4$ induces the movements of other disks. Anchor points are not shown.}
\label{fig:basic-gadget-intro}
\end{figure}

\paragraph{$6$-Center on the plane.} On the plane, we do not have the degree of freedom offered by an extra dimension, and for a good reason: $2$-center on the plane can be solved in $O(n\log n)$ time. This lack of freedom translates to more conceptual obstacles in encoding the linear system of equations in the Gap Convolution-3SUM problem. One clearly has to increase the number of centers, and in this paper, we manage to get by with $6$ centers for the encoding. We start by constructing the point set $\cpoints[-2n\ldots (2n+1)]$ on the edge between $(0, -1)$ and $(0,1)$. Then, we add two point sets $\apoints_1[-2n\ldots (2n+1)]$ and $\bpoints_1[-2n\ldots (2n+1)]$ on two edges that both contain the point $(0,1)$; see \Cref{fig:basic-gadget-intro}(left).

We will guarantee that in the solution to \sixcenter, there will be an (almost) unit disk $D_1$ centered near $(0,1)$, and most importantly, 
as in the 3D case, the coverage of $D_1$ relates to the solution to Gap Convolution-3SUM: 
\begin{quote}
     If $D_1$ covers $\apoints_1[\leq 2i]$, $\bpoints_1[ \leq 2j]$ and $\cpoints[\leq 2k]$, then $i+j + k \leq 0$. In case $i +j = k$, the tiny perturbation terms $X[i]\eps^{1.5}$  guarantee that $X[i] + X[j] + X[k]\leq 0$.
\end{quote}

To get the other direction of the inequality, namely  $i+j + k\geq 0$ and  $
X[i] + X[j] + X[k]\geq 0$ when $i+j +k = 0$, we add a symmetric gadget with two point sets $\apoints_3[-2n\ldots (2n+1)]$ and $\bpoints_3[-2n\ldots (2n+1)]$ on edges containing $(0,-1)$; see \Cref{fig:basic-gadget-intro}(right). Some of these new points will be covered by another almost unit disk $D_4$ centered near $(0,-1)$, which satisfies:
\begin{quote}
     If $D_4$ covers $\apoints_3[> 2i']$, $\bpoints_3[> 2j']$ and $\cpoints[> 2k]$, then $i'+j'+ k' \geq 0$.  In case $i' +j' + k' = 0$, we have $X[i'] + X[j'] +  X[k']\geq 0$.
\end{quote}

At this point, we need to guarantee two things: (i) points that are not covered by $D_1$ and $D_4$ will be covered by other disks, and (ii) $i = i', j = j'$, and $k = k'$; that is, $D_1$ and $D_4$ have to be "synchronized". We do so by adding two other vertical point sets, namely $\apoints_2[-2n\ldots (2n+1)]$ and $\bpoints_2[-2n\ldots (2n+1)]$, which induce four other disks as in \Cref{fig:basic-gadget-intro}(b).  Without going into details, the very high-level (and somewhat inaccurate) idea is that if $D_1$ and $D_4$ deviate from their expected positions (which encodes a solution to Gap Convolution-3SUM), then their movement induces the movements of other disks, and in the end, some points will not be covered. Therefore, we can only cover all the points by $6$ disks if and only if  Gap Convolution-3SUM has a solution. Despite clean high-level ideas, there are many intricate technical details that we have swept under the rug, which make the construction and the analysis of \sixcenter on the plane significantly more complicated than $2$-center in~$\real^3$.  However, if we relax our goal slightly, allowing up to $10$ centers, then the proof is simpler, and we will present it first in \Cref{sec:10center}.


\paragraph{$k$-Center in $\real^d$.}  We can rephrase our reduction for \twocenter in $\real^3$ and \sixcenter in $\real^2$ as encoding the solution to the linear equation $a+b = c$ in the Gap Convolution-3SUM problem, with $a = X[i], b = X[j]$ and $c = -X[k]$.  To get a lower bound that grows with $k$ as in \Cref{thm:k-center}, we encode a system of $\Theta(k)$ equations. To formalize this idea, we need a slightly different perspective, in terms of a constraint satisfaction problem (CSP) on graphs. 

Specifically, we view the solution to $a+b = c$ as a CSP on a graph with a single edge $(u,v)$: $u$ can be assigned an integer from $A[-n\ldots n]$, $v$ is from $B[-n\ldots n]$, and the constraint on the edge $(u,v)$ is defined by $C[-n\ldots n]$, in the sense that the value assigned to $u$, say~$a$, and the value assigned to $v$, say~$b$, must satisfy $a+b \in C[-n\ldots n]$. The gadgets we construct for \sixcenter on the plane can thus be used to realize \EMPH{one edge constraint}. 

We can generalize the above observation to a graph with more edges. We choose the grid graphs due to their regular structure in $\real^{d}$. More formally, an instance of the \EMPH{binary SumSet} problem on a grid graph is a tuple $J = (G , D, \mathcal{D})$ where:
\begin{itemize}
\item $G= (V,E)$ is an induced subgraph of the $d$-dimensional grid graph, where each vertex represents a variable. 
\item $D = [n^2]$ is the domain. 
\item $\mathcal{D}$ is a family of subsets of $D$ where vertex $v\in V$ is associated with a subset $D_v\in \mathcal{D}$ and each edge $e\in E$ is associated with a subset $D_{e}\in \mathcal{D}$. Specifically, we 2-color the graph $G$ such that a vertex $v$ gets assigned the color $c(v)\in \{0,1\}$. For vertices $u$ and $v$, $D_v$ is some subset of $[n]\cdot n^{c(v)}$, and $D_{uv}$ is a subset of their possible sums $\{a+b~|~(a,b)\in D_v\times D_u\}\subset [n^2]$. 
\end{itemize}

A \EMPH{solution} of $J$ is an assignment of the variables in $V$ such that each vertex~$v$ is assigned a value $v^* \in D_v$ and for every edge $e=(u,v)$, it holds that $u^*+v^* \in D_{e}$. 
 Roughly speaking, each edge encodes a linear equation; all the edges of the graph together encode a system of $O(k)$ equations, where $k$ is the number of edges of $G$, which is essentially also the number of vertices since grid graphs have constant degree. 

By a simple reduction from binary CSP on grid graphs~\cite{MS14}, we show that under ETH, there is no algorithm with running time $f(k)n^{o(k^{1-1/d})}$ for binary SumSet on grid graphs with $k$ variables and domain of size $n$ (\Cref{lm:binary-sumset}). Our main task then is to reduce the SumSet on grid graphs to $O(k)$-center by a geometric structure that realizes the constraint on each edge using $O(1)$ centers. For this overview, we focus on the plane case $d=2$, which already requires all major ideas.

\begin{figure}[htb]
\centering
\begin{subfigure}[b]{0.45\textwidth}
\centering
\includegraphics[page=6]{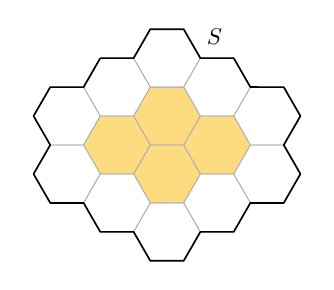}
\subcaption[]{}
 \end{subfigure}
\begin{subfigure}[b]{0.45\textwidth}
\centering
\includegraphics[page=5]{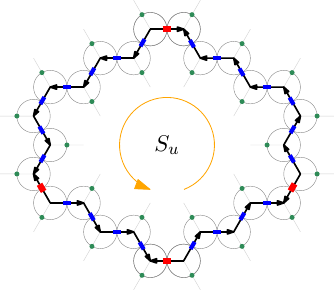}
\subcaption[]{}
 \end{subfigure}
 \caption{The basic curve is a translation of the curve $S$ on the left. On the right, points are placed along the basic curve $S_u$ of some vertex $u$, and disks that cover these points. The orientation is counterclockwise.}
 \label{subfig:2D_generalk_gadget-intro}
\end{figure}

Let $G = (V,E)$ be a subset of the 2D grid graph with $k$ vertices. We associate with each $v\in V$ a curve on the plane called a \EMPH{basic curve}, denoted by $S_{v}$. The basic curve is a translation of the \EMPH{boundary} of the union of 14 hexagonal cells depicted in \Cref{subfig:2D_generalk_gadget-intro}(a). Points, constructed from the subset $D_v$ of $v$, will be added along the hexagonal edges of the basic curves; the tentative solution to $O(k)$ centers will cover all these points by $O(1)$ disks as illustrated in \Cref{subfig:2D_generalk_gadget-intro}(b). We will place a set of points $P_v[(-100n^2)\ldots (100n^2)]$ corresponding to $D_v$ on a hexagonal edge; if a disk in a clustering of all points within a basic curve contains $P_v[(-100n^2)\dots  (2i)]$ at some index~$i$, called the \EMPH{split index}, and another disk contains $P_v[(2i+1)\dots (100n^2)]$, then~$v$ will be assigned a value $i$. Our construction guarantees that $i\in D_v$ and the split indices of other point sets along all hexagonal edges will be the same, except at some special edges, called \EMPH{shared edges}. These edges contain red points in \Cref{subfig:2D_generalk_gadget-intro}.

\begin{figure}[htb]
\centering
\begin{subfigure}[b]{0.3\textwidth}
\vspace{0pt}
\centering
\includegraphics[page=2]{figs/2D_generalk_.pdf}
s \subcaption[]{Grid graph}
 \end{subfigure}%
\begin{subfigure}[b]{0.33\textwidth}
\vspace{0pt}
\centering
\includegraphics[page=4, height=4.2cm]{figs/2D_schematic_view.pdf}
\subcaption[]{Schematic view of the reduction}
 \end{subfigure}%
 \begin{subfigure}[b]{0.33\textwidth}
 \vspace{0pt}
\centering
\includegraphics[width=\textwidth]
{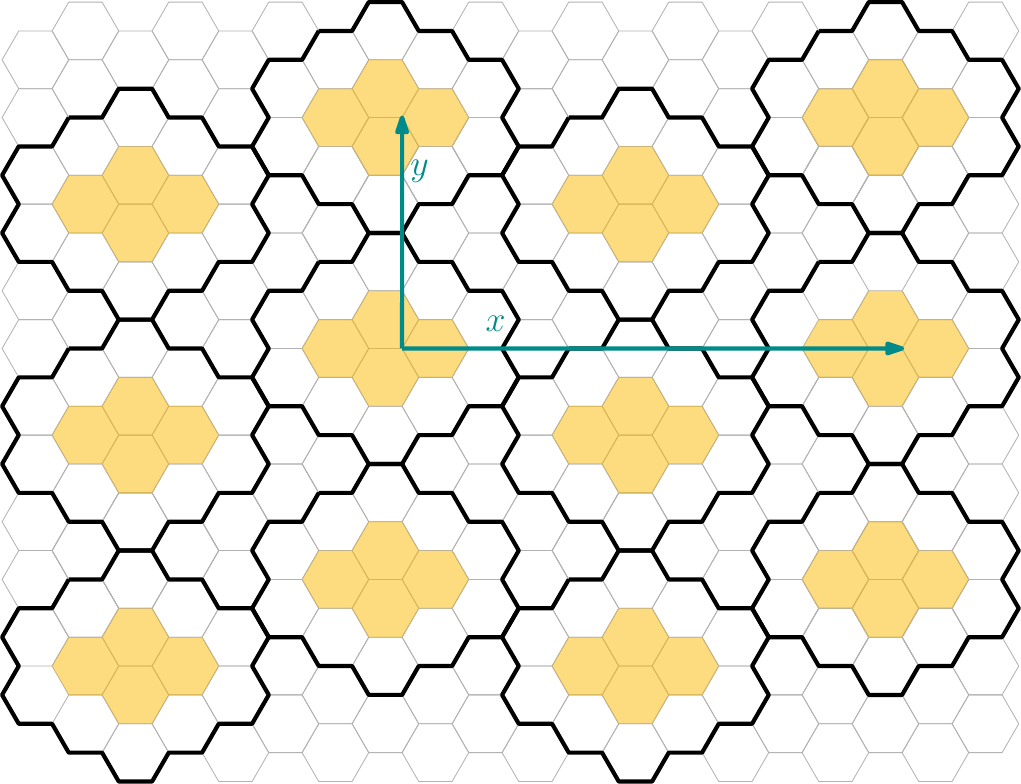}
\subcaption[]{Copies of basic curves} 
 \end{subfigure}
\caption{(a) A grid graph in $\real^2$ with $k = 9$ vertices; (b) A schematic view of the reduction; the shared edges marked in (a) are highlighted. Each arrow denotes the orientation of the edges of the basic curve. (c) Copies of nonoverlapping basic curves with the same incidence structure as the full grid graph. Translational symmetries are indicated by the green arrows.}\label{fig:d2kcenBasic-intro}
\end{figure}

Each basic curve $S_v$ is also associated with an orientation, either clockwise or counterclockwise. The orientations of all the basic curves are designed in such a way that, when we assemble them together to form the grid graph $G$, the shared hexagonal edge of two neighboring basic curves (of two neighboring vertices in $G$) has the same orientation. See \Cref{fig:d2kcenBasic-intro} for an illustration. The shared hexagonal edge between two basic curves, say $S_u$ and $S_v$, of an edge $e=(u,v)\in G$ contains points derived from $D_e$ to encode the constraint $u^* + v^* \in D_{e}$. The intuition here is the same as the intuition for the \sixcenter on the plane depicted in \Cref{fig:basic-gadget-intro}. Most of the technical details, including the placement of anchor points and the introduction of small perturbations to the points, are to guarantee that all the disks in the solution to $O(k)$-center behave as expected. The split index on each non-shared hexagonal edge of each basic curve~$S_u$ corresponds to the assigned value for vertex $u$, which together form a solution to the binary Sumset problem on $G$.

The construction for $d\geq 3$ is very similar: each vertex $u$ has a $d$-dimensional basic curve $S_u$ in $\real^{d}$. The essential additional challenge is that the basic curves now have to arrange in a way as to mimic the connectivity of the grid graph; details are given in \Cref{subsec:kcenter-Ddim}.

\section{Preliminaries}\label{sec:prelims}

In this section, we list the complexity-theoretic assumptions used in this paper. 
Our parameterized complexity results (\Cref{sec: parameterized results}) rely on the Exponential Time Hypothesis and the assumption that ${\sf W}[1] \neq {\sf FPT}$.  

\begin{hypothesis}[Exponential Time Hypothesis (ETH) \cite{IP01,IPZ01}]\label{hyp:eth}
There is some $\delta>0$ such that no $O(2^{\delta n})$-time algorithm that takes as input a 3SAT formula on $n$ variables can decide if it is satisfiable. 
\end{hypothesis}

\begin{hypothesis}[\(\mathsf{W}\lbrack 1\rbrack \neq \mathsf{FPT}\)~\cite{DowneyF13}]\label{hyp:w1} For every computable function \(f\),
there is no \(f(k)\cdot \operatorname{poly}(n)\)-time algorithm that can decide whether an input graph on \(n\) vertices contains a clique of size \(k\).
\end{hypothesis}

In the fine-grained complexity setting, our results (\Cref{sec: 2 center} and \Cref{sec:6CenterPlane}) rely on the assumed complexity of the 3SUM problem. 

\paragraph{3SUM Problem.} Given an array $A[1\ldots n]$ of integers in $\{-N,\ldots, N\}$, decide whether there exist $i,j,k\in [n]$ such that $A[i] + A[j] + A[k] = 0$. 

\begin{hypothesis}[3SUM Hypothesis \cite{gajentaan1995class}]
\label{hyp:3SUM} 
For every $\epsilon >0$, 3SUM cannot be solved in time $O(n^{2-\epsilon})$.   
\end{hypothesis} 

P{\u{a}}tra{\c{s}}cu~\cite{Patrascu10} in his seminal work found a reduction from 3SUM to the Convolution-3SUM problem defined below, ruling out subquadratic time for the latter. See also~\cite{KopelowitzPP16,ChanH20} for randomized and deterministic variants of this reduction.
As already noted in~\cite{AmirCLL14}, the range of elements can be made $\{1,\ldots,n^2\}$ using a randomized hashing reduction from~\cite{BaranDP08,Patrascu10}.

\paragraph{Convolution-3SUM problem.}  Given an array $A[1\ldots n]$ of integers in $\{1,\ldots,n^2\}$, determine whether there exist $i,j,k \in \{1,\ldots,n\}$ such that: (i) $i+j = k$ and (ii) $A[i] + A[j] = A[k]$. 

\begin{theorem}(\cite{Patrascu10})
\label{thm:conv3sum}
For every $\epsilon >0$, Convolution-3SUM cannot be solved in time $O(n^{2-\epsilon})$, unless the 3SUM Hypothesis fails. 
\end{theorem}



Next, we introduce the Gap-Convolution-3SUM problem, which has a gap in the admissible range for $i,j,k$ and is more symmetric than Convolution-3SUM. 


\paragraph{Gap-Convolution-3SUM problem.} Given an array $X[-n\ldots n]$ of $2n+1$ integers in $\{-n^2, \ldots, n^2\}$, decide whether:
\begin{itemize}
    \item  \textbf{YES Case:} $\exists~ i, j, k \in \{-\lfloor \frac{n}{100} \rfloor, \dots, \lfloor \frac{n}{100} \rfloor\}$ such that $i+j+k=0$ and $X[i] + X[j] + X[k] = 0$.
    \item \textbf{NO Case:}  $\not\exists~ i, j, k \in \{-n , \dots, n\}$ such that $i+j+k=0$ and $X[i] + X[j] + X[k] = 0$. 
\end{itemize}
Note that the two cases are disjoint but not exhaustive, making this a promise problem. An algorithm for Gap-Convolution-3SUM may return \textbf{anything} if the instance is not in the YES or NO case, but must return  \textbf{YES} in the YES case and  \textbf{NO} in the NO case

\begin{lemma}\label{lm:gap-3sum-hardness}
If Gap-Convolution-3SUM is solvable in time $O(n^{2-\delta})$ for some absolute constant $\delta > 0$, then the 3SUM Hypothesis fails. 
\end{lemma}
\begin{proof}
Given an instance of Convolution-3SUM, i.e., an array $A[1..n]$ of integers in $\{1,\ldots,n^2\}$, we construct an array $X[-n' \dots n']$ as follows:
\begin{itemize}
    \item $n' := 300n$.
    \item $X[n+i] := A[i]$ for $i \in [n]$.
    \item $X[-2n-i] := -A[i]$ for $i \in [n]$.
    \item $X[i] := 3n^2$ for all other $i$.
\end{itemize}
Note that the entries of $X$ are in $\{-n^2, \dots, 3n^2\} \subseteq \{-(n')^2, \dots, (n')^2\}$. The reduction is in $\Tilde{O}(n)$ time. Now we show the correctness of the reduction.

If $A$ is a YES instance, i.e., $\exists ~i,j,k \in [n]$ s.t. $i+j = k$ and $A[i]+ A[j] = A[k]$, then for $i' := n+i, j' := n+j, k' := -2n-k$, we have $i'+j'+k' = i+j-k = 0$ and $X[i'] + X[j'] + X[k'] = A[i] + A[j] - A[k] = 0$.  Furthermore, $i',j',k'\in \{-3n,\ldots ,3n\} = \{-\lfloor \frac{n'}{100} \rfloor, \dots, \lfloor \frac{n'}{100} \rfloor\}$. Thus, $X$ is a YES instance.

If $X$ is not in the NO case, then there exist $i', j', k'$ such that $i'+j'+k'=0$ and $X[i'] + X[j'] + X[k'] = 0$. Note that no $3n^2$-entry can be part of a solution of  $X[i'] + X[j'] + X[k'] = 0$ and thus, $i',j',k' \in \{n+1, \dots, 2n\}\cup \{-3n, \dots, -2n-1\}$. Observe that $i'+j'+k' = 0$ if and only if two of them come are $\{n+1, \dots, 2n\}$ and the other is from $\{-3n, \dots, -2n-1\}$. Without loss of generality, we assume $i',j' \in \{n+1, \dots, 2n\}$ and $k' \in \{-3n, \dots, -2n-1\}$. Let $i = i'-n, j = j'-n, k = -2n-k'$. Then (i) $i+j - k =  i'+j'+k ' = 0$, giving $i+j = k$ and (ii) $A[i]+ A[j] - A[k] = X[n+i] + X[n+j] + X[-2n-k] = X[i'] + X[j'] + X[k'] = 0$, giving $A[i] + A[j] = A[k]$. Thus, $A$ is not in the NO case. 
\end{proof}
\section{$10$-Center in the Plane}\label{sec:10center}
As a warm up to proving Theorems~\ref{thm:k-center},~\ref{thm:2-center},~and~\ref{thm:6-center}, we first consider the \tencenter problem in the plane. We show in this section by a reduction from the Gap Convolution-3SUM problem that the \tencenter problem in the plane   cannot be solved in time $O(n^{2-\varepsilon})$ for any $\varepsilon>0$ (assuming the $3$-SUM hypothesis). This will illustrate our approach in a conceptual and relatively simple way, and also almost works for more general $k$. Note that later, in~\Cref{sec:6CenterPlane}, we show a stronger result, namely, that we obtain the same statement already for $k=6$ using a series of technical improvements to the \tencenter case. We believe that understanding the $k=10$ case allows for a more structured and altogether cleaner presentation, making the remainder of the paper and our results easier to digest without getting distracted by the detailed calculations.

Our construction uses a hexagonal grid with edge length 2, and with horizontal edges. We assemble edges of the hexagonal grid to form a simple closed curve in~$\real^2$, a \EMPH{basic curve}, which will be key for our reductions. Our strategy is to place $O(n)$ points (depending on the input arrays) around copies of the basic curve, for which there exists a solution to \tencenter with radius \radius, where $\eps:={(10\cdot n)^{-100}}$, if Gap Convolution-3SUM is a YES-case and there is no solution to \tencenter with radius \radius if Gap Convolution-3SUM is a NO-case. To place the points, we proceed in several stages in which we add points of different types, each type serving a distinct purpose in controlling the admissible clusterings. We assume that $n\geq 100$, as otherwise the only possible solution in the YES case is $i=j=k=0$, which we can test in constant time. 

 We call edges and vertices of the hexagonal grid \EMPH{hexagonal edges} and \EMPH{hexagonal vertices}, respectively. The point in the center of a hexagonal grid cell is called \EMPH{hexagonal center vertex}. Each basic curve has an associated \EMPH{orientation}, either clockwise or counterclockwise, that we use in the construction of points along hexagonal edges. 

The first step is to ensure that the centers of disks with radius \radius appearing in the solution of the \kcenter problem lie (close to) the hexagonal (center) vertices. To guarantee this, we construct a set of \EMPH{anchor points} around each of these vertices. These represent the most primitive type of point in our point set. 

Already at this stage, we can remark that the technical adaptations to the more conceptual approach for \tencenter to tackle \kcenter amounts to 1) changing the basic curve and 2) ensuring that the clustering behaves as expected along the new curve. In \Cref{subsec:kcenter-Ddim}, we show that a further modification of the basic curves allows us to encode more dimensions, culminating in the proof of~\Cref{thm:k-center}.

\paragraph{Anchor Points.}
We construct anchor points that ensure that the centers of the \kcenter solution all lie roughly at a hexagonal (center) vertex. 
To this end, let $\mathcal{H}$ be the set of the hexagonal (center) vertices, where centers of disks of a solution to the \kcenter problem roughly lie. The cardinality of $\mathcal{H}$ is $k$. For concreteness, in this section, the \emph{rough} placement of a point at some location is always such that the distance between the point and the location is either at most $0.1n\eps$ or at most $0.3n\eps$.    
Defining the unit vectors
\begin{align*}
    &v_1=(1,0) && v_2=(\sqrt{3}/2, 1/2)&& v_3=(1/2, \sqrt{3}/2) && v_4=(0, 1) &&v_5=(-1/2, \sqrt{3}/2) && v_6=(-\sqrt{3}/2, 1/2),
\end{align*}
as well as $v_k=-v_{k-6}$ for $k=7, \dots, 12$ and $v_k=v_{k \mod 12}$ for $k>12$, yields  
\begin{align*}
    \mathcal{H}&\subset \{2a_1\cdot v_1+2a_3\cdot v_3+ 2a_5\cdot v_5\mid a_1,a_3,a_5\in \mathbb{Z}\}.
\end{align*}
For every point $h\in \mathcal{H}$, we place six anchor points around it (see \Cref{fig:K-cen_anchor_unique_solution}) such that the set of anchor points $\mathcal{A}$ is
\begin{align*}
    \mathcal{A}&=\{h+\da\cdot v_{2k}\mid h\in \mathcal{H}, k\in \{1, \dots, 6\}\}, \text{ with } \Delta=1-(n/10+1)\eps.
\end{align*}


\begin{figure}
    \centering
        \includegraphics[page=1]{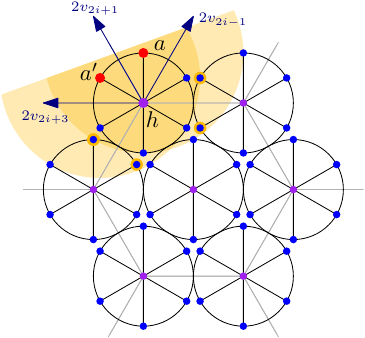}
    \caption{The hexagonal (center) vertices are drawn in purple and the anchor points are blue. The vertex $h$ is on the boundary of the convex hull of $\mathcal{H}$. Two disks of radius $2(\radius)$, centered at points of the convex hull of the hexagonal vertices, are shown in orange. The shaded points are those that can be covered by either disk.}
    \label{fig:K-cen_anchor_unique_solution}
\end{figure}

We show that there is only one way to cluster $\mathcal{A}$ into $k$ clusters with radius \radius, namely with disks that are centered roughly at the points in $\mathcal{H}$. In order to prove this, we start by showing that there does not exist a disk of radius \radius that contains more than six points of $\mathcal{A}$.

\begin{lemma}\label{lem:K-cen_anchor_at_most_6}
    Let $\mathcal{A}$ be the set of all anchor points defined by $\mathcal{H}$. Then, any disk~$D$ of radius \radius contains at most six points of $\mathcal{A}$. 
\end{lemma}

\begin{proof}
    Assume there exists a disk~$D$ with more than six anchor points. Then by construction of the anchor points, there are two anchor points $a$, $a'$ in~$D$ such that $a=h+\Delta v_{2i}$ and $a'=h'+\Delta v_{2i}$ with $h\neq h'$. Hence, $\|a-a'\|=\|h- h'\|\geq 2> 2\cdot (\radius)$. This is a contradiction to $a$ and $a'$ being contained in $D$.
\end{proof}

The next lemma is used to show that the only solution to \kcenter is to center the disks roughly at the points in~$\mathcal{H}$ and will be used in the \sixcenter construction again.

\begin{lemma}\label{lem:fixing_center}
    Let $h\in \mathbb{R}^2$, $D\subseteq\mathbb{R}^2$, $\delta\in (0,0.01)$ such that $\|d-h\|=1-\delta$ for all $d\in D$ and the angle between any two consecutive points $d, d'\in D$ is $\angle dhd'\leq \frac{3}{4}\pi$. Then, any disk of radius~$\radius$ containing $D$ has a center $c$ with $\|c-h\| < 2.7 \cdot \delta$.
\end{lemma}
\begin{proof}
    Consider a segment of angle $3\pi/4$ of the disk of radius~$(1-\delta)$ centered at $h$ that points away from~$c$.
    By the definition of the point set $D$, this segment contains a point $d\in D$. Observe that $\angle dhc\in [\pi-\frac{3}{8}\pi, \pi]$. 
    Using the cosine formula, we get that 
    \[\cos(\angle dhc)=\frac{\|h-c\|^2+\|h-d\|^2-\|d-c\|^2}{2\cdot \|h-c\|\cdot\|h-d\|}.\]
    Hence, 
    \begin{align*}
    0=&\|h-c\|^2-\left(2\cdot \|h-d\|\cdot\cos(\angle dhc)\right)\cdot\|h-c\|+\left(\|h-d\|^2-\|d-c\|^2\right)\\
    &\geq \|h-c\|^2+2\cos(\frac{3}{8}\pi)(1-\delta)\cdot \|h-c\|-2\delta+\delta^2.
    \end{align*}
    Using $\delta\leq 0.01$ and $\sqrt{a^2+b}=a\cdot \sqrt{1+\frac{b}{a^2}}\leq a\cdot\left(1+\frac{b}{2a^2}\right)=a+\frac{b}{2a}$ for all $a>0$, $b\geq 0$, it follows that
    \begin{align*}
        \|h-c\|\leq& -\cos(\frac{3}{8}\pi)(1-\delta)+\sqrt{\left(\cos(\frac{3}{8}\pi)(1-\delta)\right)^2+2\delta-\delta^2}\\
        \leq& \frac{2\delta-\delta^2}{2\cos(\frac{3}{8}\pi)(1-\delta)}\leq \frac{\delta}{\cos(\frac{3}{8}\pi)(1-\delta)}<2.7 \cdot \delta.
    \end{align*}
\end{proof}

\begin{lemma}\label{lem:K-cen_anchor_unique_solution}
    In every solution to \kcenter of radius \radius for the set~$\mathcal{A}$ of anchor points defined by~$\mathcal{H}$ with $|\mathcal{H}|=k$, for every point $h$ in~$\mathcal{H}$ there is a disk centered at a point with distance at most $0.3n\eps$ to~$h$. Further, let $\mathcal{C}$ be a point set of $k$ points such that for each point $h$ in $\mathcal{H}$ there is a point in $\mathcal{C}$ within distance $0.1n\eps$. Then, $\mathcal{C}$ is a \kcenter solution for $\mathcal{A}$ of radius at most \radius.
\end{lemma}
\begin{proof}
    The second part of the lemma follows directly by construction of the anchor points. 
    We prove the first part by induction on the number of points in $\mathcal{H}$. 
    If $\mathcal{H}=\{h\}$, then all points have to be in the same cluster. In this case we can apply \Cref{lem:fixing_center} and it follows that $\|c-h\|\leq 2.7(\frac{n}{10}+1)\eps\leq 2.7\cdot 1.1\cdot \frac{n}{10} \eps \leq 0.3n\eps$, since $n\geq 100$.
    Now let $|\mathcal{H}|=k\geq 1$. Let $h$ be a vertex of the convex hull of $\mathcal{H}$. 
    Since $h$ is a vertex of the convex hull, there exists an index $i$ such that $\mathcal{H}\cap \{h+(a2v_{2i-1}+b2v_{2i+1}+c2v_{2i+3})\mid a,b,c\in \mathbb{N}_0\}=\{h\}$. Define $a=h+\da \cdot v_{2i}$ and $a'=h+\da \cdot v_{2(i+1)}$. Then, $a$ and $a'$ are anchor points defined by $h$ and the disk $D_a$ (resp. $D_{a'}$) of radius $2\cdot (\radius)$ centered at~$a$ (resp.~$a'$) contains at most eight anchor points (see \Cref{fig:K-cen_anchor_unique_solution} and the proof of \Cref{lem:K-cen_anchor_at_most_6}). 
    Further, there are exactly the six anchor points defined by~$h$ in $D_a\cap D_{a'}$ and thus exactly $10$ anchor points in $D_a\cup D_{a'}$ (the highlighted points in \Cref{fig:K-cen_anchor_unique_solution}). By a counting argument and \Cref{lem:K-cen_anchor_at_most_6}, we know that in an optimal solution every disk contains exactly six anchor points and every anchor point is contained in exactly one disc. 
    Assume that $a$ and $a'$ are not contained in the same disc. Then, the union of the disk containing $a$ and the disk containing $a'$ contains $12$ anchor points, which leads to a contradiction. 
    Therefore, $a$ and $a'$ have to lie in the same disc. Hence, all six anchor points defined by~$h$ lie in the same disk and this disk has to be centered at a point within distance at most $0.3n\eps$ to $h$ by the calculations in the base case. By induction, it follows that in the solution to the anchor points defined by $\mathcal{H}\setminus \{h\}$ all disks are centered at points with distance at most $0.3n\eps$ to a point in $\mathcal{H}\setminus \{h\}$, which concludes the proof.
\end{proof}

\paragraph{Reduction from Gap Convolution-3SUM.}
Let $X=X[-n\dots n]$ be an instance of Gap Convolution-3SUM. Then, the task is to find indices with $i, j, k$ such that $i+j+k=0$ and $X[i]+X[j]+X[k]=0$ or to decide that no such indices with $i, j, k\in \{-\lfloor n/100\rfloor, \dots, \lfloor n/100\rfloor\}$ exist.

We define two basic curves, where each is the boundary of a hexagon of side length~$2$. They are arranged such that they share a (horizontal) hexagonal edge; see \Cref{subfig:2D_SumSet_structure}. We denote the two curves with~$S_u$ and~$S_v$. Intuitively, the two curves correspond to the indices~$i$ and~$j$. The orientation of $S_u$ is counterclockwise, while that of $S_v$ is clockwise to guarantee that the orientation on the shared edge is the same in both hexagons. We call a hexagonal edge $(a,b)$ a \EMPH{non-shared edge} if it is only part of one basic curve. Otherwise, we call it \EMPH{shared edge}. The shared edge intuitively corresponds to the index~$k$. 

By placing the six anchor points around each hexagonal vertex of the basic curves, we ensure that the centers of the \tencenter solution lie roughly at the hexagonal vertices of the two basic curves, by \Cref{lem:K-cen_anchor_unique_solution} (see \Cref{subfig:2D_SumSet_disks}).

\begin{figure}[htb]
    \centering
     \begin{subfigure}[b]{0.45\textwidth}
        \centering
        \includegraphics[page=4]{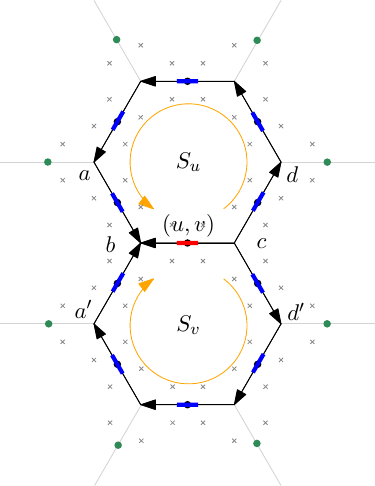}
        \subcaption[]{}
        \label{subfig:2D_SumSet_structure}
     \end{subfigure}
    \begin{subfigure}[b]{0.45\textwidth}
        \centering
        \includegraphics[page=5]{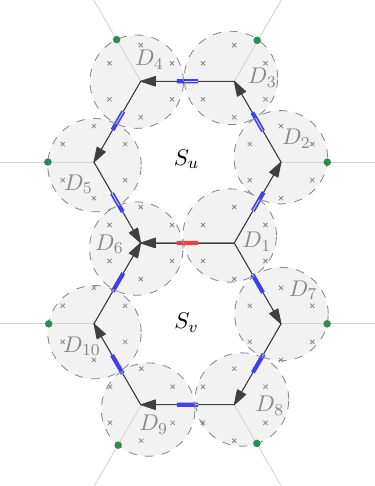}
        \subcaption[]{}
        \label{subfig:2D_SumSet_disks}
     \end{subfigure}
    \caption{A schematic view of the reduction from Gap Convolution-3SUM to \tencenter. The points of type (1) are drawn in blue on the edges of $S_u$, $S_v$, of type (2) in red on the shared edge, and of type (3) as green circles. 
    The anchor points are drawn as gray crosses.}
    \label{fig:2dsimplified}
\end{figure}

The points placed on the edges of the basic curves are defined using a specially crafted set of points associated to an array. We will make use of the definition of this point set throughout the paper.

For a given array $X=X[-n\ldots n]$ of $(2n+1)$ integers, a central point $s\in \real^d$, a vector $\Vec{w} \in \real^d$, and an offset constant $\alpha\geq 0$, we define a set of $2(2n+1)$ points in $\real^d$, denoted by $ \pset_n(X,s,\Vec{w},\alpha)[(-2n)\ldots (2n+1)]$ where for every $i\in I_n:=\{-2n, \dots, 2n+1\}$:

\begin{equation}\label{eq:Pn-def}
    \pset_n(X,s,\Vec{w},\alpha)[i] \coloneq s + \left( 3 \lfloor i/2 \rfloor \cdot \eps + X[\lfloor i/2 \rfloor] \cdot \eps^{1.5} - (-1)^i \alpha \eps \right) \cdot \Vec{w}. 
\end{equation}
In the \tencenter, \kcenter, and \sixcenter reduction, the vector is always defined as a unit vector and $\alpha=1$.
\Cref{fig:coveringbasic} (left) illustrates the definition of the point sets $\pset_n$ from three different arrays and parameters in the plane (see also \Cref{lm:covering-basic_new}).


\begin{observation}\label{obs:P_n_and_reverse}
    For a given array $A=A[-n\dots n]$, define an array $A'=A'[-n\dots n]$ such that ${A'[i]=-A[-i]}$. Then, for any point $s\in \mathbb{R}^2$, vector $\Vec{w}$, and $\alpha\geq 0$, it holds that for all $i\in I_n$
    \[\pset_n(A, s, \Vec{w}, \alpha)[i]=\pset_n(A', s, -\Vec{w}, \alpha)[-i+1].\]
\end{observation}
Throughout this paper, we use the following notation. For any $i\in I_n$, we define 
\begin{align*}
    \pset_n(A, s, \Vec{w}, \alpha)[\leq i]&:= \{\pset_n(A, s, \Vec{w}, \alpha)[i']\mid i'\leq i, i'\in I_n\}, \text{ and}\\
    \pset_n(A, s, \Vec{w}, \alpha)[> i]&:= \{\pset_n(A, s, \Vec{w}, \alpha)[i']\mid i'> i, i'\in I_n\}.
\end{align*}

Using the definition of $\pset_n(X,s,\Vec{w},\alpha)$, we describe the points placed on the edges of $S_u$ and $S_v$. The basic curve~$S_u$ (resp. $S_v$) corresponds to an array~$X_u$ (resp. $X_v$) and the shared edge of the curves $S_u$ and $S_v$ to an array $X_{(u,v)}$. In the \kcenter reduction, these arrays differ. However, in the \tencenter reduction all sets are the same and we define $X_u=X_v=X_{(u,v)}=X$.

Let $(a,b)$ be a non-shared edge that is part of the basic curve $S_u$. Then, we add the following points on $(a, b)$, which we refer to as points of type (1):
\begin{enumerate}
    \item[(1)] $\PSEdge{a}{b}=\pset_n(X_{u},p,\overrightarrow{ap},1)$, where $p$ is the midpoint of $(a,b)$.
\end{enumerate}

Similar, we add points on the non-shared edges of $S_v$.
The orientations of the shared edge $(a,b)$ of the basic curves $S_u$ and $S_v$ agree. We add points on the edge $(b,a)$, i.e., points are ordered in the opposite direction of the shared edge $(a, b)$. For the shared edge $(b,a)$, we add a point set of type (2):
\begin{enumerate}
    \item[(2)] $\PSEdge{b}{a}=\pset_{n}(X_{(u,v)},p,\overrightarrow{bp},1)$, where $p$ is the midpoint of $(b, a)$. 
\end{enumerate}

\begin{figure}
    \centering
    \includegraphics[page=1, width=0.4\textwidth]{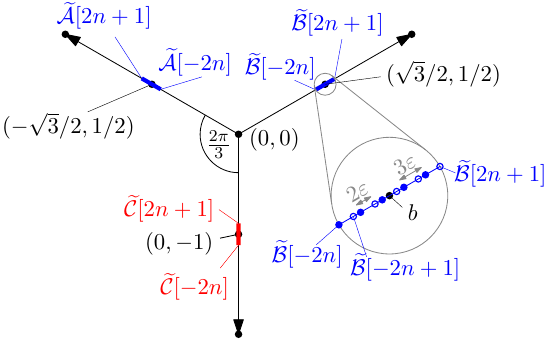}\hspace{1cm}
    \includegraphics[page=2, width=0.4\textwidth]{figs/Gadget_6Center_Plane1.pdf}
    \caption{(Left): The point sets $\widetilde{\mathcal{A}},\widetilde{\mathcal{B}},$ and $\widetilde{\mathcal{C}}$ associated to the three input arrays of \Cref{lm:covering-basic_new} (Right): A zoomed in version of the point sets, illustrating \Cref{lm:covering-basic_new}, by showing how an inequality constraint is captured by a covering condition for the constructed point sets.}
    \label{fig:coveringbasic}
\end{figure}
    
The following lemma shows how points on the shared edge together with points on their neighboring edges can be covered by a disk of radius $r=\radius$, see \Cref{fig:coveringbasic}.

\begin{lemma}\label{lm:covering-basic_new}
    For any $A[-n\dots n]$, $B[-n\dots n]$, $C[-n\dots n]$ with entries in $\{-n^2, \dots, n^2\}$, let $(z,a)$, $(z,b)$, $(z,c)$ be the three unit edges such that the angle between each pair is $2\pi/3$. Define
    \begin{align*}
        \widetilde{\mathcal{A}}:= \pset_n(A, a, \overrightarrow{za}, 1),&& \widetilde{\mathcal{B}}:= \pset_n(B, b, \overrightarrow{zb}, 1), &&\widetilde{\mathcal{C}}:= \pset_n(C, c, \overrightarrow{zc}, 1).
    \end{align*}
    For any $i, j, k\in I_n$:
    \begin{enumerate}
        \item If $i+j+k>0$, there does not exist a radius~$r:=\radius$ disk containing $\widetilde{\mathcal{A}}[i], \widetilde{\mathcal{B}}[j]$, and $\widetilde{\mathcal{C}}[k]$. \label{lm:item_i+j+k>0}
        \item If $i+j+k=0$ and if $i$ or $j$ or $k$ is odd or $A[i/2]+B[j/2]+C[k/2]>0$, then there does not exist a radius~$r$ disk containing $\widetilde{\mathcal{A}}[i], \widetilde{\mathcal{B}}[j]$, and $\widetilde{\mathcal{C}}[k]$. \label{lm:item_i+j+k=0_not}
        \item If $i+j+k=0$ and $i$ and $j$ and $k$ are even and $A[i/2]+B[j/2]+C[k/2]\leq 0$, then there exists a radius~$r$ disk containing $\widetilde{\mathcal{A}}[\leq i], \widetilde{\mathcal{B}}[\leq j]$, and $\widetilde{\mathcal{C}}[\leq k]$ and having its center within distance $3\eps(\max\{|i|, |j|\}+1)$ of $z$. \label{lm:item_i+j+k=0_yes}
    \end{enumerate}
\end{lemma}
\begin{proof}
After rotation and translation, we can assume that $z=(0,0)$, $a=(-\sqrt{3}/2, 1/2)$, $b=(\sqrt{3}/2, 1/2)$, and $c=(0, -1)$. 
Consider a disk with center~$d$ with radius~$r=\radius$ containing $\widetilde{\mathcal{A}}[i]$, $\widetilde{\mathcal{B}}[j]$, and $\widetilde{\mathcal{C}}[k]$. 
Then, it holds that
\begin{align*}
    &\radius\geq \|\widetilde{\mathcal{A}}[i]-d\|\geq (\widetilde{\mathcal{A}}[i]-d)\cdot \Vec{a}=\widetilde{\mathcal{A}}[i]\Vec{a}-d\Vec{a}=1+3\lfloor i/2\rfloor \eps+A[\lfloor i/2\rfloor] \eps^{1.5}-(-1)^i\eps-d\Vec{a}\\
    \Leftrightarrow\ &0\geq (3\lfloor i/2\rfloor+1-(-1)^i) \cdot\eps+A[\lfloor i/2\rfloor]\cdot \eps^{1.5}-\eps^{1.7}-d\Vec{a}.
\end{align*}
Similarly, it follows that 
\begin{align*}
    0\geq (3\lfloor j/2\rfloor+1-(-1)^j)\cdot\eps+B[\lfloor j/2\rfloor]\cdot \eps^{1.5}-\eps^{1.7}-d\Vec{b}, \\
    0\geq (3\lfloor k/2\rfloor+1-(-1)^k) \cdot\eps +C[\lfloor k/2\rfloor]\cdot \eps^{1.5}-\eps^{1.7}-d\Vec{c}.
\end{align*}
Summing up the three inequalities and noting that $\Vec{a}+\Vec{b}+\Vec{c}=\Vec{0}$, we obtain 
\begin{align}\label{eq:basic_1}
    0\geq \left(\sum_{m\in \{i,j,k\}}3\lfloor m/2\rfloor+1-(-1)^m\right)\cdot\eps
    +\left(A[\lfloor i/2\rfloor]+B[\lfloor j/2\rfloor]+C[\lfloor k/2\rfloor]\right)\cdot \eps^{1.5}-3\eps^{1.7}=:f.
\end{align}
Note that $3\cdot \lfloor m/2\rfloor+1-(-1)^m$ is equal to $(3/2)m$ if $m$ is even and equal to $(3/2)m+(1/2)$ if $m$ is odd. Therefore, if $i+j+k>0$, then $i+j+k\geq 1$ and 
    $f \geq (3/2)\eps-3n^2\eps^{1.5}-3\eps^{1.7}>0$,
which contradicts (\ref{eq:basic_1}) and hence shows claim \ref{lm:item_i+j+k>0}.
Similarly, if $i+j+k=0$ and $i$ or $j$ or $k$ is odd, then 
$f\geq (1/2)\eps-3n^2\eps^{1.5}-3\eps^{1.7}>0$, which contradicts (\ref{eq:basic_1}). If $i+j+k=0$ and $i, j, k$ are even, then $3(\lfloor i/2\rfloor+\lfloor j/2\rfloor+\lfloor k/2\rfloor+1)-(-1)^{i}-(-1)^{j}-(-1)^{k}=0$. If  $A[\lfloor i/2\rfloor]+B[\lfloor j/2\rfloor]+C[\lfloor k/2\rfloor]>0$, then $f\geq \eps^{1.5}-3\eps^{1.7}>0$, which again contradicts (\ref{eq:basic_1}). Hence, claim \ref{lm:item_i+j+k=0_not}.\ follows.

For claim \ref{lm:item_i+j+k=0_yes}., we consider the disk with center $d=(d_x, d_y)$, where
\begin{align*}
    d_x&:=(\sqrt{3}/2)\cdot \left((j-i)\eps+(2/3)(B[j/2]-A[i/2])\eps^{1.5}\right)\text{ and }\\
    d_y&:=(1/2)\cdot \left(3(i+j)\eps+2(A[i/2]+B[j/2])\eps^{1.5}\right).
\end{align*}
Since $i$, $j$, $k$ are even, it holds that
\begin{align*}
    \widetilde{A}[i]&=(1+3(i/2)\eps-\eps+A[i/2]\eps^{1.5})\cdot (-\sqrt{3}/2, 1/2), \\
    \widetilde{B}[j]&=(1+3(j/2)\eps-\eps+B[j/2]\eps^{1.5})\cdot (\sqrt{3}/2, 1/2)\text{, and}\\
    \widetilde{C}[k]&=(1+3(k/2)\eps-\eps+C[k/2]\eps^{1.5})\cdot (0, -1).
\end{align*}
Therefore, we have
\begin{align*}
    \|\widetilde{A}[i]-d\|^2=&(3/4)\cdot \left(1+(i/2)\eps-\eps+j\eps+((1/3)A[i/2]+(2/3)B[j/2])\eps^{1.5}\right)^2\\
    &+(1/4)\cdot \left(1-3(i/2)\eps-\eps-3j\eps-(A[i/2]+2B[j/2])\eps^{1.5}\right)^2\\
    =&1-2\eps+\eta \text{\hspace{2cm} for } |\eta|\leq100n^4\eps^2<\eps^{1.7}.
\end{align*}
It holds that $r^2=(\radius)^2\geq 1-2\eps+2\eps^{1.7}$. Hence, a radius~$r$ disk around~$d$ covers $\widetilde{A}[i]$. 
Analogously, it follows that
\begin{align*}
    \|\widetilde{B}[j]-d\|^2=&(3/4)\cdot \left(1+(j/2)\eps-\eps+i\eps+((1/3)B[j/2]+(2/3)A[i/2])\eps^{1.5}\right)^2\\
    &+(1/4)\cdot \left(1-3(j/2)\eps-\eps-3i\eps-(B[j/2]+2A[i/2])\eps^{1.5}\right)^2\leq r^2.
\end{align*}
Further, it holds that
\begin{align*}
    \|\widetilde{C}[k]-d\|^2&={d_x}^2+\left(1+3(i/2+j/2+k/2)\eps-\eps+(A[i/2]+B[j/2]+C[k/2])\eps^{1.5}\right)^2\\
    &\leq {d_x}^2+(1-\eps)^2=1-2\eps+\eta'\leq r^2\text{\hspace{2cm} for } |\eta'|\leq100n^4\eps^2<\eps^{1.7}
\end{align*}
Observe that for all $i'\leq i$ both the $x$ and $y$ coordinate of the point $\widetilde{A}[i']$ is closer to $d_x$ and $d_y$ compared to the point $\widetilde{A}[i]$. Hence, $\widetilde{A}[\leq i]$ is contained in the radius~$r$ disk around $d$. Similar it follows that $\widetilde{B}[\leq j]$ and $\widetilde{C}[\leq k]$ are also contained in the radius~$r$ disk around $d$.
The distance $\|d-z\|=\|d\|$ is bounded by
\begin{align*}
    \|d\|&\leq \|\left((\sqrt{3}/2)(i-j)\eps, (3/2)(i+j)\eps\right)\|&&+\|\left((1/\sqrt{3}(A[i/2]+B[j/2])\eps^{1.5}, (A[i/2]+B[j/2])\eps^{1.5}\right)\|\\
    &=\sqrt{(3/4)(i-j)^2+(9/4)(i+j)^2}\cdot \eps &&+ \sqrt{\left((1/\sqrt{3})2n^2\eps^{1.5}\right)^2+\left(2n^2\eps^{1.5}\right)^2}\\
    &\leq \sqrt{3i^2+3j^2+3ij}\cdot \eps&&+10n^4\eps^{1.5}\\
    &\leq 3\max\{|i|, |j|\} \cdot \eps&&+\eps.
\end{align*}
\end{proof}

\begin{remark}\label{rem:covering_basic_2eps1.7}
    The statement of \Cref{lm:covering-basic_new} in its entirety also holds when the radius is replaced by $1-\eps+2\eps^{1.7}$.
\end{remark}

The next lemma follows directly by \Cref{obs:P_n_and_reverse} and \Cref{lm:covering-basic_new}. This is used to show a connection between the points on the edges $(a, b)$, $(a', b)$, and $(c, b)$ in \Cref{fig:2dsimplified}.
\begin{lemma}\label{lm:D4-covering}
    For any $A[-n\dots n]$, $B[-n\dots n]$, $C[-n\dots n]$ with entries in $\{-n^2, \dots, n^2\}$, let $(a, z)$, $(b, z)$, $(c, z)$ be the three unit edges such that the angle between each pair is $2\pi/3$. Define
    \begin{align*}
        \widetilde{\mathcal{A}}:= \pset_n(A, a, \overrightarrow{az}, 1),&& \widetilde{\mathcal{B}}:= \pset_n(B, b, \overrightarrow{bz}, 1), &&\widetilde{\mathcal{C}}:= \pset_n(C, c, \overrightarrow{cz}, 1).
    \end{align*}
    For any $i, j, k\in I_n$:
    \begin{enumerate}
        \item If $i+j+k<0$, there does not exist a radius~$r:=\radius$ disk containing $\widetilde{\mathcal{A}}[i+1], \widetilde{\mathcal{B}}[j+1]$, and $\widetilde{\mathcal{C}}[k+1]$. 
        \item If $i+j+k=0$: If $i$ or $j$ or $k$ is odd or $A[i/2]+B[ j/2]+C[k/2]<0$, then there does not exist a radius~$r$ disk containing $\widetilde{\mathcal{A}}[i+1], \widetilde{\mathcal{B}}[j+1]$, and $\widetilde{\mathcal{C}}[k+1]$. 
        \item If $i+j+k=0$ and $i$ and $j$ and $k$ are even and $A[ i/2 ]+B[ j/2]+C[k/2]\geq 0$, then there exists a radius~$r$ disk containing $\widetilde{\mathcal{A}}[>i], \widetilde{\mathcal{B}}[> j]$, and $\widetilde{\mathcal{C}}[>k]$ and having its center within distance $3\eps(\max\{|i|, |j|\}+1)$ of $z$. 
    \end{enumerate}
\end{lemma}
\begin{proof}
    Define $\mathcal{A'}:=\pset_n(A', a, \Vec{za}, 1)$, $\mathcal{B'}:=\pset_n(B', b, \Vec{zb}, 1)$, and $\mathcal{C'}:=\pset_n(C', c, \Vec{zc}, 1)$, where $A'[i]=-A[-i]$, $B'[i]=-B[-i]$, and $C'[i]=-C[-i]$ for all $i\in I_n$. 
    Then, we can use \Cref{lm:covering-basic_new} for the point sets $\mathcal{A'}, \mathcal{B'}, \mathcal{C'}$. By \Cref{obs:P_n_and_reverse} it holds that $\mathcal{A'}[-i+1]=\mathcal{A}[i]$, $\mathcal{B'}[-i+1]=\mathcal{B}[i]$, and $\mathcal{C'}[-i+1]=\mathcal{C}[i]$ for all $i\in I_n$.
    \begin{enumerate}
        \item If $i+j+k<0$, then by \ref{lm:item_i+j+k>0} of \Cref{lm:covering-basic_new} it follows that no disk of radius $\radius$ containing $\mathcal{A'}[-i]=\mathcal{A}[i+1]$, $\mathcal{B'}[-j]=\mathcal{B}[j+1]$, and $\mathcal{C'}[-k]=\mathcal{C}[k+1]$ exists.
        \item Assume that $i+j+k=0$. If at least one of $i, j, k$ is odd or  $A[i/2]+B[j/2]+C[k/2]<0$, i.e., $A'[-i/2]+B'[-j/2]+C'[-k/2]>0$, then by \ref{lm:item_i+j+k=0_not} of \Cref{lm:covering-basic_new} the second item follows. 
        \item Assume that $i+j+k=0$ and $i, j, k$ are even and $A[i/2]+B[j/2]+C[k/2]\geq 0$. 
        Then, $A'[-i/2]+B'[-j/2]+C'[-k/2]\leq 0$. Hence, by \ref{lm:item_i+j+k=0_not} of \Cref{lm:covering-basic_new}, there exists a disk of radius $\radius$ containing $\mathcal{A'}[\leq (-i)]=\mathcal{A}[> i]$, $\mathcal{B'}[\leq (-j)]=\mathcal{B}[> j]$, and $\mathcal{C'}[\leq (-k)]=\mathcal{C}[>k]$ and having its center within distance $3\eps(\max\{|i|, |j|\}+1)$ of $z$.
    \end{enumerate}
    This concludes the proof.
\end{proof}

The next lemma is essentially a consequence of \Cref{lm:covering-basic_new} and \Cref{lm:D4-covering}.

\begin{lemma}\label{lem:shared_edge}
    Let $(a,b)$, $(b,c)$, and $(c, d)$ be hexagonal edges of $S_u$, and $(a', b)$, $(b,c)$, and $(c, d')$ hexagonal edges of $S_v$, such that $(b, c)$ is shared (see \Cref{subfig:2D_SumSet_structure}). Further, let $i$, $j\in \{0, \dots, 2n+1\}$. Then, there exist two disks $B_b$, $B_c$ of radius \radius centered at a point within distance $3\eps (\max\{|i|, |j|\}+1)$ to $b$, resp. to $c$, that cover 
    \[\PSEdge{a}{b}[>i)]\cup \PSEdge{a'}{b}[>j]\cup \PSEdge{c}{b}\cup \PSEdge{c}{d}[\leq i]\cup \PSEdge{c}{d'}[\leq j]\]
    if and only if $i$ and $j$ are even and $X_u[i/2]+X_v[j/2]+X_{(u,v)}[-(i+j)/2]=0$.
\end{lemma}
\begin{proof}
    First assume that $i$ and $j$ are even and $X_u[i/2]+X_v[j/2]+X_{(u,v)}[-(i+j)/2]=0$. Then, by \Cref{lm:D4-covering} there exists a disk $B_b$ within distance $3\eps (\max\{|i|, |j|\}+1)$ to $b$ that covers $\PSEdge{a}{b}[>i]\cup {\PSEdge{a'}{b}[>j]}\cup \PSEdge{c}{b}[>-(i+j)]$. Further, by \Cref{lm:covering-basic_new} there exists a disk~$B_c$ within distance $3\eps (\max\{|i|, |j|\}+1)$ to $c$ that covers $\PSEdge{c}{b}[\leq -(i+j)]\cup \PSEdge{c}{d}[\leq i]\cup \PSEdge{c}{d'}[\leq j]$.

    It remains to prove the other direction. First, observe that points in $\PSEdge{a}{b}[> i]\cup \PSEdge{a'}{b}[> j]$ must be covered by disk~$B_b$. Similarly, points in $\PSEdge{c}{d}[\leq i]\cup \PSEdge{c}{d'}[\leq j]$ must be covered by disk $B_c$. 
    By \Cref{lm:D4-covering}, it holds that there does not exist a disk of radius \radius that contains $\PSEdge{a}{b}[i+1]\cup \PSEdge{a'}{b}[j+1] \cup \PSEdge{c}{b}[k]$ for any $k< -i-j+1$. By \Cref{lm:covering-basic_new}, it holds that there does not exist a disk of radius \radius that contains $\PSEdge{c}{d}[i]\cup \PSEdge{c}{d'}[j] \cup \PSEdge{c}{b}[k]$ for any $k> -i-j$. Hence, $\PSEdge{c}{b}[>-(i+j)]$ is contained in $B_b$ and $\PSEdge{c}{b}[\leq -(i+j)]$ is contained in $B_c$.
    Therefore, by \Cref{lm:covering-basic_new} and \Cref{lm:D4-covering}, it holds that $i$ and $j$ are even and $X_u[i/2]+X_v[j/2]+X_{(u,v)}[-(i+j)/2]=0$.
\end{proof}

We want to enforce that there exists an even index $i\in I_n$ such that in every \tencenter solution of radius \radius for every non-shared edge $(a,b)$ of $S_u$, the points in $P(a,b)[\leq i]$ and the points in $P(a,b)[>i]$ are contained in different disks; and similarly for an even index $j$ for $S_v$. 
In order to achieve this, we add points of type (3).

For every hexagonal vertex $b$ that is incident to two non-shared edges $(a, b)$ and $(b,c)$ and no further edge, we add one point, which we call \EMPH{consistency point}:

\begin{enumerate}
    \item [(3)] The consistency point $c_{b}$ has distance $1-\eps$ to $b$ and is placed on the third incident hexagonal edge to $b$ that is not $(a, b)$ or $(b, c)$ (in \Cref{subfig:2D_SumSet_disks} the consistency points are drawn as green circles).
\end{enumerate}
Using \Cref{obs:P_n_and_reverse} and \Cref{lm:covering-basic_new}, we show that the consistency points fulfill their purpose.
\begin{lemma}\label{lem:2D_consistency_new}
    For any array $A[-n\dots n]$ with entries in $\{-n^2,\dots, n^2\}$, let $(z,a)$, $(z,b)$, $(z,c)$ be three unit edges such that the angle between each pair is $2\pi/3$. Define 
    $\widetilde{\mathcal{A}}_1=\pset_n(A, a, \overrightarrow{az},1)$, $\widetilde{\mathcal{A}}_2=\pset_n(A, b, \overrightarrow{zb},1)$, and $\widetilde{c}=c-\eps\overrightarrow{zc}$.
    \begin{enumerate}
        \item For any $i, j\in I_n$ with $i\leq j$ it holds that there does not exist a radius~$r$ disk containing $\widetilde{\mathcal{A}}_1[i]$, $\widetilde{\mathcal{A}}_2[j]$, and $\widetilde{c}$. 
        \item For any $i\in I_n$, it holds that if $i$ is odd, then there does not exist a radius~$r$ disk containing  $\widetilde{\mathcal{A}}_1[i+1]$, $\widetilde{\mathcal{A}}_2[i]$, and $\widetilde{c}$. 
        \item For any $i\in I_n$, it holds that if $i$ is even, then there exists a radius~$r$ disk containing $\widetilde{\mathcal{A}}_1[>i]$, $\widetilde{\mathcal{A}}_2[\leq i]$, and $\widetilde{c}$ and having its center within distance $3\eps (|i|+1)$ of $z$.
    \end{enumerate}
\end{lemma}
\begin{proof}
    Let $\widetilde{\mathcal{A}}:=\pset_n(A', a, \overrightarrow{za}, 1)$, $\widetilde{B}:=\widetilde{\mathcal{A}}_2$, $\widetilde{\mathcal{C}}:=\pset_n(C, c, \overrightarrow{zc}, 1)$, where $A'=A'[-n, \dots, n]$ with $A'[i]=-A[-i]$ and $C=C[-n, \dots, n]$ with $C[i]=0$ for all $i\in I_n$.
    Observe that $\widetilde{c}=\widetilde{C}[0]$ and $\widetilde{\mathcal{A}}$ is the symmetric version of $\widetilde{\mathcal{A}}_1$, i.e., $\widetilde{\mathcal{A}}[i]=\widetilde{\mathcal{A}}_1[-i+1]$ for all $i\in I_n$ (see \Cref{obs:P_n_and_reverse}). 
    Applying \Cref{lm:covering-basic_new} on $\widetilde{\mathcal{A}}$, $\widetilde{\mathcal{B}}$, $\widetilde{\mathcal{C}}$ and $k=0$ yields the claims.
\end{proof}

The next lemma shows that there exists a \tencenter solution of radius \radius in the Yes-Case of Gap Convolution-3SUM and no solution to \tencenter of radius \radius in the No-Case of Gap Convolution-3SUM. 

\begin{lemma}\label{lem:3SUMquadhardness}
    Assuming the 3-SUM hypothesis, there is no $O(n^{2-\eps})$ time algorithm for the Euclidean 10-center problem in $\real^2$ for any constant $\eps>0$.
\end{lemma}
\begin{figure}[htb]
    \centering
     \begin{subfigure}[b]{0.45\textwidth}
        \centering
        \includegraphics[page=2]{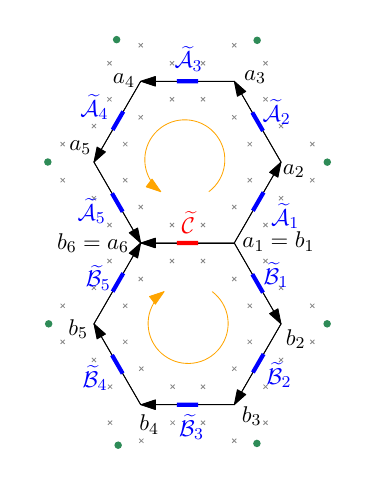}
        \subcaption[]{}
        \label{subfig:2hexagonspointset}
     \end{subfigure}
    \begin{subfigure}[b]{0.45\textwidth}
        \centering
        \includegraphics[page=3]{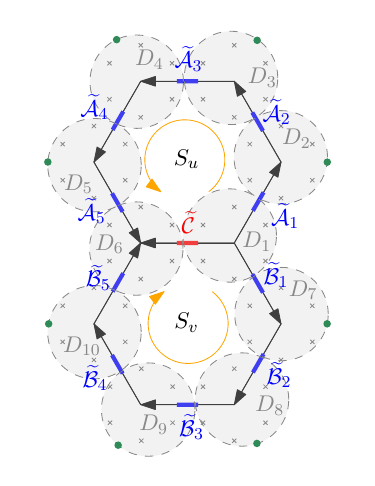}
        \subcaption[]{}
        \label{subfig:2hexagonscovered}
     \end{subfigure}
    \caption{Illustration of the point set constructed along the edges and around the vertices of two hexagons glued along an edge, illustrating the proof of \Cref{lem:3SUMquadhardness}. The points on the hexagonal edges are constructed from a single array $A$.}
    \label{fig:2hexagonsclustering}
\end{figure}

\begin{proof}
The proof proceeds by investigating the point set $\mathcal{P}$ constructed from the points of all types around vertices of two regular hexagons sharing an edge, as illustrated in \Cref{fig:2dsimplified}. We show that we can solve the Gap Convolution-3SUM problem if and only if we can decide whether there exist a solution to Euclidean 10-center problem of radius $r=\radius$ for $\mathcal{P}$. Each of the point gadgets on an edge of the two glued hexagons 
is constructed from the array $X$ of the Gap Convolution-3SUM, giving rise to the gadget point sets $\widetilde{\mathcal{A}}_1,\dots, \widetilde{\mathcal{A}}_5,\widetilde{\mathcal{C}},\widetilde{\mathcal{B}}_1,\dots,\widetilde{\mathcal{B}}_5$, whose precise definition and location with respect to the hexagonal vertices $a_l$ and $b_l$ is depicted in \Cref{fig:2hexagonsclustering}. 

Suppose first that $X$ represents a YES instance of Gap Convolution-3SUM, that is, there exists $i,j,k\in\{-\lfloor \frac{n}{100}\rfloor,\dots,\lfloor \frac{n}{100}\rfloor\}$ such that $i+j+k=0$ and $X[i]+X[j]+X[k]=0$. 
 By \Cref{lem:shared_edge}, there exists two radius $r$ disks $D_1$ and $D_6$ covering $\widetilde{\mathcal{A}}_5[> 2i]\cup \widetilde{\mathcal{B}}_5[> 2j] \cup\widetilde{\mathcal{C}}\cup \widetilde{\mathcal{A}}_1[\le 2i]\cup \widetilde{\mathcal{B}}_1[\le 2j]$. The center of $D_1$ (resp. $D_6$) is within distance $3\eps(\max(|2i|,|2j|)+1)\le 0.1n\eps$ of $a_1$ (resp. $a_6$), so covers all anchor points around $a_1$ (resp. $a_6$), by \Cref{lem:K-cen_anchor_unique_solution}.
There also exists a disk $D_l$ ($l\in\{2,\dots,5\}$) that covers $\widetilde{\mathcal{A}}_l[\le 2i]\cup \widetilde{\mathcal{A}}_{l-1}[> 2i]$ and the consistency point associated to $a_l$ as well as the anchor points around $a_l$, by \Cref{lem:2D_consistency_new}. The same arguments hold for the points in $\widetilde{\mathcal{B}}_l[\le 2i]\cup \widetilde{\mathcal{B}}_{l-1}[> 2i]$, so the entire point set can be covered by $10$ disks of radius $r$. 

Assume now that there exists a covering of the point set by $10$ disks. We will show that there exists $i,j,k\in\{-n,\dots,n\}$  such that $i+j+k=0$ and $X[i]+X[j]+X[k]=0$, i.e., $X$ is not a NO instance of Gap Convolution-3SUM. First note that the anchor points around a vertex $a_l$ force the center of each disk containing them to lie within distance $0.3n\eps$ to $a_l$, by \Cref{lem:K-cen_anchor_unique_solution}. By construction, for each gadget point set $\widetilde{\mathcal{A}}_l$, there exists a point at distance $>0.3n\eps+r$ from $a_l$, so none of the gadget point sets $\widetilde{\mathcal{A}}_1,\dots, \widetilde{\mathcal{A}}_5,\widetilde{\mathcal{C}},\widetilde{\mathcal{B}}_1,\dots,\widetilde{\mathcal{B}}_5$ can be covered by a single disk of radius $r$. This means that we can now chase the inequalities of the indices of the arrays being covered. We first group the indices covered by a single disk together, so that as we follow along the orientation of each hexagon, each disk is associated to a hexagonal vertex and covers a part of each of the gadget point sets associated to edges incident to that vertex. By \Cref{lem:2D_consistency_new}, following the orientation of the hexagonal edges along the degree two vertices, the maximal indices covered can only decrease, that is, if a disk $D$ covers both $\mathcal{A}_l[j]$ and $\mathcal{A}_{l+1}[i]$, then $i\le j$. Now, at vertex $a_6$, the two paths (starting from $a_1$) along different hexagons rejoin, which enforces these inequalities concerning the indices to be equalities (since otherwise the entire point set cannot be covered using the 10 disks), by \Cref{lm:covering-basic_new}, \Cref{lm:D4-covering} and \Cref{lem:shared_edge}. The indices further have to be even, and we obtain the inequalities regarding the sum of the entries of $X$ from the restrictions imposed by vertices $a_1$ and $a_6$ in both directions ($0\le X[i/2]+X[j/2]+X[k/2]\le 0$), concluding the proof.  
\end{proof}

\section{\texorpdfstring{$k$}{k}-Center in $\real^d$}
\label{sec: parameterized results}


In this section, we provide the proof of Theorems~\ref{thm:k-center}~and~\ref{thm:k-center-approx}. In Section~\ref{sec:binsumset}, we introduce the starting problem of our reduction for which we have a lower bound under ETH. Then, in Section~\ref{subsec:kcenter-2dim}, we prove Theorem~\ref{thm:k-center} on the plane. Finally, in Section~\ref{subsec:kcenter-Ddim}, we prove Theorem~\ref{thm:k-center} for all fixed dimensions.

\subsection{A Hard Problem: Binary SumSet on Grid Graphs}\label{sec:binsumset}

We first describe the \EMPH{binary SumSet} problem, a variant of the \EMPH{binary CSP} problem, as the ETH-hard problem to reduce from, to tackle \kcenter for general $k$.           

An instance of the \EMPH{binary CSP} problem on a grid graph is a tuple $I = (G, D, \mathcal{C})$ where
\begin{itemize}
    \item $G = (V,E)$ is an induced subgraph of the full $d$-dimensional grid graph consisting of vertices at points with integer coordinates and edges between vertices that are a distance of 1 apart. The vertex set $V$ represents a set of $|V|$ variables, one variable per vertex. 
    \item  $D$ is the set of domain values; each variable can be assigned any value from $D$. 
    \item  $\mathcal{C}=\{c_1,c_2,\dots c_q\}$ is a set of binary constraints associated to edges of $G$. Each $c_i\in\mathcal{C}$ is a pair $\{(u,v),R_{uv}\}$ associated to the edge $(u,v)$ and $R_{uv}\subset D\times D$ is a binary relation over $D$, specifying the admissible values for the pair of variables $(u,v)$.    
\end{itemize}

A \EMPH{solution} of $I$ is an assignment of the variables in $V$ where each variable $v$ is assigned a value $v^*$ in~$D$ and for every edge $(u,v)$ we have that $(u^*,v^*)\in R_{uv}$.  Marx and Sidiropoulos~\cite{MS14} showed the following hardness result.

\begin{theorem}[Marx and Sidiropoulos~{\cite[Theorem~2.6]{MS14}}]\label{thm:binary-CSP} Assuming ETH, for every $d\geq 2$, there is no algorithm with running time $f(k)n^{o(k^{1-1/d})}$ for binary CSP on grid graphs with $k$ variables and domain of size~$n$ for any computable function $f$. 
\end{theorem}

Our reduction is based on a variant of binary CSP on grid graphs, where each constraint is a sum constraint. More formally, an instance of the \EMPH{binary SumSet} problem on a grid graph is a tuple $J = (G , D, \mathcal{D})$ where:
\begin{itemize}
    \item $G= (V,E)$ is an induced subgraph of the $d$-dimensional grid graph as above, each of the $k$ vertices representing a variable. 
    \item  $D = [n^2]$ is the domain.
    \item  $\mathcal{D}$ is a family of subsets of $D$, containing a set $D_v$ for each vertex $v\in V$, and a set $D_{e}$ associated to each edge $e\in E$. Specifically, we 2-color the graph $G$ such that a vertex $v$ gets assigned the color $c(v)\in \{0,1\}$. For vertices $u$ and $v$, $D_v$ is some subset of $[n]\cdot n^{c(v)}$, and $D_{uv}$ is a subset of $\{a+b~|~(a,b)\in D_v\times D_u\}\subset [n^2]$. 
\end{itemize}

A \EMPH{solution} of $J$ is an assignment of the variables in $V$ such that each vertex~$v$ is assigned a value $v^* \in D_v$ and for every edge $e=(u,v)$, it holds that $u^*+v^* \in D_{e}$. The 2-coloring in the definition guarantees that if in a solution $u_1^*+v_1^* =u_2^*+v_2^*\in D_e$, then $u_1^*=u_2^*$ and $v_1^*=v_2^*$. 
We provide a fine-grained reduction from binary CSP to the binary SumSet problem. 


\begin{lemma}\label{lm:binary-sumset} Assuming ETH, for every $d\geq 2$, there is no algorithm with running time $f(k)\cdot n^{o(k^{1-1/d})}$ for binary SumSet on grid graphs with $k$ variables and domain of size~$n$ for any computable function $f$. 
\end{lemma}
\begin{proof} We show a reduction from binary CSP on grid graphs to the binary SumSet problem. Let $I = (G = (V,E), D, \mathcal{C})$ be an instance of binary CSP with $k$ variables and domain of size $n$. We construct a family $\mathcal{D}$ of subsets of $[N]$ with $N=n^2$ as in the definition of the binary SumSet problem.

First, we 2-color the grid graph $G$ such that the vertex $v$ gets assigned the color $c(v)\in \{0,1\}$. For vertices $u$ and $v$ of $G$, consider $D_v:=[n]\cdot n^{c(v)}$ and $D_{uv}:=\{a\cdot n^{c(v)}+b\cdot n^{c(u)}~|~(a,b)\in R_{uv}\}$. Since neighboring vertices have different colors, each pair $(a,b)\in R_{uv}$ is associated to a unique number in $D_{uv}$.   

Then $J = (G, [N], \mathcal{D})$ is an instance of binary SumSet on grid graphs with $k$ variables and domain of size $N =n^2$, where $\mathcal{D}$ is defined by all possible subsets of values for edges and vertices. It follows directly from the construction that $I$ is satisfiable if and only if $J$ is satisfiable. Thus, an algorithm with running time $f(k)\cdot N^{o(k^{1-1/d})} = f(k)\cdot n^{o(k^{1-1/d})}$ for $J$ would give an algorithm with the same running time for $I$, contradicting \Cref{thm:binary-CSP}.
\end{proof}

\subsection{The Case $d=2$} \label{subsec:kcenter-2dim}
In this section, we prove Theorem~\ref{thm:k-center} on the plane by reducing the two-dimensional binary SumSet problem to the \kcenter problem in the plane. We show that there exists a solution to \kcenter with radius \radius if and only if binary SumSet has a solution (recall that $\eps={(10\cdot n)^{-100}}$). The reduction is based on a \EMPH{basic curve} and an arrangement of congruent copies of the basic curve for each vertex of the binary SumSet instance on the plane.

We will use the machinery of \Cref{sec:10center} for the grid graph with $k$ vertices in $\real^2$. The basic curve (for each vertex) is no longer just the boundary of a hexagon, but the overarching idea will be the same. The schematic view of the reduction, for each edge, provided by \Cref{fig:2dsimplified} is still relevant, \Cref{subfig:2D_generalk_gadget} shows the curves replacing each hexagon, and \Cref{fig:d2kcenBasic} illustrates how the basic curves fit together.

Let $G = (V,E)$ be a subset of the 2D grid graph with $k$ vertices. Consider a set of four hexagonal grid cells, connected so that together they have exactly 10 neighboring grid cells and there exists exactly one horizontal hexagonal edge in the interior of the union of them (see \Cref{subfig:2D_generalk_gadget}).
Let $S$ be the outer boundary of the 10 neighboring grid cells. The \EMPH{basic curve} $S_{v}$ of each vertex $v$ of $G$ is a translation of $S$ that consists of hexagonal edges. Denoting the vertices of the grid graph $G$ by $a_{ij}$, such that the $i$ index corresponds to the coordinate of the vertex in the vertical direction, while $j$ specifies the horizontal coordinate. If $u=a_{ij}$, we also write $S_{ij}=S_u$. 

\begin{figure}[htb]
    \centering
    \begin{subfigure}[b]{0.45\textwidth}
        \centering
        \includegraphics[page=6]{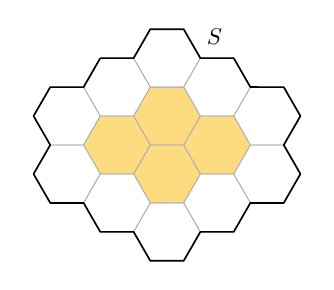}
     \end{subfigure}
    \begin{subfigure}[b]{0.45\textwidth}
        \centering
        \includegraphics[page=5]{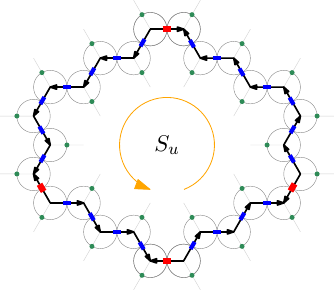}
     \end{subfigure}
     \caption{The basic curve is a translation of the curve $S$ on the left. On the right, the orientation is counterclockwise. The points of type (1) are drawn in blue, of type (2) in red, and of type (3) as green circles. 
     \label{subfig:2D_generalk_gadget}}
\end{figure}

We arrange the basic curves in such a way that the basic curves of two vertices $u$ and $v$ share an edge if and only if $(u,v)\in E$ (as in \Cref{fig:d2kcenBasic}). Moreover, two basic curves can share at most one edge. This shared edge is used to encode the constraint that for the assignment $u^*$ to $u$ and $v^*$ to $v$, it holds that $u^*+v^*\in D_{(u,v)}$.
We orient each basic curve in a consistent way, such that one can travel around the curve by following the orientations of nonshared edges (the orientation of shared edges may be inconsistent with the orientation of the other edges, as was also the case for \tencenter, see \Cref{fig:2hexagonsclustering}). Next, we glue basic curves corresponding to neighboring vertices of the grid graph along shared edges, where we only glue basic curves of opposite orientation. Specifically, for $a_{ij}$ such that $i+j$ is even, we orient the basic curve counterclockwise; otherwise, we orient it clockwise. In terms of the colors assigned in the binary SumSet instance, the orientation is reversed for basic curves associated to vertices of the grid graph of different colors. The fact that it is possible to assemble an arbitrary amount of these basic curves into a configuration that mimics the connectivity of a grid graph follows essentially from~\Cref{subfig:2D_arrangegadgets}. The figure shows that around a central basic curve, we can arrange 8 copies by removing edges appropriately from the hexagonal tiling. Furthermore, the arrangement exhibits translational symmetry, as indicated by the green arrows, which means that the basic curves do not overlap in arrangements of arbitrary size.         

\begin{figure}[htb]
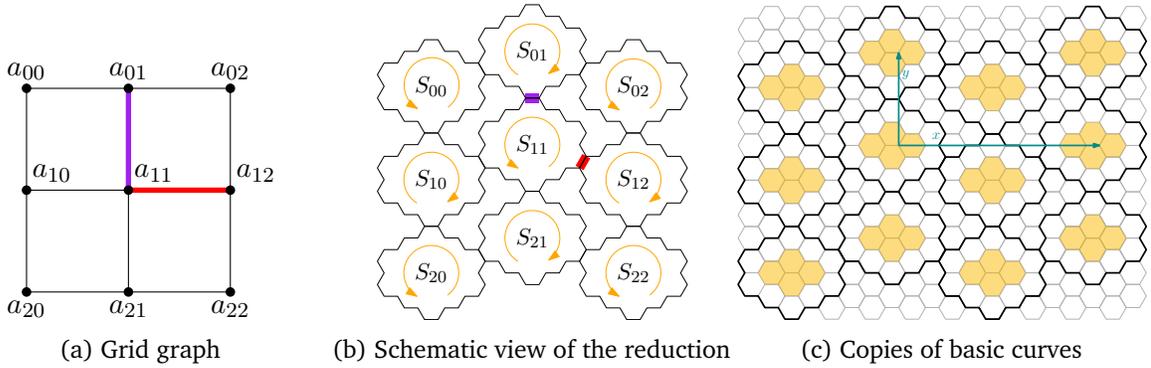

    \centering
    \begin{subfigure}[b]{0.3\textwidth}
    \vspace{0pt}
        \centering
        \includegraphics[page=2]{figs/2D_generalk_.pdf}
        \subcaption[]{Grid graph}
        \label{subfig:2D_generalk_grid_graph}
     \end{subfigure}%
    \begin{subfigure}[b]{0.33\textwidth}
    \vspace{0pt}
        \centering
        \includegraphics[page=4, height=4.2cm]{figs/2D_schematic_view.pdf}
        \subcaption[]{Schematic view of the reduction}
        \label{subfig:2D_generalk_full}
     \end{subfigure}%
     \begin{subfigure}[b]{0.33\textwidth}
     \vspace{0pt}
        \centering
        \includegraphics[width=\textwidth]
        {figs/hexagadgets.pdf}
        \subcaption[]{Copies of basic curves}
        \label{subfig:2D_arrangegadgets}
     \end{subfigure}
    \caption{(a) a grid graph in $\real^2$ with $k = 9$ vertices; (b) A schematic view of the reduction; the shared edges marked in (a) are highlighted. Each arrow denotes the orientation of the edges of the basic curve. (c) Copies of nonoverlapping basic curves with the same incidence structure as the full grid graph. Translational symmetries are indicated by the green arrows.}
    \label{fig:d2kcenBasic}
\end{figure}

Let $\mathcal{H}$ be the set containing the hexagonal vertices of a basic curve.  
Then, the cardinality $k'$ of $\mathcal{H}$ is in $O(k)$, i.e., linear in the number $k$ of vertices of the grid graph of the binary SumSet instance. Define~$\mathcal{A}$ to be the set of anchor points defined by the points in~$\mathcal{H}$ as in \Cref{sec:10center}. Then, by \Cref{lem:K-cen_anchor_unique_solution}, for every $k'$-center solution to $\mathcal{A}$ of radius \radius the centers lie close to the points in $\mathcal{H}$.

In the following, we add more points to this instance and show that a $k'$-center solution of radius \radius to this instance exists if and only if the corresponding binary SumSet instance has a solution.

\paragraph{Array construction.}
We construct an array for each set $D_v$ and $D_e$ of the binary SumSet instance. Then, we will use these arrays to construct point sets similar to the 10-center construction. The definition allows us to reuse the machinery developed in~\Cref{sec:10center} for the \tencenter problem (\Cref{lem:shared_edge} and \Cref{lem:2D_consistency_new}).

For each set $D_v$, where $v$ is a vertex of the grid graph of the binary SumSet instance, we define an array $X_v=X_v[(-100n^2) \dots (100n^2)]$
\begin{equation} \label{def:Xv}
    X_v[k]=\begin{cases}
    1& \text{if } k\in D_v \\
    2 & \text{otherwise}.
    \end{cases}
\end{equation}
Further for an edge $e$, we define the array $X_e=X_e[(-100n^2) \dots (100n^2)]$ by
\begin{equation}\label{def:Xe}
    X_e[k]=\begin{cases}
    -2& \text{if } -k\in D_e\\
    0& \text{otherwise}.
    \end{cases}
\end{equation}

By construction, we have the following observation. 
\begin{observation}\label{obs:defX_i}
    Let $v$ and $w$ be adjacent vertices. Then for any $k_1,k_2,k_3$, it holds that $X_v[k_1]+X_w[k_2]+X_{(v,w)}[k_3]\leq 0$ if and only if $k_1\in D_v$, $k_2\in D_w$, and $k_3\in D_{(v,w)}$.
\end{observation}

As in \Cref{sec:10center}, for every non-shared edge $(a,b)$ of a basic curve $S_u$, we add the point set $\PSEdge{a}{b}$ (points of type (1)) and for every shared edge $(a,b)$ we add the point set $\PSEdge{b}{a}$ (points of type (2)), and for every hexagonal vertex~$b$ that is incident to two non-shared edges, we add the consistency point~$c_b$ (points of type (3)). 


\begin{observation}\label{obs:possibleCenters}
    Let $p$ be a point of type (1), (2), or (3) on the hexagonal edge from $a$ to $b$ in a basic curve. Then, in any $k'$-center solution of radius \radius, $p$ is contained in a disk that is centered at a point with distance at most $0.3n\eps$ to $a$ or to $b$. 
\end{observation}
\begin{proof}
    The distance from $p$ to any point in $\mathcal{H}\setminus \{a, b\}$ is at least $\sqrt{3}$. Therefore, the claim follows by \Cref{lem:K-cen_anchor_unique_solution}.
\end{proof}

\begin{theorem}\label{thm:planarsumset}
    Assuming ETH, there is no $f(k)n^{o(\sqrt{k})}$ time algorithm for the \kcenter problem in $\mathbb{R}^2$.
\end{theorem}

\begin{proof}
    For a $2$-dimensional binary SumSet instance with $k$ variables and domain size $n$, the construction of the arrangement of basic curves outlined above constitutes a $k'$-center instance in~$\mathbb{R}^2$ with $n'$ points where $k'=O(k)$ and $n'=O(n^2)$ ($[n^2]$ is the domain of the binary SumSet instance). 
    We show that there exists a clustering of radius $r=\radius$ if and only if the binary SumSet instance has a solution.
    
    Assume there is a solution to binary SumSet and for any vertex~$v$ let $v^*\in D_v$ be the assignment in a feasible solution. Then, it holds that $u^*+v^*\in D_{(u,v)}$ for any edge $(u,v)$. By \Cref{lem:shared_edge}, \Cref{lem:2D_consistency_new}, and \Cref{obs:defX_i}, it follows that we can cover all points with $k'$ disks of radius $r$ as follows. 
    Let $P(a,b)$ be a point set of type (1) defined along a non-shared edge of a basic curve associated to a vertex $u$ of the grid graph. Then the points in $P(a,b)[\le  2u^*]$ are covered by disks of radius $r$ within distance $0.1n\epsilon$ from $a$, and $P(a,b)[>2u^*]$ is covered by a similar disk close to $b$. For every point set $P(a,b)$ of type (2) (lying on a shared edge of two basic curves) defined by $(u,v)$, the points in $P(a,b)[\le 2(u^*+v^*)]$ are covered by the disk with center close to $a$; $P(a,b)[>2(u^*+v^*)]$ is covered by the disk roughly centered at $b$. Since the disks are sufficiently close to vertices (points in $\mathcal{H}$) of the basic curve, by \Cref{lem:2D_consistency_new}, the consistency points (points of type (3)) and the anchor points are also covered by the collection of $k'$ disks. Along the shared edges of basic curves, the coloring in the definition of the binary SumSet instance ensures global consistency: the index values used for neighboring basic curves are identical, ensuring that the clusterings constructed in this way for neighboring basic curves is compatible for neighboring basic curves.        
    
    The other direction follows similarly to the proof of \Cref{lem:3SUMquadhardness}, in conjunction with~\Cref{obs:defX_i} (which gives the connection of the binary SumSet instance to arrays and thus the constructed point sets), which allows us to rely on our established results restricting the location of centers of disks involved in the clustering. The idea is that within each basic curve gadget, the existence of a clustering enforces inequalities for indices (and array values) as one travels along the edges of the basic curve. Every disk in the solution to the $k'$-center instance covers points of types (1) and (2), and the maximum values of the indices (along the orientation of the basic curve) that are covering correspond to index restrictions of the original array constructed from the binary SumSet instance. The requirement that every point must be covered along with the rigid inequalities in \Cref{lm:covering-basic_new}, \Cref{lem:2D_consistency_new}, and \Cref{lem:shared_edge} actually forces the collection of these inequalities to be equalities, similarly to \Cref{lem:3SUMquadhardness}, where the basic curve consisted of the edges of a single hexagon. At a shared vertex of two neighboring basic curves, the inequalities add up to yield the equality constraint relating the sum of the covered indices and their associated array entries for $X_{u}$ and $X_{(u,v)}$ constructed from the SumSet instance. By definition of the binary SumSet problem, the propagation of the indices along shared edges cannot mix index values for neighboring basic curves, since the coloring implies that the indices of two neighboring basic curves are multiples of different powers of $n$. As a result, the propagation of the index values along hexagonal edges using the clustering is consistent within each basic curve, so the inequalities at vertices of shared edges combine to give restrictions based on the incoming and outgoing edges, and these relate to the same indices and arrays.
\end{proof}

\begin{remark}\label{rem:w1}
It is worth noting that in the proof of Theorem~\ref{thm:planarsumset}, the reduction starts from a $2$-dimensional binary SumSet instance with $k$ variables and domain size $n^2$ and ends with a $k'$-center instance in~$\mathbb{R}^2$ with $n'$ points where $k'=O(k)$ and $n'=O(n)$. Since it is clear from the proof of Theorem 2.6 in \cite{MS14} that the $2$-dimensional binary SumSet problem is W[1]-hard, we also have that \kcenter in 2D is W[1]-hard.
\end{remark}

\subsection{The Case $d\geq 3$}\label{subsec:kcenter-Ddim}

In this section, we prove Theorems~\ref{thm:k-center} and \ref{thm:k-center-approx}. First, we describe how the reduction in Section~\ref{subsec:kcenter-2dim} can be generalized to higher dimensions and then towards the end of this section, we provide the proofs.

We consider yet another reinterpretation of the basic curve in the final construction (\Cref{subfig:2D_generalk_gadget}) in $d=2$. The edges shared between copies of vertically stacked copies of basic curves are associated to `ears' of the curve (portions of the basic curves separated by an \emph{ear edge}, a horizontal edge between vertices that cuts a hexagon exactly ìn half). Both \Cref{fig:d3kcenBasic}(left) and (right) show horizontally drawn ear edges (in pink) along with rotations used to connect basic curves along different linearly independent directions. Consider modifying the basic curves to lie instead as a curve in $\real^3$, obtained by folding the two ears, by $\pi/2$, along the axis formed by the ear edge, according to the same orientation. The folded ears lie in planes that are orthogonal to the plane containing the rest of the basic curve. The resulting basic curves would fit together similarly to the planar case, and we observe that the calculations and reasoning of the previous section also work with this (slightly more complicated) basic curve (we need to move the two consistency points that attach only to vertices of an ear using the same transformation as for the ear, and similarly for the anchor points around these vertices).         

This simple observation means that we can account for further dimensions of the grid graph by simply adding further ears next to the two existing ones, as illustrated in~\Cref{fig:d3kcenBasic}(right).   


\begin{figure}[htb]
    \centering
    \begin{subfigure}[b]{0.3\textwidth}
        \includegraphics[page=1]{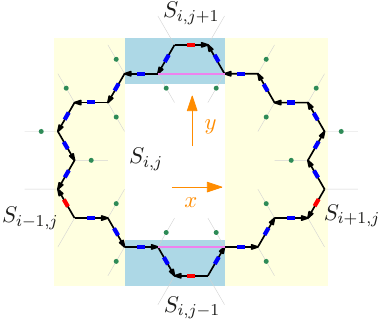}
    \end{subfigure}
    \ \ 
    \begin{subfigure}[b]{0.68\textwidth}
        \raggedleft
        \includegraphics[page=2]{figs/dD_schematic_view.pdf}
    \end{subfigure}
    \caption{(left) The basic curve for the case $d=2$, with the ears in the $y$ direction highlighted in blue. (right) The basic curve for $d=3$: another hexagonal edge and ear are added, in a plane perpendicular to the rest of the basic curve, for the third dimension.} 
    \label{fig:d3kcenBasic}     
\end{figure}

\begin{figure}
    \centering
    \includegraphics[page=4]{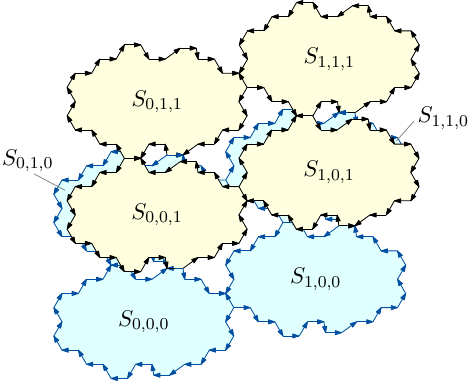}
    \caption{Schematic view of the gadget for a cube.}
    \label{fig:dD_arrangement}
\end{figure}

\paragraph{The basic curve.} Consider the basic curve of~\Cref{subfig:2D_generalk_gadget} without the top and bottom ears, so that the bottom and top edges of the resulting structure each lie on a line and consist of one ear and two hexagonal edges. We extend the top and bottom parts by inserting $(d-2)$ horizontal ear and $d-2$ horizontal hexagonal edges to each of these parts, in such a way that the resulting sequence of edges alternates between the two edge types. We now consider pairs of ears that connect to pairs of ear edges that lie directly on top of each other, as each of the blue and green colored pairs of ears in~\Cref{subfig:2D_generalk_gadget}. We then attach $d-2$ such pairs of ears such that each lies in the plane spanned by a vector in the horizontal direction and another canonical basis vector in $\real^d$, giving each pair an associated dimension. We call the resulting curve a \EMPH{$d$-dimensional basic curve}. \Cref{fig:d3kcenBasic}(right) illustrates the case $d=3$, and also how the construction naturally generalizes to higher dimensions.

\paragraph{The reduction.} Let $G$ be the grid graph with $k$ vertices in $\real^d$. The idea of the construction is still the same as in the planar case. For each vertex $v$ of $G$, we construct one $d$-dimensional basic curve. We arrange the curves so they connect along hexagonal edges on ears associated to the dimension $i$ iff their associated vertices in $G$ are neighboring in dimension $i$. See \Cref{fig:dD_arrangement} for a simple example illustrating the arrangement of a small number of copies in $\real^3$.  

Next, we orient the edges of each basic curve so that the orientation of every shared hexagonal edge is consistent. This can be done since $G$ is 2-colorable, and therefore, we could simply orient the basic curves in one color class clockwise (viewed in the plane of the basic curve without the ears), and orient the basic curves in the other color class counterclockwise.  

The reduction now works in essentially the same way as in the planar case. Points of all types are placed along the $d$-dimensional basic curve as before, where the consistency points on hexagonal edges incident to an ear (on both sides) lie in the plane of the attached ear. The same reasoning as for the planar case then yields, by~\Cref{lm:binary-sumset}, the hardness result for all dimensions.   

We are now ready to prove Theorem~\ref{thm:k-center}.

\paragraph{Proof of Theorem~\ref{thm:k-center}:}
The idea of the proof is the same as in the $2$-dimensional case (\Cref{subsec:kcenter-2dim}), namely, we exhibit a reduction from binary SumSet and invoke \Cref{lm:binary-sumset}. Given a $d$-dimensional binary SumSet instance with $k$ variables and domain size $n^2$, the first step is to convert $G$, the graph of the binary SumSet instance, into an arrangement of basic curves in $\mathbb{R}^d$. Every vertex $v$ of $G$ is replaced by a $d$-dimensional basic curve $S_v$, using the construction discussed above and illustrated for $d=3$ in \Cref{fig:d3kcenBasic}. Since $G$ is an induced subgraph of the grid graph, the $d$-dimensional basic curves are connected in the same way as the vertices of $G$. Each basic curve is assigned an orientation as before, so that at every vertex of a shared edge, either all incident edges are incoming, or outgoing (as in \Cref{fig:dD_arrangement}). 

We add anchor points around every hexagonal vertex as in the 2-dimensional setting (see \Cref{fig:K-cen_anchor_unique_solution} or \Cref{fig:2dsimplified}), in the plane spanned by two linear independent incident edges. 

Next, we convert the sets $D_{v}$ and $D_{(u,v)}$ of the binary SumSet instance into arrays $X_v$ and $X_{uv}$ according to definitions \ref{def:Xv} and \ref{def:Xe}. These arrays are used to define point sets in $\mathbb{R}^d$ (using Definition~\ref{eq:Pn-def}) along the edges of the $d$-dimensional basic curves.
Inside each basic curve $S_v$, points of type (1) associated to $X_v$ are added along the non-shared edges, and those of type (2) are placed on shared edges, corresponding to an edge $(u,v)$ in the grid graph (and the points are defined using $X_{(u,v)}$). Consistency points (points of type (3)) are added to the basic curve as in the two-dimensional case, only that the consistency points on hexagonal edges incident to an ear (on either side) lie in the plane of the attached ear. 

The resulting point set $P$, consisting of the anchor points and points of all three types, represents a $k'$-center instance in~$\mathbb{R}^d$ with $n'$ points. For a fixed dimension $d$, there are only a constant number of hexagonal vertices along a $d$-dimensional basic curve, and around each we place a constant number of points, so we have $k'=O(k)$ and $n'=O(n^2)$, and the construction takes time $O(n^2)$ per vertex of $G$. 
As before, the idea is to show that there exists a clustering of radius \radius if and only if the binary SumSet instance has a solution. As previously observed, the computations and results of the two-dimensional case carry over to the $d$-dimensional basic curve, so this is a straightforward adaptation of the proof of \Cref{thm:planarsumset}.     
\hfill$\square$

By simply observing the parameter setting in our proof of Theorem~\ref{thm:k-center}, we obtain the proof of Theorem~\ref{thm:k-center-approx}.

\paragraph{Proof of Theorem~\ref{thm:k-center-approx}}
Suppose there was an algorithm $\mathcal{A}$, that for some $d_0\ge 2$ and computable function $f$, and for every $n,k,\delta$,  took as input instances of the \kcenter problem in $\mathbb{R}^{d_0}$ on $n$ points, ran in time $f(k)\cdot (k/\delta)^{o(k^{1-1/d_0})}\cdot n^{O(1)}$, and output  a $(1+\delta)$-factor approximation  of the optimal \kcenter cost. From the proof of Theorem~\ref{thm:k-center}, we have that no algorithm running in time $g(k)\cdot n^{o(k^{1-1/d})}$, for any computable function $g$ that takes as input a \kcenter instance on $n$ points in $\mathbb{R}^d$ can distinguish if the optimal radius is at most $\radius$ or if the optimal radius is more than $1-\varepsilon+2\varepsilon^{1.7}$, where $\varepsilon={(10\cdot n)^{-100}}$ (see Remark~\ref{rem:covering_basic_2eps1.7}). We now show that existence of $\mathcal{A}$ contradicts Theorem~\ref{thm:k-center}.

Let $P$ be a \kcenter instance on $n$ points in $\mathbb{R}^{d_0}$. We simply run $\mathcal{A}$ on $P$, and in time $f(k)\cdot (k/\delta)^{o(k^{1-1/{d_0}})}\cdot n^{C}$ (for some universal constant $C$), where $\delta:=\frac{\varepsilon^{1.7}}{\radius}$, and if the optimal radius was at most $\radius$, then the output of $\mathcal{A}$ is at most $1-\varepsilon+2\varepsilon^{1.7}$. On the other hand, if the optimal radius is more than $1-\varepsilon+2\varepsilon^{1.7}$, the the output of $\mathcal{A}$ is more than  $1-\varepsilon+2\varepsilon^{1.7}$. Thus, we can decide if the optimal radius is at most $\radius$ or if the optimal radius is more than $1-\varepsilon+2\varepsilon^{1.7}$. The total running time is thus $\left(f(k)\cdot k^{o(k^{1-1/{d_0}})}\right)\cdot  \left(\frac{\radius}{\varepsilon^{1.7}}\right)^{o(k^{1-1/{d_0}})}\cdot n^{C}=g(k)\cdot \left((10n)^{170}-(10n)^{70}+1\right)^{o(k^{1-1/{d_0}})}\cdot n^C$, where \mbox{$g(k):=\left(f(k)\cdot k^{o(k^{1-1/{d_0}})}\right)$}, and this contradicts Theorem~\ref{thm:k-center} as $\left((10n)^{170}-(10n)^{70}+1\right)^{o(k^{1-1/{d_0}})}\cdot n^C = n^{o(k^{1-1/d_0})}$.
\hfill$\square$


\begin{remark}\label{rem:exp}
It is worth noting a consequence of our reduction when applying the hardness from Theorem~\ref{thm:binary-CSP} \cite{MS14} in the regime where the domain size $n$ is a universal constant. In this case, assuming ETH, there is no algorithm for binary CSP on grid graphs with $K$ variables and constant domain size running in time $2^{o(K^{1-1/d})}$ for any $d\geq 2$.

Our reduction in the proof of Theorem~\ref{thm:k-center} constructs from such a CSP instance, a \kcenter instance in $\mathbb{R}^d$ with $N=O(K)$ points (since $n$ is constant) and $k=\Theta(K)$ centers. The reduction establishes a gap, making it hard to distinguish if the optimal radius is at most $\radius$ or greater than $1-\varepsilon+2\varepsilon^{1.7}$, where $\varepsilon={(10\cdot n)^{-100}}$ (see Remark~\ref{rem:covering_basic_2eps1.7} for the soundness case justification).

Because $n$ is constant, $\varepsilon$ is also a constant. Therefore, the optimal radii in the completeness (at most $\radius$) and soundness (more than $1-\varepsilon+2\varepsilon^{1.7}$) cases are two distinct constants. This implies a constant-factor approximation gap, $(1+\varepsilon_0)$, for some constant $\varepsilon_0 > 0$.

A $(1+\varepsilon_0)$-approximation algorithm for the $d$-dimensional $k$-center problem on $N$ points running in time $2^{o(k^{1-1/d})} \cdot N^{O(1)}$ (where $k=\Theta(N)$) would solve the underlying binary CSP instance in $2^{o(K^{1-1/d})}$ time, violating ETH. Thus, we have proved that for some constant $\varepsilon_0>0$, no $(1+\varepsilon_0)$-approximation for the $d$-dimensional $k$-center problem on $N$ points is possible in time $2^{o({k^{1-1/d}})}\cdot N^{O(1)}$ (where $k=\Theta(N)$), unless ETH fails.
\end{remark}



\section{\texorpdfstring{$2$}{2}-Center in 3D}
\label{sec: 2 center}

In this section, we prove \Cref{thm:2-center} by showing a reduction from the Gap Convolution-3SUM problem. In Section~\ref{sec: 2 center reduction}, we present the reduction, and in Sections~\ref{sec: 2 center completeness}~and~\ref{sec: 2 center soundness}, we analyze the completeness and soundness of the reduction, respectively. 

\subsection{Reduction}\label{sec: 2 center reduction}
In this subsection, we present a reduction from the Gap Convolution-3SUM problem to the \twocenter problem in $\mathbb{R}^3$.

Let $X[-n \dots n]$ be an instance of the Gap Convolution-3SUM problem where each array entry is bounded by $n^2$. For notational convenience, let $M = n^2$ denote the maximum absolute value in $X$, and fix $\eps = (10n)^{-100}$. Our reduction creates the following set of points in ${\mathbb R}^3$; see \Cref{fig:point-set3D}.

\begin{figure}[!htb]
\centering
\begin{tikzpicture}[scale=3,>=stealth]

\draw[->,thick] (0,0) -- (1,0) node[right] {$x$};
 \draw[thick] (0,0) -- (-0.5,0); 
\draw[->,thick] (0,0) -- (0,1) node[above] {$z$};
\draw[thick] (0,0) -- (-0.6,-0.3); 
\draw[dashed,->,thick] (0,0) -- (1.0,0.5) node[above right] {$y$}; 

\fill (0,0) circle (0.5pt) node[below right] {$O$};

\pgfmathsetmacro{\a}{1/sqrt(2)}

\coordinate (A) at (\a,0);
\coordinate (B) at ({0 + \a*0.4}, {0 + \a*0.2});
\coordinate (C) at ({-0.5 + (-0.5)*0.4}, {0 + (-0.5)*0.2});

\fill[blue] (A) circle (0.6pt) node[below right] {$(\frac{1}{\sqrt2},0,0)$};
\fill[red](B) circle (0.6pt);
\node[red] at ($(B)+(-0.36,0.05)$) {$(0,\frac{1}{\sqrt2},0)$};
\fill[green!70!black] (C) circle (0.6pt) node[below left] {$(-\frac12,-\frac12,0)$};

 \node[blue]at ($(A)+(0.1,0.25)$) {$\tilde{\mathcal{A}}$};
\node[red] at ($(B)+(0.1,0.2)$) {$\tilde{\mathcal{B}}$};
\node[green!70!black] at ($(C)+(0.08,-0.1)$) {$\tilde{\mathcal{C}}$};

\draw[->,orange] (A) -- ++(0,0.5) node[right] {$s_A$};
\draw[dashed,thick,orange] (A) -- ++(0,-0.5);
\foreach \i in {1,2,3,4,5}{\fill[blue] ($(A)+(0,0.06*\i)$) circle (0.4pt);}
\foreach \i in {1,2,3,4,5}{\fill[blue!50] ($(A)+(0,-0.06*\i)$) circle (0.4pt);}

\draw[->,orange] (B) -- ++(0,0.5) node[right] {$s_B$};
\draw[dashed,thick,orange] (B) -- ++(0,-0.5);
\foreach \i in {1,2,3,4,5}{\fill[red] ($(B)+(0,0.06*\i)$) circle (0.4pt);}
\foreach \i in {1,2,3,4,5}{\fill[red!50] ($(B)+(0,-0.06*\i)$) circle (0.4pt);}

\draw[->,orange] (C) -- ++(0,0.5)node[right] {$s_C$};
\draw[dashed,thick,orange] (C) -- ++(0,-0.5);
\foreach \i in {1,2,3,4,5}{\fill[green!70!black] ($(C)+(0,0.06*\i)$) circle (0.4pt);}
\foreach \i in {1,2,3,4,5}{\fill[green!40!black] ($(C)+(0,-0.06*\i)$) circle (0.4pt);}


\coordinate (d0) at (0,{1+\a});
\fill[purple] (d0) circle (0.8pt) node[below] {$d^+_0$};

\coordinate (d1) at (1,{\a});
\fill[purple] (d1) circle (0.8pt) node[right] {$d^+_1$};

\coordinate (d2) at ({0+1*0.4},{\a+1*0.2});
\fill[purple] (d2) circle (0.8pt) node[above right] {$d^+_2$};

\coordinate (d3) at (-1,{\a});
\fill[purple] (d3) circle (0.8pt) node[left] {$d^+_3$};

\coordinate (d4) at ({0-0.4},{\a-0.2});
\fill[purple] (d4) circle (0.8pt) node[below right] {$d^+_4$};


\coordinate (dm0) at (0,{-1-\a});
\fill[purple] (dm0) circle (0.8pt) node[above] {$d^-_0$};

\coordinate (dm1) at (1,{-\a});
\fill[purple] (dm1) circle (0.8pt) node[right] {$d^-_1$};

\coordinate (dm2) at ({0+1*0.4},{-\a+1*0.2});
\fill[purple] (dm2) circle (0.8pt) node[above right] {$d^-_2$};

\coordinate (dm3) at (-1,{-\a});
\fill[purple] (dm3) circle (0.8pt) node[left] {$d^-_3$};

\coordinate (dm4) at ({0-0.4},{-\a-0.2});
\fill[purple] (dm4) circle (0.8pt) node[below right] {$d^-_4$};

\end{tikzpicture}
\caption{The three point sets $\apoints,\bpoints,\cpoints$ and the anchor points $\mathcal{D}$.}
\label{fig:point-set3D}
\end{figure}
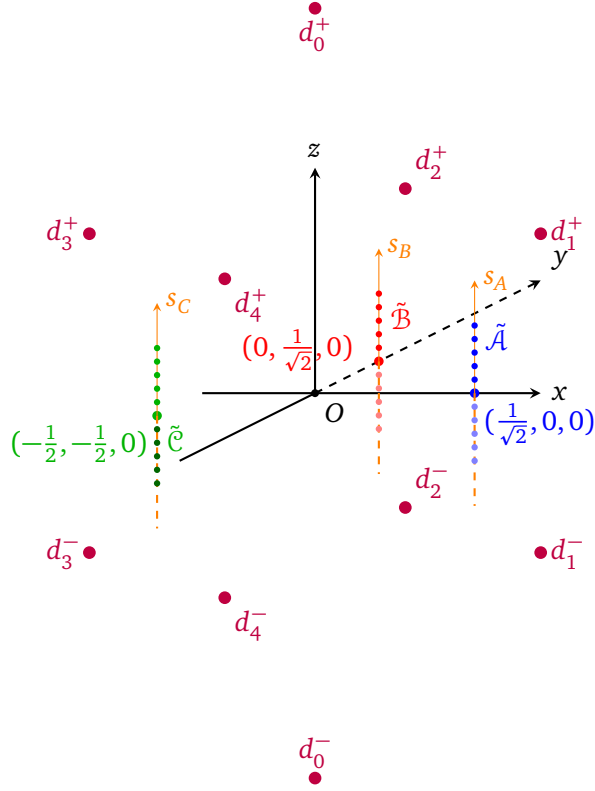

\paragraph{Canonical points:} These are the points whose coverage would tentatively correspond to a solution. There are three sets of points, built over the domain $I_{n} = \{-2n, \dots, 2n+1\}$: 

\begin{itemize}
    \item $\apoints = \pset_{n}(X, (\frac{1}{\sqrt{2}}, 0,0), (0,0,1), 1)$. 
    \item $\bpoints = \pset_{n}(X, (0, \frac{1}{\sqrt{2}},0), (0,0,1), 1)$. 
    \item $\cpoints = \pset_{n}(X, (-1/2,-1/2,0), (0,0,1/\sqrt{2}), \sqrt{2})$. 
\end{itemize}

Notice that all three point sets are ordered increasingly in the $z$-coordinate. Following the convention established in Section 4, we use $\apoints[\leq i]$ to denote the subset of points in $\apoints$ up to index $i$, and similarly for $[> i]$.

In an ideal situation, we would like the solution of $2$-center to be two (almost) unit balls, one (top ball) centered at around $(0,0,1/\sqrt{2})$ and the other (bottom ball) centered at around $(0,0,-1/\sqrt{2})$, together covering points in $\apoints\cup\bpoints\cup\cpoints$. Furthermore, the top ball will cover points in $\apoints[> 2i], \bpoints[> 2j], \cpoints[> 2k]$ while the bottom ball will cover points in $\apoints[\leq 2i], \bpoints[\leq 2j], \cpoints[\leq 2k]$ for even indices $2i, 2j, 2k \in I_{n}$ such that $i+j+k = 0$ and $X[i] + X[j] + X[k] = 0$. Therefore, a solution of $2$-center for the points will give us a solution of the Gap Convolution-3SUM problem. To get such precise control over the solution of the \twocenter, we add more anchor points as follows.

\paragraph{Anchor points:} Choose $t = 3 \lfloor n/100 \rfloor + 1/4$. The anchor points are at the extreme of the tentative two balls along the two opposite directions of each axis, consisting of the following: 
\begin{itemize}
    \item $d^+_0 = (0,0,1+\frac{1}{\sqrt{2}} + \eps)$
    \item $d^+_1 = (1-t \eps, 0, \frac{1}{\sqrt{2}}+\eps)$
    \item $d^+_2 = (0,1-t \eps,  \frac{1}{\sqrt{2}}+\eps)$
    \item $d^+_3 = (-1+t \eps, 0, \frac{1}{\sqrt{2}}+\eps)$
    \item $d^+_4 = (0,-1+t \eps,  \frac{1}{\sqrt{2}}+\eps)$
\end{itemize}

We symmetrically define $d^-_0, \ldots, d^-_4$ as $d^+_0,\ldots, d^+_4$ with negated $z$-coordinates.
The anchor points are defined as $\dpoints = \{d^+_0, \ldots, d^+_4, d^-_0, \ldots, d^-_4\}$. View $\dpoints$ as $\dpoints = \dset^+ \cup \dset^-$. This completes the description of our reduction. 

\begin{lemma}
The instance $X$ is a Yes-instance of Gap Convolution-3SUM if and only if there exist two balls of (squared) radii $1+ 3 M^2 \eps^2$ that cover all the points in $\apoints \cup \bpoints \cup \cpoints \cup \dpoints$.  
\end{lemma}

Using standard arguments, \Cref{thm:2-center} follows directly from the lemma. 

\paragraph{Proof of~\Cref{thm:2-center}:} Assume for contradiction that there is a subquadratic-time algorithm that solves the \twocenter problem in 3D. 
Given an instance $X$ guaranteed by \Cref{lm:gap-3sum-hardness}, we use our reduction to produce an instance of size $N = O(n)$ of the \twocenter problem. The presumed algorithm can be used to distinguish between the Yes- and No-instance of $2$-center in $O(N^{2-\gamma})$ time for some $\gamma >0$. Since $N= O(n)$, this algorithm also distinguishes between Yes- and No-instance of Gap Convolution-3SUM in time $O(n^{2-\gamma})$, contradicting the 3SUM hypothesis. \hfill $\square$

\subsection{Completeness}\label{sec: 2 center completeness}

In this subsection, we show that a solution of Gap Convolution-3SUM gives an optimal solution of \twocenter. For convenience, we define the following values, for integers $i, j, k \in \{-n, \dots, n\}$: 
\begin{eqnarray*}
x^*_i & = & -3i - X[i] \sqrt{\eps} \\
y^*_j & = & -3j - X[j] \sqrt{\eps} \\ 
z^*_k & = & -3k - X[k] \sqrt{\eps} 
\end{eqnarray*}
Notice that since a valid solution for Gap Convolution-3SUM always falls within the restricted bounds $[-\lfloor n/100 \rfloor, \lfloor n/100 \rfloor]$, a YES-instance restricts $x^*_i, y^*_j, z^*_k$ to be strictly bounded by $3\lfloor n/100 \rfloor + M\sqrt{\eps}$ for sufficiently large $n$. 

\begin{observation} 
\label{obs: relating x,y,z}
For all integers $i,j,k \in \{-n, \ldots, n\}$, we have 
$x^*_i + y^*_j + z^*_{k} = -3(i+j+k) - (X[i]+X[j]+X[k]) \sqrt{\eps}$. 
If $i+j+k=0$, this simplifies to $-(X[i]+X[j]+X[k])\sqrt{\eps}$.
\end{observation}

These terms also allow us to simplify the coordinates using the unified formula from \Cref{eq:Pn-def}. 
\begin{observation}
\label{obs: coordinates of points}
For all $\hat{i},\hat{j},\hat{k} \in I_{n}$, using $\odd$ as an indicator function of being odd, we can rewrite the points' coordinates as:   
\begin{itemize}
    \item $\apoints[\hat{i}] = (1/\sqrt{2}, 0, \eps(-x^*_{\lfloor \hat{i}/2 \rfloor} - (-1)^{\hat{i}})) = (1/\sqrt{2}, 0, \eps(-x^*_{\lfloor \hat{i}/2 \rfloor} + 2\odd(\hat{i}) - 1))$
    \item $\bpoints[\hat{j}] = (0, 1/\sqrt{2}, \eps(-y^*_{\lfloor \hat{j}/2 \rfloor} - (-1)^{\hat{j}})) = (0, 1/\sqrt{2}, \eps(-y^*_{\lfloor \hat{j}/2 \rfloor} + 2\odd(\hat{j}) - 1))$
    \item $\cpoints[\hat{k}] = (-\frac{1}{2}, - \frac{1}{2}, \frac{\eps}{\sqrt{2}} (-z^*_{\lfloor \hat{k}/2 \rfloor}) - (-1)^{\hat{k}}\eps) = (-\frac{1}{2}, - \frac{1}{2}, \eps(-\frac{1}{\sqrt{2}} z^*_{\lfloor \hat{k}/2 \rfloor} + 2\odd(\hat{k})-1))$
\end{itemize}
\end{observation}

Let $(i,j,k)$ be a solution to the instance $X$, i.e., $i+j+k = 0$ and $X[i]+ X[j] + X[k] = 0$. We define the following two ball centers of (squared) radii $1+ 3\eps^2 M^2$: 
\[c^+ = (x \eps, y \eps, \frac{1}{\sqrt{2}} + \eps),~c^- = - c^+\]
where $x=  x^*_i$ and $y = y^*_j$; additionally define $z= z^*_k$. 
The balls with centers $c^+, c^-$ are referred to as the top ($B^+$) and bottom ($B^-$) balls, respectively. 

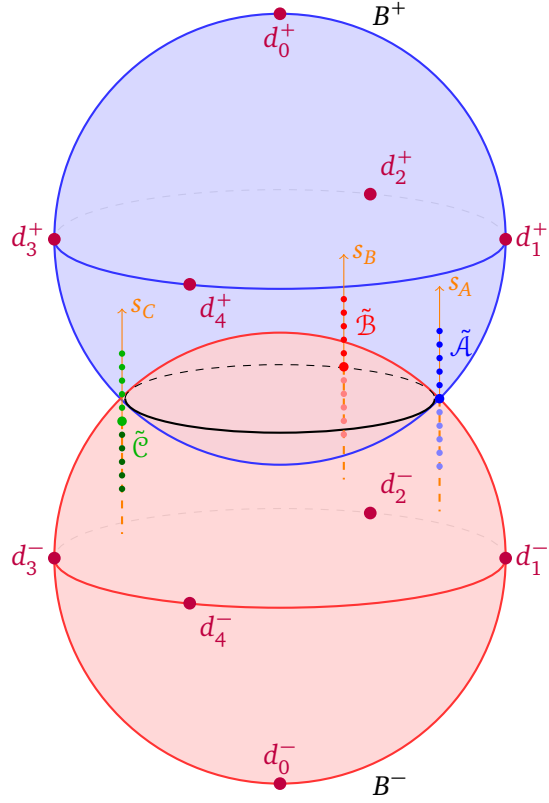
\begin{figure}[!htb]
    \centering
    \begin{tikzpicture}[scale=3]
  \pgfmathsetmacro{\h}{1/sqrt(2)}          
  \pgfmathsetmacro{\r}{1}                  
  \pgfmathsetmacro{\rin}{sqrt(1 - \h*\h)} 
  \pgfmathsetmacro{\yr}{0.22}              

  \pgfmathsetmacro{\es}{0.97}              

  \pgfmathsetmacro{\rinvis}{\es * \rin}
  \pgfmathsetmacro{\yrvis}{\es * \yr}

  \coordinate (C1) at (0, {\h});    
  \coordinate (C2) at (0, -{\h});   

  \draw[dashed,gray] (C1) ellipse ({\r} and {\yr*\r});
  \draw[dashed,gray] (C2) ellipse ({\r} and {\yr*\r});

  \fill[blue!18,opacity=0.85] (C1) circle (\r);
  \fill[red!18,opacity=0.85]  (C2) circle (\r);

  \draw[thick,blue!80] (C1) circle (\r);
  \draw[thick,red!80]  (C2) circle (\r);

  \draw[thick,blue!80] ($(C1)+(-\r,0)$) arc (180:360:{\r} and {\yr*\r});
  \draw[thick,red!80]  ($(C2)+(-\r,0)$) arc (180:360:{\r} and {\yr*\r});

  \draw[dashed,black] ({-\rinvis},0) arc (180:0:{\rinvis} and {\yrvis*\rin});
  \draw[thick,black]  ({-\rinvis},0) arc (180:360:{\rinvis} and {\yrvis*\rin});

 
  \pgfmathsetmacro{\a}{1/sqrt(2)}

  \coordinate (A) at (\a,0);
  \coordinate (B) at ({0 + \a*0.4}, {0 + \a*0.2});
  \coordinate (C) at ({-0.5 + (-0.5)*0.4}, {0 + (-0.5)*0.2});

\node[blue]  at ($(A)+(0.1,0.25)$) {$\tilde{\mathcal{A}}$};
  \node[red] at ($(B)+(0.1,0.2)$) {$\tilde{\mathcal{B}}$};
  \node[green!70!black] at ($(C)+(0.08,-0.1)$) {$\tilde{\mathcal{C}}$};

  \draw[->,orange] (A) -- ++(0,0.5) node[right] {$s_A$};
  \draw[dashed,thick,orange] (A) -- ++(0,-0.5);
  \foreach \i in {1,2,3,4,5}{\fill[blue] ($(A)+(0,0.06*\i)$) circle (0.4pt);}
  \foreach \i in {1,2,3,4,5}{\fill[blue!50] ($(A)+(0,-0.06*\i)$) circle (0.4pt);}

  \draw[->,orange] (B) -- ++(0,0.5) node[right] {$s_B$};
  \draw[dashed,thick,orange] (B) -- ++(0,-0.5);
  \foreach \i in {1,2,3,4,5}{\fill[red] ($(B)+(0,0.06*\i)$) circle (0.4pt);}
  \foreach \i in {1,2,3,4,5}{\fill[red!50] ($(B)+(0,-0.06*\i)$) circle (0.4pt);}

  \draw[->,orange] (C) -- ++(0,0.5)  node[right] {$s_C$};
  \draw[dashed,thick,orange] (C) -- ++(0,-0.5);
  \foreach \i in {1,2,3,4,5}{\fill[green!70!black] ($(C)+(0,0.06*\i)$) circle (0.4pt);}
  \foreach \i in {1,2,3,4,5}{\fill[green!40!black] ($(C)+(0,-0.06*\i)$) circle (0.4pt);}

    \fill[blue] (A) circle (0.6pt);
\fill[red](B) circle (0.6pt);
\fill[green!70!black] (C) circle (0.6pt);

\coordinate (d0) at (0,{1+\a});
\fill[purple] (d0) circle (0.8pt) node[below] {$d^+_0$};

\coordinate (d1) at (1,{\a});
\fill[purple] (d1) circle (0.8pt) node[right] {$d^+_1$};

\coordinate (d2) at ({0+1*0.4},{\a+1*0.2});
\fill[purple] (d2) circle (0.8pt) node[above right] {$d^+_2$};

\coordinate (d3) at (-1,{\a});
\fill[purple] (d3) circle (0.8pt) node[left] {$d^+_3$};

\coordinate (d4) at ({0-0.4},{\a-0.2});
\fill[purple] (d4) circle (0.8pt) node[below right] {$d^+_4$};


\coordinate (dm0) at (0,{-1-\a});
\fill[purple] (dm0) circle (0.8pt) node[above] {$d^-_0$};

\coordinate (dm1) at (1,{-\a});
\fill[purple] (dm1) circle (0.8pt) node[right] {$d^-_1$};

\coordinate (dm2) at ({0+1*0.4},{-\a+1*0.2});
\fill[purple] (dm2) circle (0.8pt) node[above right] {$d^-_2$};

\coordinate (dm3) at (-1,{-\a});
\fill[purple] (dm3) circle (0.8pt) node[left] {$d^-_3$};

\coordinate (dm4) at ({0-0.4},{-\a-0.2});
\fill[purple] (dm4) circle (0.8pt) node[below right] {$d^-_4$};

\node[right=3em] at (d0) {$B^+$};
\node[right=3em] at (dm0) {$B^-$};
\end{tikzpicture}
    \caption{The $2$-center with 2 approximate unit balls associated to a solution of Gap Convolution-3SUM in terms of a 2-clustering for the point sets $\apoints$, $\bpoints$, $\cpoints$, and $\dpoints$. }
    \label{fig:twoballs-gadget}
\end{figure}

\begin{observation}
\label{obs:monotone distance}
Since $x^*, y^*, z^*$ monotonically decrease with respect to the index (as the sequence step of $3$ dominates $2M\sqrt{\eps}$):
\begin{itemize}
    \item The squared distances between $c^+$ and $\apoints[\hat{i}], \bpoints[\hat{j}], \cpoints[\hat{k}]$ are monotonically decreasing in $\hat{i}, \hat{j}, \hat{k}$ respectively. 
    \item The squared distances between $c^-$ and $\apoints[\hat{i}], \bpoints[\hat{j}], \cpoints[\hat{k}]$ are monotonically increasing in $\hat{i}, \hat{j}, \hat{k}$ respectively. 
\end{itemize}   
\end{observation}

We argue that these balls cover all points, as shown in~\Cref{fig:twoballs-gadget}.

\begin{claim}\label{clm:top-ball}
The top ball $B^+$ covers all points $\apoints[> 2i], \bpoints[> 2j], \cpoints[> 2k]$ and $\dset^+$.      
\end{claim}
\begin{proof}
From Observation~\ref{obs:monotone distance}, we only need to verify that the top ball covers the respective odd boundary points $\apoints[2i+1], \bpoints[2j+1], \cpoints[2k+1]$. Since $2i+1$ is odd, the $z$-coordinate evaluates with $+ \eps$.
\begin{eqnarray*}
\|c^+ - \apoints[2i+1]\|^2 &\leq & (\frac{1}{\sqrt{2}} - x \eps)^2 + (y\eps)^2 + (\frac{1}{\sqrt{2}} + \eps - \eps(-x + 1))^2 \\ 
 &= & (\frac{1}{2}- \cancel{\sqrt{2} x \eps} + (x\eps)^2) + (y \eps)^2  + (\frac{1}{2}+ \cancel{\sqrt{2} x \eps}+ (x\eps)^2) \\ 
 &\leq & 1 + 3M^2 \eps^2  
\end{eqnarray*}
This implies that $\|c^+- \apoints[2i+1]\| \leq 1+ 3M^2 \eps^2$. 

Similarly, due to Observation~\ref{obs:monotone distance}, it suffices to verify that the top ball covers $\bpoints[2j+1]$. 
\begin{eqnarray*}
\|c^+ - \bpoints[2j+1]\|^2  & \leq &  (x\eps)^2 + (\frac{1}{\sqrt{2}} - y \eps)^2 + (\frac{1}{\sqrt{2}} +y \eps)^2 \\ 
&=& (x\eps)^2 + (\frac{1}{2} - \cancel{\sqrt{2} y \eps} + (y\eps)^2)+  (\frac{1}{2} +  \cancel{\sqrt{2} y \eps} + (y\eps)^2)  \\ 
&\leq& 1+3M^2 \eps^2 
\end{eqnarray*}

Finally, we verify that the top ball covers $\cpoints[2k+1]$. Since $2k+1$ is odd, $\odd(2k+1) = 1$, evaluating the coordinate with $+ \eps$.
\begin{eqnarray*}
\|c^+ - \cpoints[2k+1]\|^2 & \leq & (\frac{1}{2} + x \eps)^2 + (\frac{1}{2} + y \eps)^2 + (\frac{1}{\sqrt{2}} + \eps - (-\frac{\eps z}{\sqrt{2}} + \eps))^2  \\ 
&=& (\frac{1}{4}+ x \eps +(x \eps)^2) + (\frac{1}{4}+ y\eps + (y\eps)^2)+ (\frac{1}{2}+ z \eps + \frac{1}{2}(\eps z)^2) \\ 
&\leq& 1+ (x+y+ z) \eps + 3M^2 \eps^2 
\end{eqnarray*}
Using that $i,j,k$ is a solution to the instance $X$, we have that $x+y+z = -(X[i]+ X[j]+ X[k]) \sqrt{\eps} = 0$ and therefore the squared distance between $c^+$ and $\cpoints[2k+1]$ is at most $1+ 3M^2 \eps^2$. 

Next, we verify the distance between $c^+$ and the anchor points. In particular, 
\begin{eqnarray*}
\|c^+- d^+_0\|^2 & = & (x \eps)^2 + (y\eps)^2 + 1 \le 1 + 2M^2 \eps^2 \\ 
\|c^+- d^+_1\|^2 & = & (1-t\eps - x\eps)^2 + (y\eps)^2 \leq 1+2M^2 \eps^2 \\ 
\|c^+ - d^+_2\|^2 &= &  (x \eps)^2 + (1-t\eps - y\eps)^2  \leq 1+ 2M^2 \eps^2 \\ 
\|c^+ - d^+_3\|^2 &= & (-1+t\eps - x\eps)^2 + (y\eps)^2 \leq 1+2M^2 \eps^2 \\ 
\|c^+- d^+_4\|^2 & =& (x \eps)^2 + (-1+t\eps - y\eps)^2  \leq 1+ 2M^2 \eps^2
\end{eqnarray*}
The second and third lines use the fact that $(t+x)$ and $(t+y)$ are positive because $|x|$ and $|y|$ are bounded by $3\lfloor n/100 \rfloor + M \sqrt{\eps}$ while $t = 3\lfloor n/100 \rfloor + 1/4$. The fourth and fifth use that $(t-x)$ and $(t-y)$ are positive. 
This concludes the proof of the claim. 
\end{proof}

We can similarly show the coverage of $B^-$. 

\begin{claim}\label{clm:bottom-ball}
The ball $B^-$ covers $\apoints[\leq 2i], \bpoints[\leq 2j], \cpoints[\leq 2k]$ and $\dset^-$.      
\end{claim}
\begin{proof}
Using Observation~\ref{obs:monotone distance}, we verify that $B^-$ covers the point $\apoints[2i]$. Since $2i$ is even, the $z$-coordinate evaluates with $- \eps$.
\begin{eqnarray*}
\|c^- - \apoints[2i]\|^2 & \leq & (\frac{1}{\sqrt{2}} + x \eps)^2 + (y\eps)^2 + (-\frac{1}{\sqrt{2}} - \eps - \eps(-x - 1))^2 \\ 
  &= & (\frac 1 2+ \cancel{\sqrt{2} x \eps} + (x\eps)^2) + (y \eps)^2  + (\frac 1 2- \cancel{\sqrt{2} x \eps}+ (x\eps)^2) \\ 
 &\leq & 1 + 3M^2 \eps^2  
\end{eqnarray*}
Similarly, it suffices to verify that $B^-$ covers $\bpoints[2j]$. 
\begin{eqnarray*}
\|c^- - \bpoints[2j]\|^2  & \leq &  (x\eps)^2 + (\frac{1}{\sqrt{2}} + y \eps)^2 + (\frac{1}{\sqrt{2}} -y \eps)^2 \\ 
&=& (x\eps)^2 + (\frac{1}{2} + \cancel{\sqrt{2} y \eps} + (y\eps)^2)+  (\frac{1}{2} -  \cancel{\sqrt{2} y \eps} + (y\eps)^2)  \\ 
&\leq& 1+3M^2 \eps^2 
\end{eqnarray*}
Finally, we verify that $B^-$ covers $\cpoints[2k]$. Since $2k$ is even, $\odd(2k) = 0$, evaluating the coordinate with $-\eps$.
\begin{eqnarray*}
\|c^- - \cpoints[2k]\|^2 & \leq & (\frac{1}{2} - x \eps)^2 + (\frac{1}{2} - y \eps)^2 + (-\frac{1}{\sqrt{2}} - \eps - (-\frac{\eps z}{\sqrt{2}} - \eps))^2  \\ 
&=& (\frac{1}{4}- x \eps +(x \eps)^2) + (\frac{1}{4}- y\eps + (y\eps)^2)+ (\frac{1}{2}- z \eps + \frac{1}{2}(\eps z)^2) \\ 
&\leq& 1- (x+y+z) \eps + 3M^2 \eps^2 
\end{eqnarray*}
Using $-(x+y+z) = (X[i]+ X[j]+ X[k]) \sqrt{\eps} = 0$, the squared distance to $\cpoints[2k]$ is bounded by $1+ 3M^2 \eps^2$. 

For the anchor points, the identical calculation to \Cref{clm:top-ball} gives $\|c^- - d^-_i\|^2 \leq 1+3M^2 \eps^2$ for all $i \in \{0,1, \ldots, 4\}$, concluding the proof of the claim. 
\end{proof}

\subsection{Soundness}\label{sec: 2 center soundness}

In this section, we will show that a solution of \twocenter implies a solution of Gap Convolution-3SUM.  Consider a clustering solution that uses two balls. Define $B^+$ to be the ball that covers $d^+_0$ and $B^-$ to be that covering $d^-_0$ (notice that they cannot be the same ball since $\|d^+_0 - d^-_0\| > 2.1$). 
For the same reason (i.e., $\|d^+_0 - d^-_i\| > 2.1$  and $\|d^-_0 - d^+_i\| > 2.1$ for all $i$), $B^+$ must cover all the points in $\dset^+$ while $B^-$ covers all in $\dset^-$. 

Our argument is split into three steps. In the first step, we argue that $B^+$ cannot cover any of the sets  $\apoints$, $\bpoints$ or $\cpoints$ entirely (the same goes for $B^-$). In fact, there must be indices $\hat{i}, \hat{j}, \hat{k}$, called \EMPH{switching indices}, for which $B^+$ covers exactly the tails $\apoints[> \hat{i}]$, $\bpoints[> \hat{j}]$ and $\cpoints[> \hat{k}]$, while $B^-$ covers exactly the heads $\apoints[\leq \hat{i}],\bpoints[\leq \hat{j}], \cpoints[\leq \hat{k}]$. 
In the second step, we define $i= \lfloor \hat{i}/2 \rfloor$, $j=\lfloor \hat{j}/2 \rfloor$ and $k= \lfloor \hat{k}/2 \rfloor$, and prove that (i) $\hat{i},\hat{j}$ and $\hat{k}$ are all even and (ii) $i+j+k = 0$. This is tentatively our solution to Gap Convolution-3SUM.  In the third and final step, we argue that $X[i]+X[j]+X[k]=0$, which implies that $X$ is not a No-instance, hence concluding the proof.

Define $c^+ = (x^+ \eps, y^+\eps, \frac{1}{\sqrt{2}} + \eps + z^+)$ and  $c^- = (-x^- \eps, -y^- \eps,  -\frac{1}{\sqrt{2}} - \eps - z^-)$ as the centers of $B^+$ and $B^-$ respectively. 
We can draw the following conclusions from the fact that $B^+$ and $B^-$ cover the anchor points. 

\begin{lemma}[Ranges of variables]
\label{lem: ranges}
The following inequalities hold: 
\begin{itemize}
    \item $z^+, z^- \geq  - 3 M^2 \eps^2$ 
    \item $x^+,y^+, x^-,y^- \in [-t - 3M^2 \eps, t+ 3M^2 \eps]$. 
\end{itemize}
\end{lemma}

Intuitively, since $M^2 \eps \rightarrow 0$, this roughly means that $z^+,z^-$ are non-negative, and that each of  $x^+,x^-,y^+, y^-$ is rigidly confined by the anchor bounds to $[-t, t]$. 

\begin{proof} We show the bounds for $x^+,y^+,z^+$ and the bounds for $x^-,y^-,z^-$ follow from the same argument.  Observe that $\|c^+- d^+_0\| \geq   (1-z^+)$ and since $\|c^+- d^+_0\| \leq 1+3M^2 \eps^2$, it must be that $z^+ \geq  - 3M^2 \eps^2$. 

To bound the range of $x^+$, observe that $1+3M^2 \eps^2 \geq \|c^+ - d^+_1\| \geq 1-t \eps - x^+ \eps$, which implies $x^+ \geq -t - 3M^2 \eps$. Moreover, $1+3M^2 \eps^2 \geq \|c^+ - d^+_3\| \geq 1-t \eps + x^+ \eps$, 
which implies that $x^+ \leq t+ 3M^2 \eps$. 

To bound the range of $y^+$, consider $1+3M^2 \eps^2 \geq \|c^+ - d^+_2\| \geq 1-t \eps - y^+ \eps$, which implies that $y^+ \geq -t - 3M^2 \eps$. Moreover, $1+3M^2 \eps^2 \geq \|c^+ - d^+_4\| \geq 1-t \eps + y^+ \eps$, implying that $y^+ \leq t+3M^2 \eps$. 
\end{proof}

\paragraph{First step: existence of switching indices.} The following two lemmas relate the balls' coverage to the coordinates of the centers. 

\begin{lemma}[Top ball coverage]
\label{lem:coverage for B+}
Let $\hat{i}, \hat{j}, \hat{k} \in I_{n}$ be such that $B^+$ covers $\apoints[\hat{i}], \bpoints[\hat{j}]$ and $\cpoints[\hat{k}]$. Then, we have that  
\begin{itemize}
    \item $x^+ \geq x^*_{\lfloor \hat{i}/2 \rfloor} + 2\cdot  \even(\hat{i})- 6M^2 \eps$
    \item $y^+ \geq y^*_{\lfloor \hat{j}/2 \rfloor}+ 2 \cdot  \even(\hat{j}) - 6M^2 \eps$
    \item $x^+ + y^+ + z^*_{\lfloor \hat{k}/2 \rfloor} \leq -2\sqrt{2} \cdot \even(\hat{k})+ 8M^2 \eps$
\end{itemize}
where $\even$ is an indicator function of whether the integer is even, respectively. 
\end{lemma}

\begin{proof}
From Observation~\ref{obs: coordinates of points}, we rewrite the $z$-coordinate of $\apoints[\hat{i}]$ as $\apoints[\hat{i}]_z = \eps(-x^*_{\lfloor \hat{i}/2\rfloor} - (-1)^{\hat{i}})$.    
Therefore, 
\begin{align*}
    1+ 3M^2 \eps^2 &\geq \|c^+ - \apoints[\hat{i}]\|^2 \\
    &\geq (\frac{1}{\sqrt{2}} - x^+\eps)^2 + (y^+ \eps)^2 + (\frac{1}{\sqrt{2}} + \eps+ z^+ + \eps x^*_{\floor{\hat{i}/2}} + \eps \cdot (-1)^{\hat{i}})^2   \\ 
&\geq  1/2- \sqrt{2} x^+ \eps + 1/2 + \sqrt{2} \eps x^*_{\floor{\hat{i}/2}} + \sqrt{2} \eps (1+(-1)^{\hat{i}}) -  5M^2 \eps^2 \quad \text{(since $z^+\geq -3M^2\eps^2$)}\\
&= 1 - \sqrt{2} x^+ \eps  + \sqrt{2} \eps x^*_{\floor{\hat{i}/2}} + 2\sqrt{2} \eps \even(\hat{i}) -  5M^2 \eps^2
\end{align*}
Rearranging the resulting inequality implies that $x^+ \geq x^*_{\floor{\hat{i}/2}} + 2 \cdot \even(\hat{i}) - 6M^2 \eps$. 

Similarly, we rewrite the $z$-coordinate of $\bpoints[\hat{j}]$ as $\bpoints[\hat{j}]_z = \eps (-y^*_{\lfloor \hat{j}/2\rfloor} - (-1)^{\hat{j}})$, and the same calculation gives $y^+ \geq y^*_{\floor{\hat{j}/2}} + 2 \cdot \even(\hat{j}) - 6M^2 \eps$. 

Finally, using Observation~\ref{obs: coordinates of points}, we have $\cpoints[\hat{k}]_z = -\frac{\eps}{\sqrt{2}} z^*_{\floor{\hat{k}/2}} - \eps \cdot (-1)^{\hat{k}}$. We get: 
\begin{align*}
    1+3M^2 \eps^2 &\geq \|c^+ - \cpoints[\hat{k}]\|^2  \\
    &\geq (x^+\eps + 1/2)^2 + (y^+ \eps + 1/2)^2 + (1/\sqrt{2} + \eps+z^+ + \eps z^*_{\floor{\hat{k}/2}}/\sqrt{2} + \eps \cdot (-1)^{\hat{k}})^2  \\ 
&\geq  (x^+\eps + 1/4)  + (y^+ \eps + 1/4) + (1/2 + \eps z^*_{\floor{\hat{k}/2}} + \sqrt{2} \eps \cdot (1+(-1)^{\hat{k}}) ) - 5M^2 \eps^2 \\
&= 1 + \eps (x^+ + y^+ + z^*_{\floor{\hat{k}/2}} + 2\sqrt{2} \cdot \even(\hat{k})) - 5M^2 \eps^2
\end{align*}
In the penultimate inequality, we used $z^+ \geq -3 M^2 \eps^2$. Rearranging the inequality above gives $x^+ + y^+ + z^*_{\floor{\hat{k}/2}} \leq -2\sqrt{2} \cdot \even(\hat{k}) + 8M^2 \eps$ as desired. 
\end{proof}

The identical algebraic approach gives the matching bounds for the bottom ball:

\begin{lemma}[Bottom ball coverage]
\label{lem:coverage for B-}
Let $\hat{i}, \hat{j}, \hat{k} \in I_{n}$ be such that $B^-$ covers $\apoints[\hat{i}], \bpoints[\hat{j}]$ and $\cpoints[\hat{k}]$. Then, we have that  
\begin{itemize}
    \item $x^- \leq x^*_{\lfloor \hat{i}/2 \rfloor} - 2 \cdot \odd(\hat{i}) + 6M^2 \eps$
    \item $y^- \leq y^*_{\lfloor \hat{j}/2 \rfloor} - 2\cdot \odd(\hat{j}) + 6M^2 \eps$
    \item $x^- + y^- + z^*_{\lfloor \hat{k}/2 \rfloor} \geq 2\sqrt{2} \cdot \odd(\hat{k}) - 8M^2 \eps$
\end{itemize}
\end{lemma}

Now we are ready to show the existence of switching indices. 

\begin{lemma}
\label{lem: switching}
There exist $\hat{i}, \hat{j}, \hat{k} \in I_{n}$ such that $B^+$ covers exactly the tails $\apoints[> \hat{i}]$, $\bpoints[> \hat{j}]$, $\cpoints[> \hat{k}]$, while $B^-$ covers exactly the heads $\apoints[\leq \hat{i}],\bpoints[\leq \hat{j}], \cpoints[\leq \hat{k}]$.
\end{lemma}

\begin{proof}
To see that $B^+$ cannot cover $\apoints[-2n]$, assume for contradiction that it does. From Lemma~\ref{lem:coverage for B+}, we evaluate the lower bound with $\hat{i} = -2n$: 
$x^+ \ge x^*_{-n} + 2 - 6M^2\eps = -3(-n) - X[-n]\sqrt{\eps} + 2 - 6M^2\eps \gtrsim 3n + 2$ (where we write $a \gtrsim b$ for, say, $a \ge b - \eps^{0.2}$). 
However, Lemma~\ref{lem: ranges} implies that $x^+ \leq t +3M^2 \eps \lesssim 3\lfloor n/100 \rfloor + 1/4$. The condition $3n + 2 \lesssim 3\lfloor n/100 \rfloor + 1/4$ is a contradiction for $n \ge 1$. 
Similarly, if $B^-$ covers $\apoints[2n+1]$, Lemma~\ref{lem:coverage for B-} forces $x^- \le x^*_{n} \lesssim -3n$, which contradicts the lower anchor bound $x^- \ge -t-3M^2 \eps \gtrsim -3\lfloor n/100 \rfloor - 1/4$.  
Define $\hat{i}$ as the minimum integer $i'$ such that $B^+$ covers $\apoints[i' +1]$, so (by monotonicity of distance) the ball $B^+$ covers only $\apoints[> \hat{i}]$, while $B^-$ covers  $\apoints[\leq \hat{i}]$. Note that $-2n \leq \hat{i} < 2n+1$.

By identical bounding, if $B^+$ covers $\bpoints[-2n]$, $y^+ \gtrsim 3n + 2$, contradicting $y^+ \lesssim 3\lfloor n/100 \rfloor + 1/4$. If $B^-$ covers $\bpoints[2n+1]$, $y^- \lesssim -3n$, contradicting $y^- \gtrsim -3\lfloor n/100 \rfloor - 1/4$. We define $\hat{j}$ similarly.

Similarly, if $B^+$ covers $\cpoints[-2n]$, Lemma~\ref{lem:coverage for B+} enforces $x^+ + y^+ + z^*_{-n} \lesssim 0 \implies x^+ + y^+ + 3n \lesssim 0 \implies x^+ + y^+ \lesssim -3n$. However, the anchor points rigidly enforce $x^+, y^+ \gtrsim -3\lfloor n/100 \rfloor - 1/4 \implies x^+ + y^+ \gtrsim -6\lfloor n/100 \rfloor - 1/2$. The requirement $-6\lfloor n/100 \rfloor - 1/2 \lesssim -3n$ is a contradiction for $n \ge 1$. 
Symmetrically, if $B^-$ covers $\cpoints[2n+1]$, it satisfies $x^- + y^- + z^*_{n} \gtrsim 0 \implies x^- + y^- - 3n \gtrsim 0 \implies x^- + y^- \gtrsim 3n$, violating the upper anchors $x^- + y^- \lesssim 6\lfloor n/100 \rfloor + 1/2$. Thus there exists a switching index $\hat{k}$ where $B^-$ covers $\cpoints[\leq \hat{k}]$ while $B^+$ covers  $\cpoints[> \hat{k}]$. 
\end{proof}

\paragraph{Second step: showing $i+j+k=0$.} Let $\hat{i},\hat{j}$ and $\hat{k}$ be the indices in \Cref{lem: switching} and define $i = \floor{\hat{i}/2}$, $j = \floor{\hat{j}/2}$, and $k = \floor{\hat{k}/2}$. We claim that (a) the indices $\hat{i}$, $\hat{j}$, $\hat{k}$ are even and (b) $i+j+k = 0$.  

By applying Lemma~\ref{lem:coverage for B+} since $B^+$ covers $\apoints[\hat{i}+1], \bpoints[\hat{j}+1], \cpoints[\hat{k}+1]$, we learn that 
\begin{eqnarray*}
x^+ &\geq& x^*_{\lfloor (\hat{i}+1)/2 \rfloor} + 2\cdot  \even(\hat{i}+1)- 6M^2 \eps  \\ 
y^+ &\geq& y^*_{\lfloor (\hat{j}+1)/2 \rfloor}+ 2 \cdot  \even(\hat{j}+1) -6M^2 \eps \\ 
x^+ + y^+ + z^*_{\lfloor (\hat{k}+1)/2 \rfloor} &\leq& -2\sqrt{2} \cdot \even(\hat{k}+1)+ 8M^2 \eps
\end{eqnarray*}
Subtracting the sum of the first two inequalities from the third gives us 
\begin{equation}
\label{eq: upper-1}
 0 \geq x^*_{\lfloor (\hat{i}+1)/2 \rfloor} + y^*_{\lfloor (\hat{j}+1)/2 \rfloor} + z^*_{\lfloor (\hat{k}+1)/2 \rfloor}  + 2\even(\hat{i}+1) + 2\even(\hat{j}+1) + 2\sqrt{2}\even(\hat{k}+1) - 20M^2 \eps
\end{equation}

Similarly, by applying Lemma~\ref{lem:coverage for B-} since $B^-$ covers $\apoints[\hat{i}], \bpoints[\hat{j}], \cpoints[\hat{k}]$, we have 
\begin{eqnarray*}
x^- &\leq& x^*_{\lfloor \hat{i}/2 \rfloor} - 2 \cdot \odd(\hat{i}) + 6M^2 \eps \\ 
y^- &\leq& y^*_{\lfloor \hat{j}/2 \rfloor} - 2\cdot \odd(\hat{j}) + 6M^2 \eps \\ 
x^- + y^- + z^*_{\floor{\hat{k}/2}} &\geq & 2\sqrt{2} \cdot \odd(\hat{k}) - 8M^2 \eps
\end{eqnarray*}
Subtracting the sum of the first two inequalities from the third gives us 
\begin{equation}
\label{eq: lower-1}
x^*_{\lfloor \hat{i}/2 \rfloor} + y^*_{\lfloor \hat{j}/2 \rfloor} + z^*_{\lfloor \hat{k}/2 \rfloor} \geq  2\odd(\hat{i}) + 2\odd(\hat{j}) + 2\sqrt{2}\odd(\hat{k}) - 20M^2 \eps  
\end{equation}

\begin{claim} The indices 
$\hat{i}, \hat{j}$ and $\hat{k}$ are even. 
\end{claim}
\begin{proof}
By subtracting \Cref{eq: upper-1} from \Cref{eq: lower-1}, we get: 
\begin{equation}
\label{eq:even-odd}
    \begin{split}
        &(x^*_{\floor{\hat{i}/2}} - x^*_{\floor{(\hat{i}+1)/2}} ) + (y^*_{\floor{\hat{j}/2}} - y^*_{\floor{(\hat{j}+1)/2}} ) + ( z^*_{\floor{\hat{k}/2}}- z^*_{\floor{(\hat{k}+1)/2}})\\& \qquad \geq 2(\odd(\hat{i})+\even(\hat{i}+1)) + 2(\odd(\hat{j})+\even(\hat{j}+1)) + 2\sqrt{2}(\odd(\hat{k})+\even(\hat{k}+1))-  40M^2\eps 
    \end{split}
\end{equation}
Using the fact that $\even(p+1) = \odd(p)$, we can rewrite the right-hand side as $4\odd(\hat{i}) + 4\odd(\hat{j}) + 4\sqrt{2}\odd(\hat{k}) - 40M^2\eps$.

Notice that for each sequence $q^* \in \{x^*, y^*, z^*\}$, if $p$ is even then $\floor{p/2} = \floor{(p+1)/2}$ and thus $q^*_{\floor{p/2}} - q^*_{\floor{(p+1)/2}} = 0$, while if $p$ is odd, the step yields $q^*_{\floor{p/2}} - q^*_{\floor{(p+1)/2}} \in 3 \pm 2M\sqrt{\eps}$. Thus we have:
\begin{equation}
    3\odd(\hat{i}) + 3\odd(\hat{j}) + 3\odd(\hat{k}) + 6M\sqrt{\eps} \geq 4\odd(\hat{i}) + 4\odd(\hat{j}) + 4\sqrt{2}\odd(\hat{k}) - 40M^2\eps
\end{equation}

Rearranging isolates the odd multipliers against the error terms:
\begin{equation}
    \odd(\hat{i}) + \odd(\hat{j}) + (4\sqrt{2}-3)\odd(\hat{k}) \leq 6M\sqrt{\eps} + 40M^2\eps
\end{equation}

Because the coefficients for the binary $\odd$ indicators are all $\geq 1$, and $M = n^2$ and $\eps = (10n)^{-100}$, the right side is strictly less than $1$. The inequality mathematically collapses unless $\odd(\hat{i})=\odd(\hat{j})=\odd(\hat{k})=0$, proving all indices are even. 
\end{proof}

\paragraph{Third step: showing $X[i]+X[j]+X[k]=0$.} Using the fact that $\hat{i}, \hat{j}$ and $\hat{k}$ are even, both penalties drop to zero. Substituting into the upper bound (\Cref{eq: upper-1}): 
\[ 0 \geq x^*_i + y^*_j + z^*_k - 20M^2 \eps \]
Using Observation~\ref{obs: relating x,y,z}, we get $0 \geq -3(i+j+k) - (X[i]+X[j]+X[k]) \sqrt{\eps} - 20M^2 \eps$. Because the minimum step for the outer gap is $\pm 3$, while the inner error term is strictly bounded by $3M\sqrt{\eps} \ll 1$, we must have $i+j+k \geq 0$.

Similarly, substituting the even indices into the lower bound (\Cref{eq: lower-1}): 
\[ x^*_i + y^*_j + z^*_k \geq -20\cdot M^2 \eps \]
This yields $-3(i+j+k) - (X[i]+X[j]+X[k]) \sqrt{\eps} \geq -20M^2 \eps$, forcing $i+j+k \leq 0$. Therefore $i+j+k=0$.

With the outer coordinates neutralizing to zero, the remaining inequalities strictly bound the array entry sum:
\[X[i]+X[j]+X[k] \in [-20M^2\sqrt{\eps}, 20M^2\sqrt{\eps}] \]
Since $20M^2\sqrt{\eps} < 0.5$ and the array contains integers, $X[i]+X[j]+X[k] = 0$. 

Hence, we have successfully found indices $i,j,k$ such that $i+j+k=0$ and $X[i]+X[j]+X[k]=0$. 
\section{\texorpdfstring{$6$}{6}-Center in the Plane}\label{sec:6CenterPlane}

In this section, we prove \Cref{thm:6-center} by showing a reduction from the Gap Convolution-3SUM problem. From the input array, we will construct in $O(n)$ time a set of $O(n)$ points $P$ on the plane such that:
\begin{quote}
     $P$ can be covered by exactly 6 disks of radius $\radius$ if the instance of the  Gap Convolution-3SUM problem  is YES. Conversely,  if the instance of the  Gap Convolution-3SUM problem is NO, then $P$ cannot be covered by exactly 6 disks of radius $\radius$.    
\end{quote}

Therefore, a truly subquadratic-time algorithm for the \sixcenter problem would yield a truly subquadratic-time algorithm for the Gap Convolution-3SUM problem, contradicting the 3SUM hypothesis.  The idea is to adapt the approach for \tencenter in \Cref{sec:10center} by removing some of the hexagonal vertices (recall that these roughly correspond to locations of the centers of circles defining the clustering). Each of the hexagons used in the construction of the point set for \tencenter is essentially replaced with a quadrilateral, eliminating the need for $4$ locations for centers, at the cost of an added layer of technical considerations. The ensuing necessary adjustment of the placement of the points along the two quadrilaterals is precisely what this section addresses.    

In this section, we will first provide the details of the reduction from the Gap Convolution-3SUM problem to the \sixcenter problem, and then analyze the completeness and soundness of the reduction.

\subsection{The Reduction}\label{subsec:6cen-reduction}

\begin{figure}[htbp]
    \centering
    \includegraphics[page=4, width=0.7\linewidth]{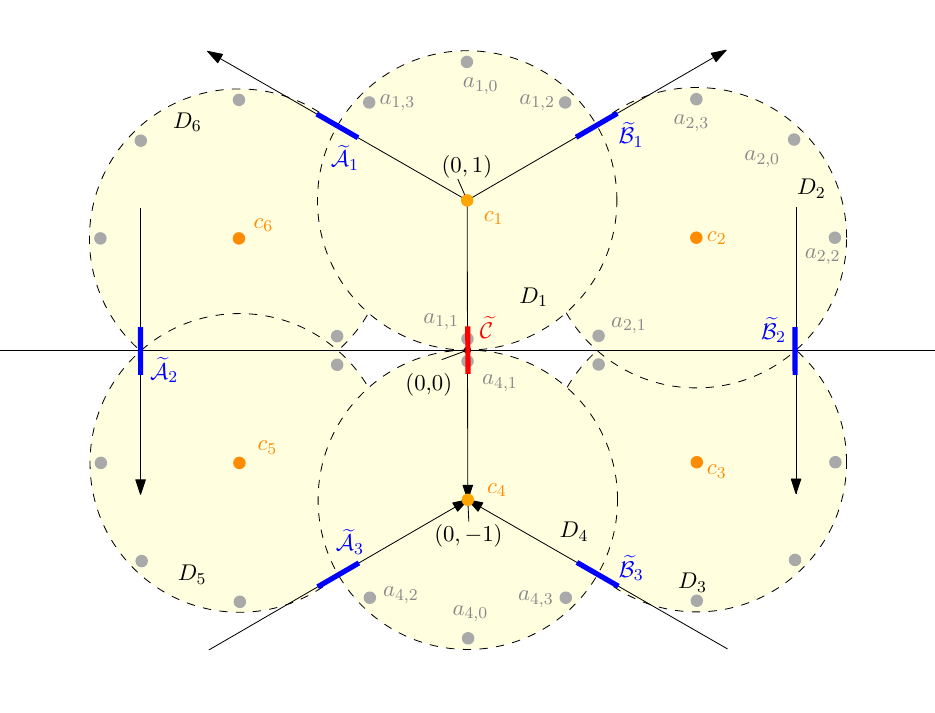}

    \caption{Placement of points in \Cref{eq:6cen-points} and anchor points on the plane. The dashed circles of radius $1-\eps + \eps^{1.7}$ are centered at $c_1,\ldots, c_6$. The arrows indicate the direction of the point sets that we will place. }
    \label{fig:6cen-anchorpoints}
\end{figure}

For a given array $X=X[-n\ldots n]$ of $(2n+1)$ integers, a central point $s\in \real^d$, a vector $\Vec{w} \in \real^d$, we define a set of $2(2n+1)$  points in $\real^d$, denoted by $ \pset_n(A,s,\Vec{w})[(-2n)\ldots (2n+1)]$ where for every $i\in I_n:=\{-2n, \dots, 2n+1\}$:

\begin{equation}\label{eq:Pn-def-simpl}
    \pset_n(X,s,\Vec{w})[i] \coloneq s + \left( 3 \lfloor i/2 \rfloor \cdot \eps + X[\lfloor i/2 \rfloor] \cdot \eps^{1.5} - (-1)^i \eps \right) \cdot \Vec{w}. 
\end{equation}

Note that $\pset_n(X,s,\Vec{w})$ above is a simplified version of $\pset_n(X,s,\Vec{w},\alpha)$ in \Cref{eq:Pn-def} with the offset constant $\alpha = 1$. 

For every point $a \in  \real^2$, we define \EMPH{$\bar{a}$} to be its mirror over the $y$-axis. That is, if $a = (a_x,a_y)$, then $\bar{a} = (-a_x,a_y)$.

\paragraph{The point sets.} First, we define the following points (see \Cref{fig:6cen-anchorpoints}):
\begin{equation}\label{eq:6cen-points}
    \begin{split}
        o = (0,0), \quad p = (\sqrt{3}/2,3/2), &\quad c = (0,1), \quad d = (\sqrt{3},1), \\
       q = (\frac{\sqrt{3}}{2} + \frac{\sqrt{7}}{2}, 0), & \quad \hat{q} = q - ((\frac{11}{\sqrt{7}}-\sqrt{3})\eps,0)
    \end{split}
\end{equation}
Note that $\hat{q}$ is within $O(\eps)$ distance from $q$, and $\|p,q\| = 2$. We will use the following parameter:
\begin{equation}\label{eq:param-6cen}
 \gamma = \frac{\sqrt{21}}{6} - \frac{1}{2}
\end{equation}
Next, we construct the following 7 sets of points from $X$ (see \Cref{fig:cluster-gadget}): 
\begin{equation}\label{eq:6cen-pointsets}
    \begin{split}
        \cpoints &\coloneq \pset_n(X, o, \vec{co})\\
        \apoints_1 &\coloneq \pset_n(X, \bar{p}, \vec{c\bar{p}})\\
        \apoints_3 &\coloneq \pset_n(X, -p, \vec{cp})\\
        \bpoints_1 &\coloneq \pset_n(X, p, \vec{cp})\\
        \bpoints_3 &\coloneq \pset_n(X, -\bar{p}, \vec{c\bar{p}})\\
        \apoints_2 &\coloneq \pset_n(X, -\hat{q}, \gamma \cdot \vec{co})\\
        \bpoints_2 &\coloneq \pset_n(X, \hat{q}, \gamma \cdot \vec{co})
    \end{split}
\end{equation}

\paragraph{Anchor points.}  Let $c_1 = c$, $c_2 = \frac{p+q}{2}$, $c_3 = -\frac{\bar{p}+\bar{q}}{2}$, $c_4 = -c$, $c_5 = -c_2$ and $c_6 = -c_3$ roughly be the expected disk centers.
We now define a set $A$ of anchor points, consisting of four points per disk placed at a distance $\Delta \coloneq 1 - (\frac{n}{10} + 1)\eps$ from each expected center as follows (see \Cref{fig:6cen-anchorpoints}):
\begin{itemize}
    \item $a_{1,0} = c_1 + \Delta \cdot (0, 1)$, $a_{1,1} = c_1 - \Delta \cdot (0, 1)$, $a_{1,2} = c_1 + \Delta \cdot (\frac{1}{\sqrt{2}}, \frac{1}{\sqrt{2}})$ and $a_{1,3} = c_1 + \Delta \cdot (-\frac{1}{\sqrt{2}}, \frac{1}{\sqrt{2}})$ for disk $D_1$,
    \item $a_{2,0} = c_2 + \Delta \cdot (\frac{1}{\sqrt{2}}, \frac{1}{\sqrt{2}})$, $a_{2,1} = c_2 - \Delta \cdot (\frac{1}{\sqrt{2}}, \frac{1}{\sqrt{2}})$, $a_{2,2} = c_2 + \Delta \cdot (0,1)$ and $a_{2,3} = c_2 + \Delta \cdot (1,0)$ for disk $D_2$,
    \item then for the remaining disks, for $\ell = 0,\dots,3$, we set $a_{6,\ell} = \overline{a_{2,\ell}}$, $a_{3,\ell} = -a_{6,\ell}$, $a_{4,\ell} = -a_{1,\ell}$ and $a_{5,\ell} = -a_{2,\ell}$.
\end{itemize}

This completes our reduction.

\begin{figure}[!htp]
    \centering
    \includegraphics[width=0.7\linewidth, page=1]{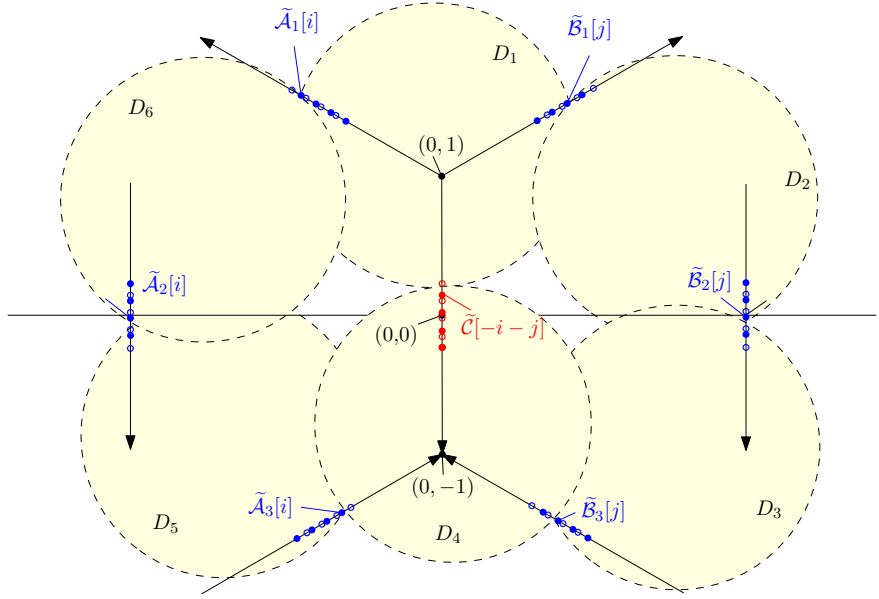}
    \caption{The cluster gadget for \sixcenter defined in Step 1 and 2 together with a possible solution to the \sixcenter problem with 6 disks of radius \radius.}
    \label{fig:cluster-gadget}
\end{figure}

\subsection{Properties of the Points}

We first prove an auxiliary lemma.

\begin{lemma}\label{lem:avg_center}
    Let $\delta \in (0,1)$ and $x,y \in \real^2$ with $\|x-y\| = 2 - 2\delta$.
    Then, the center of any disk of radius $\radius$ containing $x$ and $y$ has distance at most $\sqrt{2\delta}$ from $\frac{x+y}{2}$.
\end{lemma}
\begin{proof}
    By translating and rotating, we can assume without loss of generality that  $x = -y$ and $||x|| = ||y|| = 1-\delta$. Let $d$ be the center.  Then:
    \begin{equation*}
        \begin{split}
            2 \geq 2(\radius) &\geq ||x-d||^2 + ||y-d||^2 \\
            &= ||x||^2 + ||y||^2 - 2\cancel{\langle (x+y), d\rangle} + 2||d||^2\\
             & \geq (1-2\delta)+ (1-2\delta) + 2||d||^2 \qquad\text{(since $x+y = 0$)}\\
             & = 2 - 4\delta + 2 ||(x+y)/2 - d||^2  \qquad\text{(since $x+y = 0$)}
        \end{split}
    \end{equation*}
    Thus, $||(x+y)/2 - d||^2 \leq 2 \delta$, as desired. 
\end{proof}

Recall that $A$ is the set of anchor points; we can now show that these force the disks to be positioned roughly as expected. 

\begin{lemma}\label{lm:6cen-anchor}
    Let $D$ be some set of $6$ disks of radius $\radius$ with centers $d_1, \dots, d_6$.
    \begin{enumerate}
        \item If $\|d_\ell - c_\ell\| \leq 0.1 n \eps$ for all $\ell \in [6]$, then the disks in $D$ cover $A$. 
        \item If the disks in $D$ cover $A$, then we can label their centers such that $\|d_\ell - c_\ell\| \leq 0.3 n \eps$ for all $\ell\in [6]$.
    \end{enumerate}
\end{lemma}
\begin{proof}
    For the first item, consider an arbitrary anchor point $a_{\ell,s} \in A$.
    By construction $\|c_\ell - a_{\ell,s}\| = \Delta$, thus by the triangle inequality:
    \begin{align*}
        \|d_\ell - a_{\ell,s}\|
        &\leq \|d_\ell - c_\ell\| + \|c_\ell - a_{\ell,s}\| \\
        &\leq 0.1 n \eps + \Delta = 1-\eps \leq \radius.
    \end{align*}
    Hence each anchor point is covered by the relevant disk.

    For the second item, first note that $a_{1,0}, \dots, a_{6,0}$ have pairwise distance at least $2$ and must therefore be covered by distinct disks; we will use $d_\ell$ for the center of the disk covering $a_{\ell,0}$.
    Next, we show that $\|a_{\ell,1} - a_{\ell',0}\| \geq 2$ for any $\ell \neq \ell'$ (see \Cref{fig:6cen-anchorpoints}).
    For pairs whose disks are not vertically adjacent, this clearly holds.
    For $a_{4,1}$ and $a_{1,0}$ (and symmetrically $a_{1,1}$ and $a_{4,0}$) the distance is exactly $2$.
    We will prove that the distance between $a_{2,1}$ and $a_{3,0}$ is also more than $2$ (and symmetrically, the distances between $a_{3,1}$ and $a_{2,0}$, between $a_{5,1}$ and $a_{6,0}$, and between $a_{6,1}$ and $a_{5,0}$).
    For this, recall that

    \begin{align*}
    a_{2,1}
    &= \frac{p+q}{2} - \Delta \cdot \left( \frac{1}{\sqrt{2}}, \frac{1}{\sqrt{2}} \right)
    = \left( \frac{\sqrt{3}}{2} + \frac{\sqrt{7}}{4} - \frac{\Delta}{\sqrt{2}},\ \frac{3}{4} - \frac{\Delta}{\sqrt{2}} \right), \\
    a_{3,0}
    &= -a_{6,0} = -\overline{a_{2,0}}
    = \overline{\left( \frac{p+q}{2} + \Delta \cdot \left( \frac{1}{\sqrt{2}}, \frac{1}{\sqrt{2}} \right) \right)}
    = \left( \frac{\sqrt{3}}{2} + \frac{\sqrt{7}}{4} + \frac{\Delta}{\sqrt{2}},\ -\frac{3}{4} - \frac{\Delta}{\sqrt{2}} \right).
\end{align*}
Thus, $a_{2,1} - a_{3,0} = \left(-\Delta \cdot \sqrt{2}, \frac{3}{2}\right)$ and
\begin{align*}
    \|a_{2,1} - a_{3,0}\|^2
    &= 2\Delta^2 + \frac{9}{4} \\
    &\geq 2\left( 1 - 2\left( \frac{n}{10} + 1 \right)\eps \right) + \frac{9}{4} \\
    &= 4 + \frac{1}{4} - 4\left( \frac{n}{10} + 1 \right)\eps \\
    &> 4,
\end{align*}
    meaning their distance is also more than $2$.

    Thus, $a_{\ell,1}$ and $a_{\ell,0}$ must be covered by the same disk, for any $\ell$.
    By \Cref{lem:avg_center}, $d_\ell$ is now within distance $\sqrt{2(\frac{n}{10} + 1)\eps}$ of $\frac{a_{\ell,0} + a_{\ell,1}}{2} = c_\ell$.
    Now, $a_{2,\ell}$ and $a_{3,\ell}$ must also be covered by the disk with center $d_\ell$, as they are too distant from the other expected centers $c_{\ell'}$.
    Finally, this lets us use \Cref{lem:fixing_center} to conclude that $d_\ell$ is within distance $2.7(\frac{n}{10} + 1)\eps \leq 0.3 n \eps$ from $c_\ell$, since we can assume $n \ge 100$.
\end{proof}

The relations between $\cpoints$, $\apoints_1$, $\apoints_3$, $\bpoints_1$ and $\bpoints_3$ are exactly as in the 10-center case.
What remains to prove is that the relations involving $\apoints_2$ and $\bpoints_2$ are also as expected.

\begin{lemma}\label{lm:B12-prop}
    \begin{enumerate}
        \item For any $i,j \in I_n$ with $i \leq j$ there does not exist a disk of radius $\radius$ containing $\bpoints_1[i]$ and $\bpoints_2[j]$. \label{lm:item_B1B2ij}
        \item For any $j \in I_n$ with $j + 1 \in I_n$ and $j$ odd, there does not exist a disk of radius $\radius$ containing $\bpoints_1[j+1]$ and $\bpoints_2[j]$.\label{lm:item_B1B2jj+1}
        \item For any $j \in I_n$ with $j + 1 \in I_n$ and $j$ even, there exists a disk of radius $\radius$ containing $\bpoints_1[>j]$ and $\bpoints_2[\leq j]$ that has its center within distance $(|j| + 1)\eps$ from $\frac{p+q}{2}$. \label{lm:item_B1B2jj}
    \end{enumerate}
\end{lemma}
\begin{proof}  Recall that the point sets are defined as:
\begin{align*}
    \bpoints_1[i] &= p + \left(3\lfloor i/2 \rfloor \eps + X[\lfloor i/2 \rfloor] \cdot \eps^{1.5} - (-1)^i \eps \right) \cdot \left(\frac{\sqrt{3}}{2}, \frac{1}{2}\right) \\
    \bpoints_2[j] &= \hat{q} + \left(3\lfloor j/2 \rfloor \eps + X[\lfloor j/2 \rfloor] \cdot \eps^{1.5} - (-1)^j \eps \right) \cdot (0, -\gamma)
\end{align*}
where $p = (\frac{\sqrt{3}}{2}, \frac{3}{2})$, $q = (\frac{\sqrt{3}}{2} + \frac{\sqrt{7}}{2}, 0), \hat{q} = q - ((\frac{11}{\sqrt{7}}-\sqrt{3})\eps,0)$, and $\gamma = \frac{\sqrt{21}}{6} - \frac{1}{2}$.

We expand the squared distance $\|\bpoints_1[i] - \bpoints_2[j]\|^2$ using the inner product:
\begin{align*}
    \|\bpoints_1[i] - \bpoints_2[j]\|^2 &= (\bpoints_1[i] - \bpoints_2[j])(\bpoints_1[i] - \bpoints_2[j]) \\
     &= (p-q) \cdot (p-q) \\
    &\qquad + 2\left(3\lfloor i/2 \rfloor \eps + X[\lfloor i/2 \rfloor] \cdot \eps^{1.5} - (-1)^i \eps \right) \cdot \left(\frac{\sqrt{3}}{2}, \frac{1}{2}\right) \cdot (p-q) \\
    &\qquad - 2\left(3\lfloor j/2 \rfloor \eps + X[\lfloor j/2 \rfloor] \cdot \eps^{1.5} - (-1)^j \eps \right) \cdot (0, -\gamma) \cdot (p-q) \\
    &\qquad + 2\left( \left(\frac{11}{\sqrt{7}} - \sqrt{3}\right) \eps, 0 \right) \cdot (p-q) + \xi
\end{align*}
with $|\xi| \le 1000 n^4 \eps^2 < \eps^{1.7}$.

Plugging in $p-q = (-\frac{\sqrt{7}}{2}, \frac{3}{2})$, we get:

\begin{align*}
    (p-q) \cdot (p-q) &= 4\\
    (\frac{\sqrt{3}}{2}, \frac{1}{2}) \cdot (p-q) &= \frac{3 - \sqrt{21}}{4}\\
    (0, -\gamma) \cdot (p-q) &= -\frac{3}{2}\gamma = \frac{3 - \sqrt{21}}{4}\\
    \left( \left( \frac{11}{\sqrt{7}} - \sqrt{3} \right) \cdot \eps, 0 \right) \cdot (p - q) &= \left( -\frac{11}{2} + \frac{\sqrt{21}}{2} \right) \cdot \eps
\end{align*}

Then the squared distance $\|\bpoints_1[i] - \bpoints_2[j]\|^2$ simplifies to:
\begin{align*}
    \|\bpoints_1[i] - \bpoints_2[j]\|^2 &= 4 + \left(\underbrace{ \left( 3\lfloor j/2 \rfloor \eps + X[\lfloor j/2 \rfloor] \cdot \eps^{1.5} - (-1)^j \eps \right) }_{\coloneq x_j}\right. \\
    &\qquad \left. -  \underbrace{\left( 3\lfloor i/2 \rfloor \eps + X[\lfloor i/2 \rfloor] \cdot \eps^{1.5} - (-1)^i \eps \right) }_{\coloneq x_i}\right) \cdot \frac{\sqrt{21} - 3}{2} \\
    &\qquad - (11 - \sqrt{21})\eps + \xi
\end{align*}

Let $r = 1 - \eps + \eps^{1.7}$. Observe that:
\begin{equation}\label{eq:2r-sqrt}
    \begin{split}
        (2r)^2 &= 4(1 - \eps + \eps^{1.7})^2 \\
    &= 4(1 - 2\eps + 2\eps^{1.7} + \eps^2 + \eps^{3.4} - 2\eps^{2.7}) \\
    &\in (4 - 8\eps + (8\pm 0.4)\eps^{1.7})
    \end{split}
\end{equation}

Now we consider three different cases, which correspond to the three items in the lemma.

\noindent\textbf{Case 1:} if $i \leq j$, then $x_j \geq x_i$, so:
\begin{equation}
\begin{split}
     \|\bpoints_1[i] - \bpoints_2[j]\|^2 &\geq 4 - (11 - \sqrt{21})\eps - \eps^{1.7} \\
     &\geq 4-6.42\eps - \eps^{1.7} > (2r)^2 \qquad \text{(by \Cref{eq:2r-sqrt})}
\end{split}
\end{equation}

Thus, \cref{lm:item_B1B2ij} of the lemma holds.

\noindent\textbf{Case 2:} $i = j + 1$ and $j$ is odd.  Then we have:
\begin{align*}
    x_j &= 3\lfloor j/2 \rfloor \eps + X[\lfloor j/2 \rfloor] \cdot \eps^{1.5} + \eps \\
    x_i &= 3(\lfloor j/2 \rfloor + 1) \eps + X[\lfloor i/2 \rfloor] \cdot \eps^{1.5} - \eps~,
 \end{align*}
giving $x_j - x_i \geq -\eps - 2n^2 \eps^{1.5}$. Thus, 
\begin{align*}
    \|\bpoints_1[i] - \bpoints_2[j]\|^2 &\ge 4 - \frac{\sqrt{21} - 3}{2}\eps - (11 - \sqrt{21})\eps - \underbrace{3n^2 \eps^{1.5}}_{\leq 0.1\eps} \\
    &\geq 4 - \underbrace{\frac{19 - \sqrt{21}}{2}}_{\leq 7.21}\eps - 0.1\eps \\
    &> (2r)^2 \qquad \text{(by \Cref{eq:2r-sqrt})}
\end{align*}
Thus, \cref{lm:item_B1B2jj+1} of the lemma holds.

\noindent\textbf{Case 3:}  $i = j + 1$ and $j$ is even.  Let:
\begin{align*}
 d  \coloneqq &\frac{1}{2}(\bpoints_1[j+1] + \bpoints_2[j])\\
    &= \left( \frac{\sqrt{3}}{2} + \frac{\sqrt{7}}{4}, \frac{3}{4} \right) - \frac{1}{2} \left( \left( \frac{11}{\sqrt{7}} - \sqrt{3} \right) \eps, 0 \right) + \eps\cdot \left( \frac{\sqrt{3}}{4}, \frac{\sqrt{21}}{12} \right) \\
    &\quad + \left( 3 \cdot \frac{j}{2} \eps + X\left[ \frac{j}{2} \right] \cdot \eps^{1.5} \right) \cdot \left( \frac{\sqrt{3}}{4}, \frac{1}{2} - \frac{\sqrt{21}}{12} \right)\\
    &= \left( \frac{\sqrt{3}}{2} + \frac{\sqrt{7}}{4}, \frac{3}{4} \right) + \eps \left( \frac{3}{4}\sqrt{3}-\frac{11}{2\sqrt{7}}, \frac{\sqrt{21}}{12} \right)  \\
    &\quad + \left( 3 \cdot \frac{j}{2} \eps + X\left[ \frac{j}{2} \right] \cdot \eps^{1.5} \right) \cdot \left( \frac{\sqrt{3}}{4}, \frac{1}{2} - \frac{\sqrt{21}}{12} \right).
\end{align*}
In what follows, we will show (by straightforward yet somewhat tedious) calculations that:
\begin{enumerate}
    \item   $\bpoints_1[> j]$ is contained in the radius-$r$ disk with center $d$.
    \item $\bpoints_2[\le j]$ is contained in the radius-$r$ disk with center $d$.

\end{enumerate}

Then, we observe that:
\begin{equation}
    \|d - \frac{p+q}{2}\| = \left\| \eps \cdot \left( \frac{3}{4}\sqrt{3} - \frac{11}{2\sqrt{7}}, \frac{\sqrt{21}}{12} \right) + \left( 3 \cdot \frac{j}{2} \eps + X\left[\frac{j}{2}\right] \cdot \eps^{1.5} \right) \cdot \left( \frac{\sqrt{3}}{4}, \frac{1}{2} - \frac{\sqrt{21}}{12} \right) \right\|
\end{equation}

By applying the triangle inequality and pulling out factors:
\begin{equation}
   \begin{split}
       \|d - \frac{p+q}{2}\|  &\le \eps \cdot \underbrace{\left\| \left( \frac{3}{4}\sqrt{3} - \frac{11}{2\sqrt{7}}, \frac{\sqrt{21}}{12} \right) \right\|}_{\leq 0.9} + \left| 3 \cdot \frac{j}{2} \eps + X\left[\frac{j}{2}\right] \cdot \eps^{1.5} \right| \cdot \underbrace{\left\| \left( \frac{\sqrt{3}}{4}, \frac{1}{2} - \frac{\sqrt{21}}{12} \right) \right\|}_{\leq 0.45}\\
       &\le \left( \frac{3}{2} \cdot 0.45 \cdot |j| + 0.9 \right) \eps + \underbrace{0.45 n^2 \eps^{1.5}}_{\leq 0.1 \eps} \\
    &\le (|j| + 1) \cdot \eps
   \end{split}
\end{equation}
which proves \cref{lm:item_B1B2jj} of the lemma. 

Now we focus on showing that  $\bpoints_1[> j]$ is contained in the radius-$r$ disk with center $d$. For any $\ell \in I_n$, we have:
\begin{align*}
    \|\bpoints_1[\ell] - d\|^2 &= \left( p - \left( \frac{\sqrt{3}}{2} + \frac{\sqrt{7}}{4}, \frac{3}{4} \right) \right) \cdot \left( p - \left( \frac{\sqrt{3}}{2} + \frac{\sqrt{7}}{4}, \frac{3}{4} \right) \right) \\
    &\quad - 2 \eps \cdot \left( \frac{3}{4}\sqrt{3} - \frac{11}{2\sqrt{7}}, \frac{\sqrt{21}}{12} \right) \cdot \left( p - \left( \frac{\sqrt{3}}{2} + \frac{\sqrt{7}}{4}, \frac{3}{4} \right) \right) \\
    &\quad - 2 \left( 3 \cdot \frac{j}{2} \eps + X\left[ \frac{j}{2} \right] \eps^{1.5} \right) \cdot \left( \frac{\sqrt{3}}{4}, \frac{1}{2} - \frac{\sqrt{21}}{12} \right) \cdot \left( p - \left( \frac{\sqrt{3}}{2} + \frac{\sqrt{7}}{4}, \frac{3}{4} \right) \right) \\
    &\quad + 2 \left( 3 \lfloor \ell/2 \rfloor \eps + X[\lfloor \ell/2 \rfloor] \eps^{1.5} - (-1)^\ell \eps \right) \cdot \left( \frac{\sqrt{3}}{2}, \frac{1}{2} \right) \cdot \left( \underbrace{p - \left( \frac{\sqrt{3}}{2} + \frac{\sqrt{7}}{4}, \frac{3}{4} \right)}_{= (-\sqrt{7}/4,3/4)} \right) \\
    &\quad + \xi' \quad \text{(where  $|\xi'| \le 1000 n^4 \eps^2 < \eps^{1.7}$)}\\
    &= 1 - \frac{11 - \sqrt{21}}{4} \eps + \left( 3 \cdot \frac{j}{2} \eps + X\left[ \frac{j}{2} \right] \eps^{1.5} \right) \cdot \frac{\sqrt{21} - 3}{4} \\
    &\quad - \left( 3 \lfloor \ell/2 \rfloor \eps + X[\lfloor \ell/2 \rfloor] \eps^{1.5} - (-1)^\ell \eps \right) \cdot \frac{\sqrt{21} - 3}{4} + \xi'.
\end{align*}

For $\ell = i = j + 1$, since $j$ is even, $\lfloor \ell/2 \rfloor = j/2$. Thus:
\begin{equation}
    \|\bpoints_1[\ell] - d\|^2 = 1 - \frac{11 - \sqrt{21}}{4} \eps - \frac{\sqrt{21} - 3}{4} \eps + \xi' = 1 - 2\eps + \xi' \le r^2.
\end{equation}

For $\ell > i = j + 1$, we have $\lfloor \ell/2 \rfloor \ge j/2 + 1$:
\begin{align*}
     \|\bpoints_1[\ell] - d\|^2  &\le 1 - \frac{11 - \sqrt{21}}{4} \eps - 2 \cdot \frac{\sqrt{21} - 3}{4} \eps + 10 n^2 \eps^{1.5} \\
    &= 1 - 2\eps - \frac{\sqrt{21} - 3}{4} \eps + 10 n^2 \eps^{1.5} \le 1 - 2\eps \le r^2.
\end{align*}

Hence $\bpoints_1[> j]$ is contained in the radius-$r$ disk with center $d$.

Finally, we show that $\bpoints_2[\le j]$ is contained in the radius-$r$ disk with center $d$. First, we have:
\begin{align*}
    \|\bpoints_2[\ell] - d\|^2 &= \left(q - \left(\frac{\sqrt{3}}{2} + \frac{\sqrt{7}}{4}, \frac{3}{4}\right)\right) \cdot \left(q - \left(\frac{\sqrt{3}}{2} + \frac{\sqrt{7}}{4}, \frac{3}{4}\right)\right) \\
    &\quad - 2 \left(\left(\frac{11}{\sqrt{7}} - \sqrt{3}\right)\eps, 0\right) \cdot \left(q - \left(\frac{\sqrt{3}}{2} + \frac{\sqrt{7}}{4}, \frac{3}{4}\right)\right) \\
    &\quad + 2 \cdot \left(3\lfloor \ell/2 \rfloor \eps + X[\lfloor \ell/2 \rfloor] \cdot \eps^{1.5} - (-1)^\ell \eps\right) \cdot (0, -\gamma) \cdot \left(q - \left(\frac{\sqrt{3}}{2} + \frac{\sqrt{7}}{4}, \frac{3}{4}\right)\right) \\
    &\quad - 2 \cdot \eps \cdot \left(\frac{3}{4}\sqrt{3} - \frac{11}{2\sqrt{7}}, \frac{\sqrt{21}}{12}\right) \cdot \left(q - \left(\frac{\sqrt{3}}{2} + \frac{\sqrt{7}}{4}, \frac{3}{4}\right)\right) \\
    &\quad - 2 \left(3 \cdot \frac{j}{2} \eps + X\left[\frac{j}{2}\right] \eps^{1.5}\right) \cdot \left(\frac{\sqrt{3}}{4}, \frac{1}{2} - \frac{\sqrt{21}}{12}\right) \cdot \left(\underbrace{q - \left(\frac{\sqrt{3}}{2} + \frac{\sqrt{7}}{4}, \frac{3}{4}\right)}_{= (\frac{\sqrt{7}}{4}, -\frac{3}{4})}\right) \\
    &\quad + \xi'' \quad \text{ (with } |\xi''| \le 1000 n^4 \eps^2 < \eps^{1.7} \text{)}\\
     &= 1 - \frac{11 - \sqrt{21}}{4} \eps + \left(3\lfloor \ell/2 \rfloor \eps + X[\lfloor \ell/2 \rfloor] \eps^{1.5} - (-1)^\ell \eps\right) \cdot \frac{\sqrt{21} - 3}{4} \\
    &\quad - \left(3 \cdot \frac{j}{2} \eps + X\left[\frac{j}{2}\right] \eps^{1.5}\right) \frac{\sqrt{21} - 3}{4} + \xi''
\end{align*}

For $\ell = j$, since $j$ is even, $\lfloor \ell/2 \rfloor = j/2$ and $(-1)^\ell = 1$:
\begin{equation*}
    \|\bpoints_2[\ell] - d\|^2 = 1 - \frac{11 - \sqrt{21}}{4} \eps - \eps \cdot \frac{\sqrt{21} - 3}{4} + \xi'' = 1 - 2\eps + \xi'' \le r^2
\end{equation*}

For $\ell < j$, then $\lfloor \ell/2 \rfloor \le j/2 - 1$:
\begin{align*}
    \|\bpoints_2[\ell] - d\|^2  &\le 1 - \frac{11 - \sqrt{21}}{4} \eps + \frac{\sqrt{21} - 3}{4} \cdot (-3\eps + \eps + 2n^2 \eps^{1.5}) + \xi'' \\
    &\le 1 - 2\eps - \underbrace{\frac{\sqrt{21} - 3}{4}}_{\ge 0.39} \eps + \underbrace{2n^2 \eps^{1.5} + \xi''}_{\leq 0.1 \eps} \le 1 - 2\eps \le r^2
\end{align*}

Hence $\bpoints_2[\le j]$ is contained in the radius-$r$ disk with center $d$, as desired.
\end{proof}

Note that statements symmetric to \Cref{lm:B12-prop} hold for $(\bpoints_2, \bpoints_3)$, $(\apoints_1, \apoints_2)$, and $(\apoints_2, \apoints_3)$ (using reflections and \Cref{obs:P_n_and_reverse}). 

\subsection{Correctness: Proof of Theorem~\ref{thm:6-center}}

We assume that $n$ is sufficiently large. Observe that the reduction in \Cref{subsec:6cen-reduction} can be done in $\Tilde{O}(n)$ time.  It remains to verify that the correspondence between the YES case and the NO case of the Gap Convolution-3SUM problem and the solution to the 6-center problem.

\paragraph{YES case:} Suppose the Gap Convolution-3SUM instance $X[-n \dots n]$ is a YES instance: there exist indices $i, j, k \in \{-\lfloor \frac{n}{100} \rfloor, \dots, \lfloor \frac{n}{100} \rfloor\}$ such that $i+j+k=0$ and $X[i] + X[j] + X[k] = 0$.

We argue that the constructed point set can be covered by 6 radius-$r$ disks:

\begin{itemize}
    \item By \Cref{lm:covering-basic_new}, there exists a disk $D_1$ covering $\apoints_1[\leq 2i] \cup \bpoints_1[\leq  2j] \cup \cpoints[\leq  2k]$.
    \item Symmetrically, there is a disk $D_4$ covering $\apoints_3[> 2i] \cup \bpoints_b[> 2j] \cup \cpoints[> 2k]$.
\end{itemize}
By \Cref{lm:B12-prop},
\begin{itemize}
    \item  there is a disk $D_2$ covering $\bpoints_1[>2j] \cup \bpoints_2[\le 2j]$.
    \item there is a disk $D_3$ covering $\bpoints_2[>2j] \cup \bpoints_3[\le 2j]$.
    \item there is a disk $D_6$ covering $\apoints_1[>2i] \cup \apoints_2[\le 2i]$.
    \item there is a disk $D_5$ covering $\apoints_2[>2i] \cup \apoints_3[\le 2i]$.
\end{itemize}

Moreover, \Cref{lm:covering-basic_new} and \Cref{lm:B12-prop} ensure that the centers of $D_1, \dots, D_6$ are within distance $3\eps \cdot (\max\{|2i|, |2j|\} + 1) \le 0.1 n \eps$ from the respective centers $c_1, \dots, c_6$. Thus, by  \Cref{lm:6cen-anchor}, these disks cover all anchor points. Consequently, all points in the construction can be covered by 6 radius-$r$ disks.

\paragraph{NO case:} Suppose that all points can be covered by 6 radius-$r$ disks. We argue that $\exists i, j, k \in \{-n, \dots, n\}$ s.t. $i+j+k=0$ and $X[i] + X[j] + X[k] = 0$ and hence $X$ is not a NO instance of Gap Convolution-3SUM.

By  \Cref{lm:6cen-anchor}, the 6 disks have centers $d_1, \dots, d_6$ that are within distance $0.3n\eps$ of $c_1, \dots, c_6$. Thus:
\begin{itemize}
    \item  disk $D_1$ does not contain all points in $\bpoints_1$, since the furthest point in $\bpoints_1$ has distance $>  1 + 3n\eps$ from $c_1$, but $D_1$ has radius $r = 1 - \eps + \eps^{1.7}$ and its center has distance $0.3n\eps$ from $c_1$.
    \item  Similarly, $D_1$ does not contain all points in $\bpoints_1, \apoints_1$ or $\cpoints$, 
    \item  and $D_4$ does not contain all points in $\bpoints_3, \apoints_3$ or $\cpoints$.
\end{itemize}

Next, we argue that:

\begin{claim}
    $D_2$ does not contain all points in $\bpoints_1$ or $\bpoints_2$. (Similar statements hold for $D_3$ vs $\bpoints_2, \bpoints_3$, $D_5$ vs $\apoints_2, \apoints_3$, and $D_6$ vs $\apoints_1, \apoints_2$).
\end{claim}
\begin{proof} Recall that  $\bpoints_1[-2n] = p + (-3n\eps - \eps + X[-n]\cdot \eps^{1.5}) \cdot \left(\frac{\sqrt{3}}{2}, \frac{1}{2}\right)$, $c_2 = \frac{p+q}{2}$, $p = \left(\frac{\sqrt{3}}{2}, \frac{3}{2}\right)$ and $q = \left(\frac{\sqrt{3}}{2} + \frac{\sqrt{7}}{2}, 0\right)$. Note that $\|p-q\|=2$. Hence,
\begin{align*}
    \|\bpoints_1[-2n] - c_2\| &\ge (\bpoints_1[-2n] - c_2) \cdot \frac{p-q}{\|p-q\|} \\
    &= \frac{1}{4}(p-q) \cdot (p-q) + (-3n\eps - \eps + X[-n] \cdot \eps^{1.5}) \cdot \left(\frac{\sqrt{3}}{2}, \frac{1}{2}\right) \cdot \underbrace{\frac{p-q}{2}}_{= (-\frac{\sqrt{7}}{4},\frac{3}{4})} \\
    &= 1 + (3n\eps + \underbrace{\eps - X[-n]\cdot \eps^{1.5}}_{\geq 0}) \cdot \frac{\sqrt{21} - 3}{8} \\
    &\ge 1 + 3 \cdot \frac{\sqrt{21} - 3}{8} \cdot n\eps > 1 + 0.5 n\eps 
\end{align*}
But $D_{2}$ has radius $r \le 1$ and its center is within distance $0.4n\eps$ from $c_{2}$, so $D_{2}$ cannot contain $\bpoints_{1}[-2n]$. 

The argument for $\bpoints_{2}[2n+1]$ is similar. Recall that
\begin{equation}
    \bpoints_{2}[2n+1] = \hat{q} + (3n\eps + X[n]\eps^{1.5} + \eps) \cdot (0, -\gamma)
\end{equation}
where $\gamma = \frac{\sqrt{21}}{6} - \frac{1}{2}$. We calculate the distance from $c_{2}$:
\begin{align*}
    \|\bpoints_{2}[2n+1] - c_{2}\| &\ge (\bpoints_{2}[2n+1] - c_{2}) \cdot \frac{q-p}{\|q-p\|} \\
    &= \frac{1}{4}(q-p)\cdot(q-p) + (3n\eps + X[n]\eps^{1.5} + \eps) \cdot (0, -\gamma) \cdot \underbrace{\frac{p-q}{2}}_{= (\frac{\sqrt{7}}{4},-\frac{3}{4})} \\
    &\quad - ((\frac{11}{\sqrt{7}} - \sqrt{3})\eps, 0) \cdot \frac{q-p}{2} \\
    &= 1 + (3n\eps + X[n]\eps^{1.5} + \eps) \cdot \frac{\sqrt{21} - 3}{8} - \frac{11 - \sqrt{21}}{4} \cdot \eps \\
    &= 1 + n\eps \cdot \frac{3\sqrt{21} - 9}{8} + X[n]\eps^{1.5} \cdot \frac{\sqrt{21} - 3}{8} - \frac{25 - 3\sqrt{21}}{8} \cdot \eps \\
    &\ge 1 + 0.59n\eps - \underbrace{0.2n^{2}\eps^{1.5}}_{\leq 0.1\eps} - 1.5\eps \\
    &> 1 + 0.5n\eps \quad \text{ since we can assume } n \ge 100.
\end{align*}
So $D_{2}$ cannot contain $\bpoints_{2}[2n+1]$. 
\end{proof}

Now we resume our argument. Let:
\begin{itemize}
   \item $i_{1}, j_{1}, k$ be the maximal indices s.t. $D_{1}$ covers $\apoints_{1}[i_{1}], \bpoints_{1}[j_{1}], \cpoints[k]$.
    \item $j_{2}$ is the maximal index s.t. $D_{2}$ covers $\bpoints_{2}[j_{2}]$.
    \item $j_{3}$ is the maximal index s.t. $D_{3}$ covers $\bpoints_{3}[j_{3}]$.
    \item $i_{2}$ is the maximal index s.t. $D_{6}$ covers $\apoints_{2}[i_{2}]$.
    \item $i_{3}$ is the maximal index s.t. $D_{5}$ covers $\apoints_{3}[i_{3}]$.
\end{itemize}
Then:
\begin{itemize}
    \item $D_2$ covers $\bpoints_1[j_1+1], \bpoints_2[j_2]$. 
\item $D_3$ covers $\bpoints_2[j_2+1], \bpoints_3[j_3]$. 
\item $D_6$ covers $\apoints_1[i_1+1], \apoints_2[i_2]$. 
\item $D_5$ covers $\apoints_2[i_2+1], \apoints_3[i_3]$. 
\item $D_4$ covers $\apoints_3[i_3+1], \bpoints_3[j_3+1], \cpoints[k+1]$.
\end{itemize}

\Cref{lm:B12-prop} and  \Cref{lm:covering-basic_new} imply that this is only possible if:
\begin{equation*}
    \begin{split}
        j_1 &\ge j_2 \ge j_3\\
    i_1 &\ge i_2 \ge i_3\\
    i_1&+j_1+k \le 0 \\
    -i_3&-j_3-k \le 0
    \end{split}
\end{equation*}
which is only possible if all these inequalities are equalities. Furthermore, all of the indices $j_1=j_2=j_3$, $i_1=i_2=i_3$, and $k$ must be even, as otherwise some inequality is strict. Moreover, by \Cref{lm:covering-basic_new}, $X[\frac{i_1}{2}] + X[\frac{j_1}{2}] + X[\frac{k}{2}] \le 0$ (applied to $D_1$) and also $X[\frac{i_1}{2}] + X[\frac{j_1}{2}] + X[\frac{k}{2}]  \ge 0$ (applied to $D_4$), so $X[\frac{i_1}{2}] + X[\frac{j_1}{2}] + X[\frac{k}{2}]= 0$. Hence, $i_1, j_1, k$ is a solution, which means $X$ is not a NO instance of Gap Convolution 3SUM.

\paragraph{Acknowledgments.} This work was initiated during the “Fine-Grained \& Parameterized Computational Geometry” workshop at the Lorentz Center. We also thank Sándor Kisfaludi-Bak and Dániel Marx for useful discussions.

\bibliographystyle{alpha}
\bibliography{temp}

\end{document}